\documentclass[11pt]{amsart}
\usepackage{graphicx,tikz,amsmath,amsfonts,amssymb,amsthm, paralist,xcolor,slashed,accents,esint,mathrsfs,comment}%

\usepackage[utf8]{inputenc}%

\usepackage[a4paper,top=2.5cm,bottom=3cm,left=2cm,right=2cm]{geometry}%

\usepackage[title]{appendix}

\usepackage[hidelinks]{hyperref}

\newcommand{\RR}{\mathbb{R}}

\newcommand{\NN}{\mathbb{N}}
\newcommand{\TT}{\mathbb{T}}

\newcommand{\osc}{\mathrm{osc}}
\newcommand{\av}{\mathrm{av}}
\newcommand{\de}{\mathrm{d}}
\newcommand{\EdS}{\mathrm{EdS}}
\newcommand{\FL}{\mathrm{FL}}

\newcommand{\mc}{\mathcal}

\newcommand{\ol}{\overline}
\newcommand{\wh}{\widehat}

\newcommand{\lesa}{\lesssim}

\newcommand{\pa}{\partial}

\newcommand{\ra}{\rightarrow}
\newcommand{\ve}{\varepsilon}

\newcommand{\pav}{\partial_{\mathbf{v}}}

\newcommand{\idmet}{\bar{g}_0}
\newcommand{\idsff}{\bar{k}_0}
\newcommand{\idrho}{\bar{\rho}_0}

\numberwithin{equation}{section}
\newtheorem{theorem}{Theorem}[section]
\newtheorem{proposition}[theorem]{Proposition}

\newtheorem{lemma}[theorem]{Lemma}
\theoremstyle{definition}
\newtheorem{definition}[theorem]{Definition}
\theoremstyle{remark}
\newtheorem{remark}[theorem]{Remark}

\begin{document}

\title[Nonlinear stability of  Einstein--de Sitter Universes]{Nonlinear stability of  Einstein--de Sitter Universes}
\author[L.~Bernhardt]{Louie Bernhardt}
\address{School of Mathematics and Statistics, University of Melbourne, Parkville 3010, VIC, Australia. \newline louie.bernhardt@student.unimelb.edu.au}

\author[D.~Fajman]{David Fajman}
\address{Faculty of Physics, University of Vienna, Boltzmanngasse 5, 1090 Vienna, Austria. \newline David.Fajman@univie.ac.at}

\author[Z.~Wyatt]{Zoe Wyatt}
\address{DPMMS, Faculty of Mathematics, Wilberforce Road, Cambridge, CB3 0WB, United Kingdom. \newline zoe.wyatt@maths.cam.ac.uk}
\date{July 10, 2026}


\begin{abstract}
The Einstein--de Sitter universe is the prevailing model used in cosmology to describe the cold dark matter-dominated epoch of the universe. 
This model is a spatially homogeneous and isotropic spacetime undergoing decelerated expansion, and is linearly unstable under the Einstein--Euler equations with a pressureless fluid equation of state. We show that every initial data set for the Einstein--Euler equations on $\mathbb{T}^3$ with a near-flat metric and positive fluid energy density converges to a flat metric under the Einstein--Euler flow with a polytropic equation of state. This means the metric asymptotes to an Einstein--de Sitter spacetime. In particular, this settles the question of whether the Einstein--de Sitter model can be nonlinearly stable for an appropriate matter model. 
\end{abstract}

\maketitle

\section{Introduction}
The standard model of cosmology posits that the large-scale structure of the universe is well-described by a $(1+3)-$dimensional spatially homogeneous and isotropic FLRW spacetime 
$$\mathcal{M} = (0,\infty) \times \Sigma, \quad g = -dt^2 + a(t)^2 g_\Sigma \,. $$
This spacetime is filled with certain matter components and arises as a solution to the Einstein field equations. Following a brief inflationary period, there are believed to have been three consecutive epochs, each with distinguished behaviour depending on which matter component is assumed to dominate at that point in time. These comprise the radiation-dominated epoch, followed by the cold dark matter-dominated epoch, and finally the current epoch dominated by dark energy. The matter in each of the first two epochs is modelled by an idealised fluid satisfying the relativistic Euler equations, while for the third a positive cosmological constant $\Lambda>0$ is typically used. 

Within each epoch, a natural and essential question concerns the nonlinear stability of the spacetime used to model that time period. Broadly speaking, there are three features which drastically affect the mathematical analysis of such a stability question: 
\begin{enumerate}
    \item[(A)] the rate of volume growth $a(t)$ of the spatial hypersurfaces over time, which critically depends on the epoch in question;
    \item[(B)] the possibly dominant effect of the matter or $\Lambda$ on the spacetime geometry;
    \item[(C)] the confinement of gravitational waves propagating on the spatial domain $\Sigma$, which in cosmology is typically assumed to be compact. 
\end{enumerate}
These features influence the effects of dispersion, dissipation and nonlinear interactions that are key to the rigorous mathematical study of nonlinear stability under the Einstein--Euler equations. 
In normalised units these equations read
\begin{subequations} \label{intro:Einstein-Euler}
    \begin{align}
        R_{\mu\nu}[\bar{g}] - \frac{1}{2}\bar{g}_{\mu\nu}R[\bar{g}] + \Lambda \bar{g}_{\mu\nu} &=(\bar{\rho}+\bar{p})\bar{v}_\mu \bar{v}_\nu + \bar{p}\,\bar{g}_{\mu\nu}\,,
    \end{align}
    where  $(\mathcal M, \bar{g})$ is a Lorentzian manifold and $\Lambda \geq 0$ is a constant. 
The relativistic fluid is described by its energy density $\bar{\rho}$, pressure $\bar{p}$, 4-velocity $\bar{v}^\mu$ normalised by $\bar{g}(\bar{v},\bar{v})=-1$. The fluid satisfies the relativistic Euler equations
        \begin{align}\label{intro:Euler-rho}
        \bar{v}^\mu \bar{\nabla}_\mu \bar{\rho} + (\bar{\rho} + \bar{p})\bar{\nabla}_\mu \bar{v}^\mu &= 0\,,\\ \label{intro:Euler-mom}
        (\bar{\rho}+\bar{p})\bar{v}^\mu \bar{\nabla}_\mu \bar{v}^\nu + \big( \bar{v}^\mu \bar{v}^\nu+\bar{g}^{\mu\nu}\big) \bar{\nabla}_\mu \bar{p} &=0\,.
    \end{align}
\end{subequations}
The equations \eqref{intro:Euler-rho}--\eqref{intro:Euler-mom} need to be closed by a thermodynamic equation of state. Two canonical barotropic relations studied in astrophysics and cosmology are the \textit{linear} and \textit{polytropic} relations, respectively given by:
\begin{align}
    \label{intro-linear-eos}\bar{p}(\bar{\rho})& = K \bar{\rho}\,, \\
    \label{intro-poly-eos}\bar{p}(\bar{\rho}) &= C \bar{\rho}^{1+\frac{1}{n}}\,.
\end{align} 
Here, $K \geq 0$, $C, n>0$ are constants, with the parameter $n$ called the polytropic index. The polytropic constant $C$ scales the absolute pressure of the fluid at a given density, representing either a measure of specific entropy in thermal gases or a characteristic parameter governing the fluid's stiffness. One computes the fluid sound speed by evaluating $c_s^2 =\frac{\de \bar{p}}{\de \bar{\rho}}$ at fixed entropy, and so in an expanding cosmology where the energy density dilutes, we see that the sound speed for a polytropic fluid decreases over time while the linear fluid has a constant speed of sound given by $c_s = \sqrt{K}$. 

In this paper we focus on the cold dark matter-dominated epoch, which is modelled by an FLRW solution called the Einstein--de Sitter (EdS) universe. Established by Einstein and de Sitter in 1932 \cite{EindeS:EdS:32}, the EdS model reads\footnote{In this article, we take $\TT^3 = [-\pi,\pi]^3$ with the ends identified. Here the Euclidean metric $\delta_{ij}$ can be replaced with any 3-dimensional flat torus.}
\begin{equation}\label{sec:intro.eq:eds}
\mathcal{M}_{\EdS} =(0,\infty)\times\mathbb T^3, \quad  \bar{g}_{\EdS} =-dt^2+t^{\frac43}\sum_{i,j=1}^3 \delta_{ij}dx^i dx^j, \quad \bar{\rho}_{\EdS} = \rho_0 t^{-2}, \quad \Lambda =0. 
\end{equation}
This is a solution to the Einstein--Euler field equations, and the fluid is described by its energy density $\bar{\rho}$, 4-velocity $\bar{v}^\mu$ and, critically, an assumed pressureless equation of state $\bar{p}=0$ i.e. $K=0$ in \eqref{intro-linear-eos}. Such a matter model is called a \textit{dust fluid}.

Numerous heuristic and numerical works, outlined in Section \ref{sec:cosmo-jeans} below, suggest that the EdS model is unstable under the Einstein-dust field equations. From a nonlinear perspective, this raises a significant question regarding the validity of describing a major epoch in our universe by a dynamically unstable model. In the present article, we provide a resolution:

\begin{theorem}[Main Theorem]\label{main-theorem-rough}
    There is an open family of nonlinearly stable solutions to the Einstein--Euler equations on $\mathbb{T}^3$ with $\Lambda =0$ and polytropic equation of state \eqref{intro-poly-eos} with polytropic index $n>3$. These solutions are geodesically complete to the future and possess the same leading order asymptotics as the Einstein--de Sitter universe \eqref{sec:intro.eq:eds}. More precisely, for each solution $(\bar{g},\bar{\rho},\bar{v})$ in this family, there exists a coordinate system $(t,x^1,x^2,x^3)$, such that as $t \ra \infty$,
    \begin{subequations}\label{sec:intro.eq:asymptotics}
        \begin{gather}
        |\bar{g}_{00} - g_{00}^\infty| \ra 0\,, \qquad \sum_{i,j=1}^3|t^{-\frac{4}{3}}\bar{g}_{ij} - g_{ij}^\infty| \ra 0\,,\qquad |t^2\bar{\rho}-\rho^\infty|  \ra 0\,,\\
        \sum_{i=1}^3|t^{-\frac{2}{3}}\bar{g}_{0i}| \ra 0,\qquad \sum_{i=}^3|t^{\frac{2}{3}}\bar{v}^i|\ra 0\,,
    \end{gather}
    \end{subequations}
    where $(g_{00}^\infty,g_{ij}^{\infty},\rho^\infty)$ are constants such that $g_{00}^\infty < 0$, $ \rho^\infty>0$, and $g_{ij}^{\infty}$ are the components of a symmetric, positive definite $3\times 3$ matrix.
\end{theorem}

See Theorem~\ref{sec:stability.thm:main} for a full statement of the above result.
\begin{remark}[Homogenisation of solutions]
    The asymptotics \eqref{sec:intro.eq:asymptotics} imply that the stable solutions $(\bar{g},\bar{\rho},\bar{v})$ from Theorem~\ref{main-theorem-rough} \emph{homogenise} over time, meaning they converge to some spatially homogeneous spacetime (and fluid) as $t \ra \infty$. It turns out that the system \eqref{intro:Einstein-Euler} admits a family of spatially homogeneous solutions with a quiet fluid,\footnote{By quiet fluid, we mean that for a given spacetime, the fluid 4-velocity $\bar{v}$, which we take to be unit timelike, is normal to the spatially homogeneous hypersurfaces which foliate that spacetime, i.e $\bar{v}^\mu = \delta_0^\mu$.} all of which have the same long-time growth/decay rates of their metric and fluid components as the Einstein--de Sitter universe. By applying a particular coordinate transformation, one can show that every member of this family of spatially homogeneous solutions can essentially be reduced to a single FLRW solution \eqref{intro:Einstein-Euler}; see already Appendix~\ref{app:FLRW}. Thus, our result implies that this \emph{family} of spatially homogeneous spacetimes is nonlinearly stable as solutions to \eqref{intro:Einstein-Euler}, moreover every spacetime in a neighbourhood of this family converges as $t \ra \infty$ to a (potentially different) member of the family.
\end{remark}
\subsection{Einstein--Euler dynamics in cosmological settings}
\subsubsection{Dark energy dominated evolution}\label{sec:intro-dark-energy}
The main mechanism by which boundedness of perturbations is achieved in cosmological stability results is dilution, which is strong when spacetime expansion is fast. Historically, this has been most well understood in the regime of the dark energy-dominated epoch.   Roughly speaking, a cosmological constant $\Lambda>0$ generates
expansion of the form $a(t) \sim e^{Ht}$ where $H = \sqrt{\Lambda/3}$, for all cases of the  curvature of the metric on $\Sigma$. Cosmological models with $\ddot{a}>0$ are called \textit{accelerated}, and in these models the feature (A) listed above is the singular driving force behind global stability. This was first identified in \cite{BrauerRendallReula:1994}, followed by rigorous work of \cite{Rin:deSitterENSFstab08, RodSpe:irrEulerdeSitterstab:13} and many others \cite{Rin:powerlaw:09,Spe:EulerdeSitterstab:12,  LubKro:deSitterradstab11, HadSpe:deSitterduststab:15,Oli:EulerdeSitterstab:16,  Friedrich:2017,  LeFlochWei-Chaplygin,  Gong_2024, FouMarOli:tilted:24, FouMarOli:tilted2:25}. In all these works, the fast spatial expansion creates damping terms in the Einstein--matter equations which
dilutes nonlinear perturbations and ensures future global existence for small perturbations around the homogeneous background solutions. 

\subsubsection{Fluid evolution on fixed matter-dominated spacetimes}
The radiation and dust dominated epochs of our universe, however, do not enjoy the  fast spacetime expansion detailed in Section \ref{sec:intro-dark-energy} above. The general Friedmann–Lema\^itre spacetime used to model these two epochs reads
\begin{equation}\label{sec:intro.eq:FL}
  \mathcal{M}_{\FL} =(0,\infty)\times\mathbb T^3, \quad  \bar{g}_{\FL} =-dt^2+t^{\frac{4}{3(1+K)}}\sum_{i,j=1}^3 \delta_{ij}dx^i dx^j, \quad  \bar{\rho}_{\EdS}= \rho_0 t^{-2}, \quad \Lambda = 0\,.
\end{equation}
For $K \in [0,1]$, \eqref{sec:intro.eq:FL} is a solution to the Einstein--Euler equations \eqref{intro:Einstein-Euler} with linear equation of state \eqref{intro-linear-eos}. The solutions are called \textit{decelerated} since $\ddot{a}<0$ and also \textit{matter-dominated} since the type of matter field, manifested in the choice of $K$, drives the precise expansion rate of the spacetime. 
 Note that in the pressureless dust setting when $K = 0$ the solution  \eqref{sec:intro.eq:FL} reduces to the EdS solution \eqref{sec:intro.eq:eds}. When modelling the radiation epoch, standard cosmological models take $K=1/3$ which is called a \textit{radiation fluid}. In this case, the expansion rate in \eqref{sec:intro.eq:FL} becomes the very slow rate $a(t)=t^{\frac{1}{2}}$.  

In this article, matter-domination plays a remarkably dominant role. Unlike nonlinearly stable Einstein solutions such as Minkowski or subextremal Kerr, the EdS solution is \textit{not} a vacuum solution. 
This in particular implies that stabilisation of the matter field is necessary to allow for stability of the spacetime. Consequently, and as first initiated by \cite{Spe:relEulerstab:13} and continued in \cite{WeiJDE2018, Fajman2021, Fajetal:linEulerstab:24}, a natural question is to determine conditions, such as on the expansion rate $a(t)$ or equation of state parameters, that either cause a fluid to stabilise on a fixed cosmological geometry or develop shocks in finite-time. 

For fluids with a linear equation of state \eqref{intro-linear-eos} and non-vanishing pressure $K>0$ on a fixed FLRW geometry, recent works by the second two authors and Maliborski, Ofner and Oliynyk indicate that precisely the EdS expansion rate $a(t)=t^{2/3}$ is the sharp threshold for stabilisation \cite{Fajetal:decelEuler1:25, Fajetal:decelEuler2:25, Fajetal:decelEuler3:25}. On this threshold and below, fluids with non-vanishing pressure do not homogenise (in fact they form shocks), while for faster expansion they homogenise i.e. the spatial gradient of the expansion-normalised  energy density converges to zero. 
Indeed the mechanism for fluid stabilisation presented in \cite{Fajetal:decelEuler2:25} relies on the non-trivial pressure of the fluid to couple together the dynamics of the fluid energy density and fluid velocity. On the other hand the equations \eqref{intro:Euler-rho} and \eqref{intro:Euler-mom} decouple in the  pressureless case and, while the dust fluid enjoys regular behaviour on fixed FLRW geometries with slow expansion rates $a(t)  \sim  t^{\frac12 +\delta}$ \cite{Spe:relEulerstab:13, Ju2025}, the expansion normalised fluid density does \emph{not} homogenise. 
Based on these results it seems that the dust-dominated EdS solution \eqref{sec:intro.eq:eds} expands too slowly to allow the fluid to homogenise and hence cannot be stable as a solution to the Einstein--Euler system. 
 
\subsubsection{Einstein--de Sitter and the polytropic equation of state}
In light of the previous subsection, we seek a stable matter model which will allow for a background solution that asymptotes to the EdS model. Unlike the pure dust matter used in \eqref{sec:intro.eq:eds} however, we require a pressure term to facilitate homogenisation of the fluid itself, which would generate a stronger stabilisation effect than for a linear equation of state. 
The observation of compatibility with dust asymptotics and the necessity of positive pressure led to a recent work by the second author and Marshall \cite{FajMar:polytropic:26}, via numerical studies in Gowdy symmetry, suggesting that EdS is nonlinearly stable as a solution to the Einstein--Euler system with a polytropic equation of state. Indeed, our main result confirms the conjecture of \cite{FajMar:polytropic:26}.

\subsection{Mathematical challenges and outline of the proof}
Several mathematical aspects in our proof of Theorem \ref{main-theorem-rough} require significant novel ideas to overcome challenges from all three features (A), (B), and (C). 
The inherent decelerated expansion appearing in our problem leads to non-integrable error terms which require careful analysis. Indeed all other cosmological Einstein-matter stability results to date require  expansion at a faster rate than that present in our theorem. Furthermore, the confinement of the spatial $\TT^3$ topology rules out any potential dispersive effect as a source of decay. However the \underline{most challenging obstacle} comes from the dynamics of the fluid and the metric, which are both \emph{highly and linearly coupled}. On the one hand the fluid, both at leading order and in its higher-order perturbations, drives the dynamics and expansion rate of the metric. On the other hand, the metric also drives the behaviour of the fluid at all orders. 
This coupling influences almost every part of our analysis, in particular the energy estimates for the metric and the fluid, and is a new effect not present in related works on slowly expanding spacetimes.
Additionally, the polytropic fluid has a more intricate decay structure than what is present in fluids with linear equation of state \cite{Fajetal:decelEuler2:25}. We must capture this structure, carefully balancing decay of the density and velocity, in order to control the nonlinearities that arise in our analysis.

Structurally, the proof of Theorem~\ref{main-theorem-rough} is ordered as follows: $(i)$ a transformation of the Einstein--Euler system into a wave-transport system for conformally rescaled evolution variables by using a sourced wave gauge; $(ii)$ ODE estimates for the averaged values of the evolution variables; $(iii)$ energy estimates for derivatives of the metric perturbations; and $(iv)$ energy estimates for derivatives of the fluid variables.
 We give an overview of these aspects below.   For notation used in the following paragraphs, we advise the reader to consult Section~\ref{subsec:notation}.

\subsubsection{The Einstein--Euler system in wave gauge}\label{subsubsec:gauged-eqs}
We first transform the evolution variables $(\bar{g},\bar{\rho},\bar{v})$ into expansion-normalised quantities $(g,\rho,v)$ by a conformal transformation involving the factor:
$$\tau = \int_{\TT^3} \frac{\de t}{t^{2/3}} =3t^{1/3}\,.$$
Using this transformation, the EdS metric from \eqref{sec:intro.eq:eds} takes the conformally flat form\footnote{Strictly speaking, we begin with an EdS metric of the form $\bar{g}_{\EdS} = -\de t^2 + (3t)^{4/3}\delta_{ij}\de x^i \de x^j$ and change to $\tau = (3t)^{1/3}$, which yields the expression \eqref{sec:intro.eq:eds-conf}. We adopt this convention for the remainder of the article.}
\begin{equation}\label{sec:intro.eq:eds-conf}
    \bar{g}_{\EdS} = \tau^4\big(-\de \tau^2 + \delta_{ij}\de x^i \de x^j\big)\,.
\end{equation}
In view of this, we conformally transform our solution to the Einstein--Euler equations \eqref{intro:Einstein-Euler} with $\Lambda = 0$ and  solve for the rescaled metric $g = \tau^{-4}\bar{g}$, which can be viewed as a perturbation of the Minkowski metric. We also carry out appropriate rescalings for the fluid variables $(\rho,v) = (\tau^{6}\bar{\rho},\tau^{2}\bar{v})$, and treat the rescaled quantities $(g,\rho,v)$ as the fundamental evolution variables.

Next, we break the diffeomorphism invariance of the equations \eqref{intro:Einstein-Euler} by imposing a \emph{sourced wave gauge}. 
This transforms the Einstein equations into a system of nonlinear damped wave equations for the metric components (see Section~\ref{sec:gauged-eqs} and Definition~\ref{def:gauge-fixed-system}). Schematically, these damped wave equations for $u \in \{g^{\mu\nu}: \mu, \nu = 0,1,2,3\}$ take the form
\begin{equation}\label{sec:intro.eq:wave-schem}
    \wh{\square}_g u + \frac{2\eta}{\tau}g^{00}\pa_\tau u = F_w\,,
\end{equation}
where $\wh{\square}_g = g^{\alpha\beta}\pa_\alpha \pa_\beta$ is the reduced wave operator, $\eta\in (0,\infty)$ is the \emph{damping parameter}, and $F_w$ represents the various inhomogeneous source terms which depend on the evolution variables.
Our choice of gauge means that the damping term, involving $\eta$ on the LHS of \eqref{sec:intro.eq:wave-schem}, takes different values for the metric lapse $g^{00}$, shift $g^{0i}$, and spatial components $g^{ij}$. This approach of sourced wave gauges draws from work in the setting of  accelerated cosmologies, in particular work of Ringstr\"om \cite{Rin:deSitterENSFstab08} followed by \cite{Rin:powerlaw:09, Spe:EulerdeSitterstab:12, RodSpe:irrEulerdeSitterstab:13, HadSpe:deSitterduststab:15, Oli:EulerdeSitterstab:16, HinVas:KdSexpandingstab:24}. 

The study of decelerated cosmologies using sourced wave gauges was only recently initiated by the first author \cite{Ber:decelerated:25}, in the context of the Einstein--scalar field system with exponential potential.  This has been extended in recent joint work by the first and third authors, contained in the thesis \cite{Ber:thesis:26}, to prove stability of decelerated solutions to the Einstein--Euler-scalar field system. However, for this model the scalar field is the mechanism driving the expansion of the solutions and the fluid does \emph{not} influence the metric to leading order.
In the present work, an \underline{essential insight} is to isolate the leading-order effect of the fluid on the metric in a way that captures the cosmological expansion and drives decay in the system. We achieve this by introducing a novel  matter-dependent gauge source term: see Section \ref{intro:bded-sec} and \eqref{sec:intro.eq:decaying-average} below.

Finally, under the sourced wave gauge the Euler equations \eqref{intro:Euler-rho}--\eqref{intro:Euler-mom} take the form
\begin{subequations}\label{sec:intro.eq:euler-schem}
    \begin{align}
        &v^\mu \pa_\mu \rho + \chi_1(\rho)\pa_\mu v^\mu = F_\rho\,,\label{sec:intro.eq:euler-schem1}\\
        &v^\mu \pa_\mu v^i + \frac{\chi_2(\rho)}{\tau^{\frac{6}{n}}}\big(v^\mu v^i + g^{\mu i}\big)\pa_\mu \rho + \frac{2}{\tau}v^0 v^i = F_v\,. \label{sec:intro.eq:euler-schem2}
    \end{align}
\end{subequations}
Here $\chi_1,\chi_2$ are certain functions of $\rho$ that converge to some positive constant, and $F_\rho$, $F_v$ are again source terms. The quantity $\frac{2}{\tau}v^0 v^i$ in \eqref{sec:intro.eq:euler-schem2} acts as a damping term. 

\subsubsection{Boundedness of the spatial averages}\label{intro:bded-sec}
An important feature in the study of evolution equations with compact spatial topology (say $\TT^3$ to be concrete) is that the dynamics of the evolution variables, which we schematically denote by $u$, decouple into two pieces:
\begin{equation}
    u_{\av} = \fint_{\TT^3}u(x)\de^3 x, \qquad u_{\osc} = u - u_\av\,.
\end{equation}
Here we have a spatially-independent \emph{averaged} contribution $u_\av$ and a spatially-dependent \emph{oscillatory remainder}.
The spatial averages are functions of the time coordinate $\tau$ only, and one can derive ODEs for the averages by integrating the evolution equations over the torus. The remainder term is controlled by the higher-order derivatives of the evolution variables by the Poincar\'e inequality.  This was observed for the relativistic Euler equations in \cite{Fajetal:decelEuler2:25}, for the Einstein-scalar field system in \cite{Ber:decelerated:25}, and for the vacuum Einstein equations in \cite{Choquet-Bruhat:2001qwz, HuneauStingoWyatt2025}. 

In Section~\ref{sec:avg-ests}  we control  the spatial averages of the metric and matter using the evolution equations, while the lapse $g^{00}$ and shift $g^{0i}$ metric averages are improved using the sourced wave gauge constraint. This follows the approach developed in \cite{Ber:decelerated:25}. 
However unlike this latter work, the matter field now drives the expansion rate, and so many of the averaged evolution variables do not decay to their background values. In particular, the averaged lapse $g^{00}_\av$  and fluid density $\rho_\av$ do not converge to the background solution (which turn out to be $-1$ and $12$, after appropriate rescalings), but to nearby values. This would be a significant issue for proving boundedness of \emph{all} spatial averages, as the fluid density and lapse appear as critical source terms in the equations for the spatial metric components $g^{ij}$, $i,j=1,2,3$. Using only boundedness of these quantities would produce logarithmic growth in the spatial metric perturbations $g^{ij}-\delta^{ij}$, which would lead to a breakdown of the proof. 

We overcome this issue by including a term in the gauge source functions (see Definition \ref{def:gauge-choice}) that depends explicitly on the spatial averages of the lapse and fluid density. Specifically, this modification involves the quantity
\begin{equation}\label{sec:intro.eq:decaying-average}
    (\rho_\av - 12) + 12(g_\av^{00}+1)\,.
\end{equation}
Significantly, we show in Proposition~\ref{sec:avg-ests.prop:avg-ests} that the Euler continuity equation \eqref{sec:intro.eq:euler-schem1}---modified by the wave gauge condition---implies decay for the combination appearing in \eqref{sec:intro.eq:decaying-average}. Moreover upon further inspection, the gauge-fixed equations for the spatial metric components $g^{ij}$ contain $g_\av^{00}$ and $\rho_\av$ in the exact combination of \eqref{sec:intro.eq:decaying-average}, and so these source terms are no longer critical. This observation is precisely what allows us to prove boundedness of $g^{ij}_\av$, and hence of all averaged quantities.

\subsubsection{Energy estimates for the metric wave equations}
Next we analyse the higher-order behaviour of the evolution variables $(g,\rho,v)$, captured by derivatives of these quantities, using $L^2$-based energies. We find that these terms decay more rapidly than the spatial averages (some of which do not decay at all), which results in a homogenisation of the solutions. The basic energy functional adapted to the damped wave equation \eqref{sec:intro.eq:wave-schem} for the metric components reads
\begin{equation}
    E_w[u] = \frac{1}{2}\int_{\TT^3} \big(|g^{00}|(\pa_\tau u)^2 + g^{ij}\pa_i u \pa_j u\big) \,,
\end{equation}
and this energy is equivalent to the norms $E_w \sim \|\pa_\tau u\|_{L^2(\TT^3)}^2 + \|\pa_x u\|_{L^2(\TT^3)}^2$. Thanks to the damped structure in \eqref{sec:intro.eq:wave-schem}, the energy obeys an estimate like
\begin{equation*}
    \pa_\tau E_w[u] \leq -\frac{2\eta}{\tau}\int_{\TT^3} |g^{00}|(\pa_\tau u)^2 + C\|\pa_\tau u\|_{L^2}\|F_w\|_{L^2} + \text{h.o.t}\,.
\end{equation*}
The term involving $-2\eta$ has a good sign, but does not control the full energy.
We rectify this by adding a \emph{correction term} to the energy $E_w$ (cf. \cite{Choquet-Bruhat:2001qwz}), defining
\begin{equation}\label{intro:cor-energy}
    E_{w,\eta}[u] = E_w[u] + \frac{\eta}{\tau}\int_{\TT^3}|g^{00}|u\pa_\tau u\,.
\end{equation}
We arrive at an energy estimate of the form
\begin{equation}\label{intro:energy-est-Ew}
    \pa_\tau E_{w,\eta}[u] \leq -\frac{2\eta}{\tau}E_{w,\eta}[u] + C\|\pa_\tau u\|_{L^2}\|F_w\|_{L^2} + \text{h.o.t}\,.
\end{equation}
The term involving $-2\eta $ now does control the whole energy. 
We apply this energy in Section \ref{sec:energy-met} to derivatives of the metric components, which allows us to apply the Poincar\'e inequality in order to show that the correction term is indeed lower-order. In order to prove decay of the corrected energy though, we must control $F_w$. This in fact requires a subtle commutator argument, which we explain below in Section~\ref{subsubsec:coupling-overview} after first estimating the fluid variables.

\subsubsection{Sharp energy estimates for the fluid variables}\label{subsubsec:fluid-energy-overview}
Next, we carry out $L^2$-based energy estimates for the fluid variables  to establish decay of derivatives of $(\rho,v)$. The basic fluid energy $E_f$ (see Definition~\ref{def:base-fluid-energy}) behaves like
\begin{equation*}
    E_f \sim \int_{\TT^3} (\pa_x \rho)^2 + \tau^{\frac{6}{n}}\int_{\TT^3} (\pa_x v)^2.
\end{equation*}
The extra weight $\tau^{\frac{6}{n}}$ compensates for the $\tau^{-\frac{6}{n}}$ weight present in the Euler momentum equations (see \eqref{sec:intro.eq:euler-schem2}). These weights are a defining feature of the polytropic equation of state, and they necessitate a delicate analysis that goes considerably beyond what is required for the case of a linear equation of state \cite{Fajetal:decelEuler2:25}. One can show that $E_f$ satisfies an energy estimate like
\begin{equation*}
    \pa_\tau E_f \leq -\frac{2-\frac{3}{n}}{\tau}\int_{\TT^3} \tau^{\frac{6}{n}}(\pa_x v)^2 + C\|\pa_x \rho\|_{L^2}\|F_\rho\|_{L^2} + C\tau^{\frac{6}{n}}\|\pa_x v\|_{L^2}\|F_v\|_{L^2} + \text{h.o.t}\,.
\end{equation*}
Once again, there is a bulk term with a good sign in this estimate, but it does not control the entire energy. We introduce a suitable correction term (see \eqref{sec:energy-fluid.eq:correction-def}). The corrected fluid energy $E_{f,\eta;1}$ with damping parameter  $\eta = 1-\frac{3}{2n}$ obeys an estimate of the form
\begin{equation}\label{intro:fluid-mod-energyest}
    \pa_\tau E_{f,\eta;1} \leq -\frac{2-\frac{3}{n}}{\tau}E_{f,\eta;1} + C\|\pa_x \rho\|_{L^2}\|F_\rho\|_{L^2} + C\tau^{\frac{6}{n}}\|\pa_x v\|_{L^2}\|F_v\|_{L^2} + \text{h.o.t}\,.
\end{equation}
Provided the source terms decay sufficiently fast, we can then prove decay  like 
\begin{equation}\label{intro:rhov-decay}
    |\pa_x \rho| \lesa \tau^{-(1-\frac{3}{2n}-\delta)}\,, \quad |\pa_x v| \lesa \tau^{-(1+\frac{3}{2n}-\delta)}\,, \quad 0<\delta\ll 1\,.
\end{equation}
While these decay rates agree with the numerical rates found  in \cite{FajMar:polytropic:26}, an immediate issue may arise from the non-integrable rate for the fluid density $\pa_x \rho$. This makes possible the dangerous situation of error terms appearing in the energy estimates that are not integrable in time, and so they cannot be controlled by a Gr\"onwall-type inequality.

To prevent such a problematic situation arising, we first observe that many of the error terms involve products of both $\rho$ and $v$. Whenever this occurs, we can ``absorb'' the bad decay of the $\rho$ factor into the good decay of the $v$ factor (i.e. the additional weight in the velocity part of the energy). More significantly though, we prove in Proposition~\ref{sec:stability.prop:fluid-time-deriv} that the  fluid variables $\pav \rho$, $\pav v^i$, where $\pav := v^\mu \pa_\mu$  is the \textit{material derivative}, actually exhibit improved decay compared to \eqref{intro:rhov-decay}. In particular, we find\footnote{In fact $\pav v^i$ also decays faster than its spatial derivatives, but this is not as essential to the analysis.} 
\begin{equation}
    \label{intro:pav-rho-decay}
    |\pav \rho| \lesa \tau^{-(1+\frac{3}{2n})}
\end{equation} 
and this  \emph{is} integrable in time.
This allows us to estimate all remaining dangerous error terms in the energy estimates \eqref{intro:fluid-mod-energyest}. 

\subsubsection{Linear coupling of the metric and fluid equations}\label{subsubsec:coupling-overview}
We return now to the issue of the strong coupling between the metric and fluid. For the higher derivatives of $(g,\rho,v)$, this coupling manifests in the inhomogeneous source terms schematically represented in \eqref{sec:intro.eq:wave-schem}--\eqref{sec:intro.eq:euler-schem} by $F_{w}$, $F_\rho$, $F_v$.
These source terms produce two major and related issues. The first issue is that the slowest-decaying (linear) terms present in the source terms (see $\mathcal{L}_h, \mathcal{L}_\rho, \mathcal{L}_v$ in Proposition~\ref{sec:gauged-eqs.prop:gauged-eqs}) do not have a straightforward hierarchy that can be exploited. There are slowly-decaying terms that depend on $(\rho,v)$ in the metric equation source terms, and conversely there are slowly-decaying terms that depend on $g$ in the Euler equation source terms. Even worse, the source terms decay at too slow a rate to close the energy estimates. For example, if the rescaled fluid density decays as in \eqref{intro:rhov-decay}, then certain $\rho$-dependent source terms for the metric equations imply $|\pa_x g| \lesa \tau^{-(2-\frac{3}{2n}-\delta)}$. This in turn would imply decay of the fluid density like $|\pa_x \rho| \lesa \tau^{-(1-\frac{9}{2n}-\delta)}$, which is a loss in decay of $O(\tau^{\frac{3}{n}})$. 

As mentioned above, the material derivatives of the fluid variables \eqref{intro:pav-rho-decay} have better decay than the spatial derivatives \eqref{intro:rhov-decay}. To take advantage of this improved decay, we employ a vectorfield commutator argument for the metric components, commuting $\pav$ through the wave equations \eqref{sec:intro.eq:wave-schem} to obtain
\begin{equation}
    \wh{\square}_g (\pav u) + \frac{2\eta}{\tau}g^{00}\pa_\tau (\pav  u) = \pav F_w + \text{h.o.t}\,.
\end{equation}
Crucially, the quantities $\pav F_w$ now exhibit better decay than $F_w$ or $\pa_x F_w$. This gives improved decay for $\pav g$ and one can recover improved decay for all derivatives of the metric by an elliptic estimate (see Lemma~\ref{sec:energy-met.lem:elliptic}). An analogous approach has been used to recover top-order regularity for the Einstein-dust system in \cite{HadSpe:deSitterduststab:15,FajOfnWya:EEstab1:24}. However, our method here to derive improved \emph{decay} for material derivatives of the fluid variables is, to our knowledge, new in the Einstein--Euler setting.

Once we commute the wave equations for $g^{\mu\nu}$ with the material derivative, we are left with source terms in both metric equations and fluid equations that decay at the same critical rate as the damping terms (i.e. like the term involving $-2\eta$ in \eqref{intro:energy-est-Ew}). These additional terms, however, have ``bad signs'' and ``critical decay'' from the perspective of the energy estimates.

Our insight here is to use some remaining gauge freedom to pick gauge source functions that amplify the damping coefficients and thus can absorb these critically decaying error terms.  We fix gauge source functions with terms that produce damping terms in the wave equations for the lapse and shift, dependent on a fixed parameter $\kappa$; see Definition~\ref{def:gauge-choice}. By fixing $\kappa$ sufficiently large, we are able to make the damping terms in the lapse and shift equations as large as we like. Unfortunately, we do not have the freedom to choose arbitrarily large damping terms in the wave equations for the spatial metric components. Thus one final--and necessary--observation  that we make in Section~\ref{sec:boot-imp} is that these critical source terms depend only on the lapse, shift, and fluid variables. 

\begin{remark}[The polytropic index $n$ and constant $C$]\label{sec:intro.rmk:polytropic-index}
    In our main theorem, we  restrict the polytropic index $n$ in the equation of state \eqref{intro-poly-eos} to $n > 3$. This restriction is invoked in our construction of the corrected energy for the fluid in Section~\ref{sec:energy-fluid}; see specifically Definition~\ref{def:corrected-base-fluid-energy} and the computation that follows. In agreement with numerical evidence from \cite{FajMar:polytropic:26}, we do not expect homogenisation of the fluid to occur when $0 < n < 3$. Fluid homogenisation is likely a necessary mechanism for the nonlinear stability of slowly expanding spacetimes to hold, and so we conjecture that solutions to Einstein--Euler with polytropic index $0 < n < 3$ are \emph{unstable}.  The question of stability at the critical index $n = 3$ is at present unclear and invites further investigation. We also note that different restrictions on the polytropic index arise in the setting of exponentially expanding solutions to the Einstein--Euler system \eqref{intro:Einstein-Euler} with polytropic equation of state and, importantly, with $\Lambda > 0$ \cite{LiuWei2021}. We believe this is a natural consequence of the marked difference between decelerated and exponentially accelerated geometries.

    On the other hand, Theorem~\ref{main-theorem-rough} holds for all positive values of $C$, the polytropic constant from \eqref{eq:polytropic-eos}. The value of $C$ does not play a meaningful role in our analysis.
\end{remark}

\subsection{Context in general relativity and cosmology}
\subsubsection{Context within general relativity}
The list of known nonlinearly stable $(1+3)-$ dimensional spacetimes for the Einstein equations with vanishing cosmological constant, either in vacuum or with a physically realistic matter model,  consists of Minkowski spacetime \cite{ChrKlai:minkowskistab:93}, the Kerr family \cite{KlaiSzef:Kerr:23,dafermos2021nonlinearstabilityschwarzschildfamily,hintz2026nonlinearstabilitysubextremalkerr} and the Milne universe \cite{AndMon:Milnestab:11}. The result presented in this paper complements that list with the first nonlinear stability result for a physically realistic matter-dominated cosmological spacetime. Moreover, among the stable models with compact spatial topology, it is the one with the slowest expansion rate. 

Our methodological approach is sufficiently sharp to verify the numerical decay rates observed in \cite{FajMar:polytropic:26}. Additionally, our proof is sufficiently robust to include additional matter models beside the polytropic fluid, which are compatible with the asymptotic behaviour of our solutions, or to consider charged matter in the presence of Maxwell fields. This paves the way to establish stability of the EdS model in the presence of multiple co-existing matter fields.

The results in this paper also demonstrate the stability mechanism for a spacetime with close coupling to a homogenising matter field that drives geometric properties of the spacetime. It hence serves as a blueprint to potentially replace the dominating matter type, here polytropic fluids, by another model with the same features. For example, one could consider a relativistic Boltzmann model with collision kernel compatible with the EdS expansion rate (see recent work that poses the Boltzmann equation on fixed FLRW backgrounds \cite{Strain:2026fka}). In fact, for collisionless massless Vlasov matter, Taylor has shown that in the presence of dispersion, i.e. with non-compact spatial slices, a radiation dominated model is stable for spherically symmetric perturbations \cite{Tay:decelFLRWstab:24}.

Finally, we remark that our methods indicate that insufficient homogenisation of the spatial metric would prevent homogenisation in the fluid. This is not only relevant from a geometric perspective but apparently a necessary feature to enable stability in the first place. The present result confirms that within the class of $\Lambda=0$ stable cosmological spacetimes, the spatial metric asymptotically homogenises towards the future. This behaviour---which occurs as well for the Milne model, where the  expansion-normalised spatial metric converges to one with constant negative curvature---is different from the class of solutions undergoing accelerated expansion. In the case of accelerated expansion, timelike and null geodesics eventually become causally disconnected from each other. This causes perturbations to ``freeze in place'' over time, preventing homogenisation from occurring \cite{Fri:deSitterstab:86, Rin:deSitterENSFstab08}.

\subsubsection{Context within cosmology}\label{sec:cosmo-jeans}
The heuristic consensus in theoretical cosmology is that the EdS model is unstable. Standard linear perturbation theory applied to the Einstein--dust equations finds growth of certain fluid density perturbations for small initial data. This corresponds to the linear instability of the EdS model within the set of solutions to the Einstein-dust system as discovered by Sachs and Wolfe in their seminal paper \cite{SaWo67} by a Jeans-type analysis. A recent numerical work by Marshall also indicates there is an instability for nonlinear perturbations of the Friedmann–Lema\^itre spacetimes \eqref{sec:intro.eq:FL}, including the case of EdS \cite{Mar:instabilityFLRW:25}.

The growth of matter density perturbations on an EdS background is a standard explanation for structure formation during cosmological evolution \cite{Weinberg:2008zzc}. 
While the instability explains structure formation, it seemingly contradicts large-scale homogeneity of the universe, which creates a significant contradiction from the nonlinear perspective.
A direct consequence of our results in this paper is the justification of EdS as a model for the large scale structure of the matter-dominated epoch of the Universe. Indeed the polytropic fluid might be considered to asymptotically approach a pressureless dark matter component in this epoch which drives the large scale homogenisation. The absence of growing density perturbations in our solutions does not contradict that structure formation occurred during this period of cosmological evolution. Coupling additional matter models, which form structures possibly on smaller scales, is a natural generalisation of our result  to address the structure formation process. 

\subsection{Outline of the paper}
In Section~\ref{sec:prelim} we present several preliminary equations and notational conventions. Then, in Section~\ref{sec:gauged-eqs} we derive the Einstein--Euler system in a sourced wave gauge, and address the local well-posedness theory of the resulting gauge-fixed system. In Section~\ref{sec:stability} we state a full version of the main stability theorem, and present the bootstrap argument.
In Section~\ref{sec:avg-ests}, we derive the ODE system satisfied by the spatially averaged evolution variables. Then, in Section~\ref{sec:energy-met} we establish higher-order energy estimates for the metric perturbations and subsequently in Section~\ref{sec:energy-fluid} we do the same for higher derivatives of the fluid variables. Finally, in Section~\ref{sec:boot-imp} we improve the bootstrap assumptions on the higher-order derivatives of the evolution variables, completing the proof of the main theorem. 

\subsection{Acknowledgements}
The authors thank Elliot Marshall for helpful discussions. 
This research was funded in part by the Austrian Science Fund (FWF)
projects \textit{Matter-dominated cosmology}\newline (10.55776/PAT7614324) and \textit{Dynamics of matter in the
decelerated epoch} (10.55776/ PAT1953025).


\section{Preliminaries}\label{sec:prelim}
\subsection{Field equations and background solutions}
\begin{definition}[Einstein field equations with a polytropic fluid] \label{def:background-asymp-eds}
The Einstein field equations for a Lorenztian manifold $(\mathcal M, \bar{g})$ read
\begin{equation}\label{sec:intro.eq:einstein}
    R_{\mu\nu}[\bar{g}] - \frac{1}{2}R\bar{g}_{\mu\nu} = T_{\mu\nu}, \qquad T_{\mu\nu} = (\bar{\rho}+\bar{p})\bar{v}_\mu \bar{v}_\nu + \bar{p}\bar{g}_{\mu\nu}.
\end{equation}
Using  
$\bar{\nabla}$ the covariant derivative for $\bar{g}$, the relativistic Euler equations are determined by the conservation of energy-momentum $\bar{\nabla}^\mu T_{\mu\nu} = 0$.  
By projecting these equations onto both $\bar{v}$ and the $\bar{g}$-orthogonal complement of $\bar{v}$, it is well known that the Euler equations can be written as
\begin{subequations}
    \label{sec:intro.eq:euler}
        \begin{align}
        \bar{v}^\mu \bar{\nabla}_\mu \bar{\rho} + (\bar{\rho} + \bar{p})\bar{\nabla}_\mu \bar{v}^\mu &= 0,\\
        (\bar{\rho}+\bar{p})\bar{v}^\mu \bar{\nabla}_\mu \bar{v}^\nu + \bar{\Pi}^{\mu\nu}\bar{\nabla}_\mu \bar{p} &=0,
    \end{align}
\end{subequations}
    where $\bar{\Pi}^{\mu\nu}:=\bar{v}^\mu \bar{v}^\nu+\bar{g}^{\mu\nu}$ is the projection operator onto the $\bar{g}$-orthogonal complement of $\bar{v}$. 
    For $C, n >0$ constants, the polytropic equation of state reads
\begin{equation}\label{eq:polytropic-eos}
    \bar{p}(\bar{\rho}) = C \bar{\rho}^{1+\frac{1}{n}}.
\end{equation} 
\end{definition}

\begin{definition}[Homogeneous background solution]\label{def:aeds}
 The Einstein--Euler system \eqref{intro:Einstein-Euler} with polytropic equation of state \eqref{eq:polytropic-eos} admits a spatially homogeneous solution $(\mc{M},\bar{g}_n,\bar{\rho}_n,\bar{v}_n)$, which in coordinates $(t,x^1,x^2,x^3)$ reads
\begin{equation}
    \mc{M} = (0,\infty)\times \TT^3,\qquad \bar{g}_n = -\de t^2 + a^2(t)\sum_{i=1}^3 (\de x^i)^2,\qquad \bar{\rho}_n =b(t),\qquad \bar{v}_n^\mu = \delta_0^\mu,
\end{equation}
where the functions $(a(t),b(t))$ satisfy the ODEs
\begin{align}
    3\Big(\frac{\dot{a}(t)}{a(t)}\Big)^2 - b(t) = 0,\qquad
    \frac{\dot{b}(t)}{b(t)\big(1+b^{\frac{1}{n}}(t)\big)} + 3\frac{\dot{a}(t)}{a(t)} =0.
\end{align}
Asymptotic analysis given in Appendix \ref{app:FLRW} reveals that $a(t) \sim t^{2/3}$, $b(t) \sim t^{-2}$ as $t \ra \infty$.  
Alternatively, we can express this solution in coordinates $(\tau,x)$, where $\tau$ is the conformal time coordinate
\begin{equation*}
    \tau(t) = \int_{\TT^3} \frac{\de t}{a(t)} \sim t^{1/3}.
\end{equation*}
The spatially homogeneous solution then takes the form
\begin{equation}
    \mc{M} = (0,\infty)\times \TT^3,\qquad \bar{g}_n = a^2(\tau)\Big(-\de \tau^2 + \sum_{i=1}^3 (\de x^i)^2\Big),\qquad \bar{\rho}_n =b(\tau),\qquad \bar{v}_n^\mu = \frac{1}{a(\tau)}\delta_0^\mu,
\end{equation}
where $a(\tau) \sim \tau^2$, $b(\tau) \sim \tau^{-6}$ as $\tau \ra \infty$.
\end{definition}

\subsection{Notation}\label{subsec:notation}
We list some relevant notational conventions used in this article.
\subsubsection{Indices}We use Roman letters to denote spatial indices $a, b, i, j \cdots \in \{1, 2, 3 \}$, and Greek letters to denote spacetime indices $\alpha, \beta, \cdots \in \{0, 1, 2, 3\}$. The Einstein summation convention for repeated upper and lower indices is adopted. Unless stated otherwise, we raise and lower indices with respect to the conformal metric $g$ introduced in Section~\ref{sec:gauged-eqs}.

\subsubsection{Differential operators}
For a given coordinate chart $(x^0 = \tau,x^1,x^2,x^3)$, we denote the coordinate vector-fields
\begin{equation}
    \pa_\mu = \frac{\pa}{\pa x^\mu},
\end{equation}
and we often write $\pa_0 = \pa_\tau$ for the time derivative. We perform all computations with respect to the frame $\{\pa_\mu\}_{\mu=0,1,2,3}$.
Given a multi-index $I = (I_1,I_2,I_3) \in \NN^3$, the operator $\pa_x^I$ denotes the composition of spatial derivatives
\begin{equation}
    \pa_x^I = \pa_1^{I_1}\cdot \pa_2^{I_2}\cdot \pa_3^{I_3}.
\end{equation}
For a spacetime function $u$, we will use the schematic notation $\pa u$ to denote any/all spacetime derivatives of $u$, and $\pa_x u$ to denote any/all purely spatial derivatives of $u$.

In our analysis of metric perturbations, we use both the reduced wave operator $\wh{\square}_g$ and the damped wave operator $P_\eta$, for constant $\eta>0$, which are defined by
\begin{equation}
    \wh{\square}_g:= g^{\alpha\beta}\pa_\alpha\pa_\beta, \quad P_\eta := \wh{\square}_g + \frac{2\eta}{\tau}g^{00}\pa_\tau \,.
\end{equation}

\subsubsection{Norms and spatial averages}
Let $u$ be a function on $\TT^3$. We often write the integral of $u$ over $\TT^3$ as
\begin{equation}
    \int_{\TT^3} u := \int_{\TT^3}u(x)\,\de^3 x,
\end{equation}
omitting the volume form $\de^3 x = \de x^1 \wedge \de x^2 \wedge \de x^3$. 
For $u \in L^\infty(\TT^3)$, we denote the spatial average of $u$ over $\TT^3$ by
\begin{equation}
    u_{\av} := \fint_{\TT^3}u(x)\de^3 x. 
\end{equation}

The Sobolev space $W^{N,p}(\mathbb{T}^3)$ with $N \geq 0$ is defined as the completion of the function space $C^\infty(\mathbb{T}^3,\mathbb{R})$ with respect to the norm
\begin{equation}
\|u\|_{W^{N,p}} = \begin{cases} \begin{displaystyle}\biggl( \sum_{0\leq |\alpha|\leq N} \int_{\mathbb{T}^3} |\pa_x^{\alpha} u|^p\biggl)^{\frac{1}{p}}  \end{displaystyle} & \text{if $1\leq p < \infty $}\,, \\
 \begin{displaystyle} \max_{0\leq \ell \leq N}\sup_{x\in \mathbb{T}^3}|\pa_x^{\ell} u(x)|  \end{displaystyle} & \text{if $p=\infty$}\,.
\end{cases}
\end{equation}
We employ the standard notation $L^p=L^p(\mathbb{T}^3)=W^{0,p}(\mathbb{T}^3)$, and $H^N = H^N(\TT^3) = W^{N,2}(\TT^3)$, although note that we introduce notation for the Sobolev norm of the spatial components of the fluid velocity  
\[\|v\|_{H^N} := \sum_{i=1}^3 \|v^i\|_{H^N}\].

\subsubsection*{Constants and inequalities}
We write $A \lesssim B$ to indicate $A\leq C B$, whenever $C$ is a constant that depends only on only the regularity index $N$ and parameters of the Einstein--Euler system (like the polytropic index $n$), and not on evolution variables or coordinates. We will also write $A\sim B$ when both $A\lesa B$ and $B \lesa A$. Some constants which play a significant role in our analysis will be labelled uniquely, specifically in Proposition~\ref{sec:avg-ests.prop:avg-ests} and Theorem~\ref{sec:boot-imp.thm:higher-ests}.

\subsection{Function space inequalities}
We recall the following standard results that can be found in \cite{Heb:sob:00}. 
\begin{lemma}[Sobolev embedding]
Let $N \in \mathbb{N}$, $N \geq 2$. There exists a constant $C>0$ such that if
$u \in H^N(\mathbb{T}^3)$, then $u \in W^{N-2,\infty}(\mathbb{T}^3)$
and
\begin{equation}
\|u\|_{W^{N-2,\infty}}
\le C\|u\|_{H^N}.
\end{equation}
\end{lemma}

\begin{lemma}[Sobolev product estimate]
Let $N \in \mathbb{N}$, $N \geq 2$. Then there exists a constant $C>0$
such that if $u,v \in H^N(\mathbb{T}^3)$, the product
$u\cdot v \in H^N(\mathbb{T}^3)$, and
\begin{equation}
\|u\cdot v\|_{H^N}
\le C\|u\|_{H^N}\|v\|_{H^N}.
\end{equation}
\end{lemma}

\begin{lemma}[Poincar\'e inequality]
There exists a constant $C>0$ such that for all
$u\in H^1(\mathbb{T}^3)$, we have the following inequality:
\begin{equation}
\|u - u_\av\|_{L^2}
\le
C\|\pa_x u\|_{L^2}.
\end{equation}
\end{lemma}
\section{Gauge-fixed system}\label{sec:gauged-eqs}
In this section, we derive a gauge-fixed version of the Einstein--Euler system with polytropic equation of state. First, we reformulate the Einstein equations and Euler equations in terms of suitably rescaled evolution variables. Then, we impose a sourced wave gauge with carefully chosen source functions, and write down the resulting gauge-fixed system that is central to the global existence result proved in this paper. Following that, we discuss local-wellposedness of the gauge-fixed system, and prove that the gauge condition defining our choice of wave gauge is preserved by the gauge-fixed system, provided the Einstein constraint equations and the wave gauge condition are satisfied initially.

\subsection{Reformulation of the Einstein--Euler system}
We reformulate the Einstein--Euler system \eqref{intro:Einstein-Euler} in terms of a conformally rescaled metric $g = \Omega^2 \bar{g}$, and rescaled fluid variables. The purpose of these rescalings are that they normalise the asymptotics of the evolution variables, which simplify their analysis. Following \cite{Ber:decelerated:25}, we  additionally study the inverse metric components $g^{\mu\nu} = (g^{-1})^{\mu\nu}$. We begin first with a Lemma relating the Ricci tensor of two conformal metrics; see \cite{Ber:decelerated:25} for a proof.

\begin{lemma}[Ricci curvatures for conformal metrics]\label{sec:gauged-eqs.lem:ricci-conf}
    Let $\bar{g}$ be a Lorentzian metric, and introduce the conformal metric $g = \Omega^2\bar{g}$, where $\Omega\in C^\infty(\mc{M}\ra \RR)$ is a scalar function that is smooth and positive everywhere. Then the Ricci curvature of the metrics $g$ and $\bar{g}$ are related by the identity
    \begin{align}
        \label{sec:gaugedeqs.eq:ricciconformal}
        R^{\mu\nu}[g] =\>&g^{\mu\alpha}g^{\nu\beta}R_{\alpha\beta}[\bar{g}] + g^{\delta\lambda}\pa_\lambda g^{\mu\nu} \pa_\delta (\log \Omega) - 2g^{\lambda(\mu}\pa_\lambda g^{\nu)\delta}\pa_\delta (\log \Omega) \nonumber\\
        &- g^{\mu\nu} \Big(g^{\alpha\beta}\pa_\alpha \pa_\beta (\log\Omega) - 2g^{\alpha\beta}\pa_\alpha (\log \Omega) \pa_\beta (\log \Omega) - \Gamma^\lambda \pa_\lambda (\log \Omega)\Big)\nonumber\\
        &-2g^{\mu\alpha}g^{\nu\beta}\Big(\pa_\alpha\pa_\beta(\log \Omega) + \pa_\alpha (\log \Omega)\pa_\beta (\log \Omega)\Big),\nonumber
    \end{align}
    where $\Gamma^\lambda:= g^{\alpha\beta}\Gamma_{\alpha\beta}^\lambda(g)$.
\end{lemma}

Next, we state another lemma which gives a standard decomposition of the Ricci tensor into a scalar wave operator, non-hyperbolic
principal terms which become lower-order upon fixing the gauge, and semilinear terms. Again, we refer to \cite{Ber:decelerated:25} for a proof.
\begin{lemma}[Ricci curvature wave decomposition]\label{sec:gauged-eqs.lem:ricci-wave}
    The Ricci curvature of a Lorentzian metric $g$ satisfies the identity
    \begin{equation}
        2R^{\mu\nu}[g] = \wh{\square}_g g^{\mu\nu} + 2g^{\lambda(\mu}\pa_\lambda \Gamma^{\nu)}-\mc{F}^{\mu\nu}(g,\pa g),
    \end{equation}
    where the $\mc{F}^{\mu\nu}$ are the functions
    \begin{multline}\label{sec:gaugedeqs.eq:semilinearterms}
        \mc{F}^{\mu\nu}(g,\pa g) = g^{\alpha\beta}g^{\lambda\delta}\pa_\lambda g^{\mu\nu}\pa_{(\alpha} g_{\beta)\delta}-\frac{1}{2}g^{\alpha\beta}g^{\lambda\delta}\pa_\lambda g^{\mu\nu}\pa_\delta g_{\alpha\beta} + 2g^{\lambda(\mu}g^{\nu)\delta}\pa_\delta g^{\alpha\beta}\pa_\alpha g_{\beta\lambda}\\- \frac{1}{2}g^{\lambda(\mu}g^{\nu)\delta}\pa_\lambda g^{\alpha\beta}\pa_\delta g_{\alpha\beta}
        -g^{\alpha\beta}g^{\lambda(\mu|}\pa_\alpha g^{|\nu)\delta}\pa_\delta g_{\beta\lambda} - g^{\alpha\beta}g^{\lambda(\mu|}\pa_\alpha g^{|\nu)\delta}\pa_\beta g_{\lambda\delta}.
    \end{multline}
\end{lemma}

With these two lemmas, we can now derive a conformal version of the Einstein--Euler system \eqref{intro:Einstein-Euler}. First, we introduce the notation for the rescalings. 
\begin{definition}[Conformal rescalings $(g, \rho, p, v, \chi, \Pi)$]\label{def:conformal-rescalings}
Consider the conformally rescaled metric $g_{\mu\nu} = \Omega^2 \bar{g}_{\mu\nu}$  with a conformal factor of $\Omega = \tau^{-2}$. We define the rescaled inverse metric, fluid energy density, pressure and 4-velocity by:
    \begin{equation}\label{sec:gauged-eqs.eq:conf-rescaling}
        (g^{\mu\nu},\rho,v^i) := \big(\tau^{4}\bar{g}^{\mu\nu},\,\tau^{6}\bar{\rho},\,\tau^2 \bar{v}^i\big),
    \end{equation}
    and $p := \tau^{6+\frac{6}{n}}\bar{p} = C\rho^{1+\frac{1}{n}}$, $\chi:= \rho + \tau^{-\frac{6}{n}}p$, $\Pi^{\mu\nu} := v^\mu v^\nu + g^{\mu\nu}$. Note that under this rescaling, the vectorfield $v$ additionally satisfies the constraint
    \begin{equation}\label{sec:gauged-eqs.eq:fluid-constr}
        g_{\alpha\beta}v^\alpha v^\beta = -1\,.
    \end{equation}
\end{definition}

\begin{proposition}\label{sec:gauged-eqs.prop:ungauged-sys}
    Let $(\mc{M},\bar{g})$ be a $4-$dimensional Lorentzian manifold, let $\bar{\rho}$ be a scalar function on $(\mc{M},\bar{g})$, and let $\bar{v}$ be a unit timelike vectorfield on $(\mc{M},\bar{g})$. Fix a chart $(x^0 = \tau,x^1,x^2,x^3)$ on $\mc{M}$. Then $(\mc{M},\bar{g},\bar{\rho},\bar{v})$ solves the Einstein--Euler system \eqref{intro:Einstein-Euler} with polytropic equation of state \eqref{eq:polytropic-eos} if and only if the quantities $(\mc{M},g,\rho,v)$ solve 
    \begin{align}
        \label{sec:gauged-eqs.eq:ungauged-sys-1}
        \wh{\square}_g g^{\mu\nu} = &- 2g^{\lambda(\mu}\pa_\lambda \Gamma^{\nu)} -\frac{4}{\tau}g^{0\lambda}\pa_\lambda g^{\mu\nu}  + \frac{8}{\tau}g^{\lambda(\mu}\pa_\lambda g^{\nu)0} + \frac{12}{\tau^2}g^{00}g^{\mu\nu}--\frac{4}{\tau}\Gamma^0 g^{\mu\nu}--\frac{24}{\tau^2}g^{0\mu}g^{0\nu}\\
        &+\frac{1}{\tau^2}\rho\big(2v^\mu v^\nu + g^{\mu\nu}) + \frac{1}{\tau^{2+\frac{6}{n}}}p\big(2v^\mu v^\nu - g^{\mu\nu}\big) + \mc{F}^{\mu\nu},\qquad \mu,\nu=0,1,2,3,\nonumber
    \end{align}
    and 
    \begin{subequations}\label{sec:gauged-eqs.eq:ungauged-sys-2}
        \begin{align}
            v^\mu \pa_\mu \rho + \chi\pa_\mu v^\mu = &-\chi v^\mu \Gamma_{\mu\lambda}^\lambda- \frac{6}{\tau^{1+\frac{6}{n}}}pv^0,\\
            v^\mu \pa_\mu v^i + \frac{p'(\rho)}{\tau^{\frac{6}{n}}\chi}\Pi^{\mu i}\pa_\mu \rho  +\frac{2}{\tau}\Pi^{0i} =& -v^\mu v^\lambda\Gamma_{\mu \lambda}^i + \frac{6\rho p'(\rho)}{\tau^{1+\frac{6}{n}}\chi}\Pi^{0i}, \qquad i=1,2,3\,.
        \end{align}
    \end{subequations}

\end{proposition}

\begin{proof}
    First we prove that the Einstein equations \eqref{intro:Einstein-Euler} are equivalent to \eqref{sec:gauged-eqs.eq:ungauged-sys-1}.  It follows from Lemma~\ref{sec:gauged-eqs.lem:ricci-conf} that the Ricci curvatures of $g$, $\bar{g}$ are related by
    \begin{equation}\label{sec:gauged-eqs.prop:ungauged-sys-proof-1}
        R^{\mu\nu}[g] =\,g^{\mu\alpha}g^{\nu\beta}R_{\alpha\beta}[\bar{g}] - \frac{2}{\tau}g^{0\lambda}\pa_\lambda g^{\mu\nu} + \frac{4}{\tau}g^{\lambda(\mu}\pa_\lambda g^{\nu)0}
        + \frac{6}{\tau^2}g^{00}g^{\mu\nu}--\frac{2}{\tau}\Gamma^0g^{\mu\nu}-\frac{12}{\tau^2}g^{0\mu}g^{0\nu}.
    \end{equation}
    The Einstein equations are equivalent to
    \begin{equation*}
        2g^{\mu\alpha}g^{\nu\beta}R_{\alpha\beta}[\bar{g}] = \tau^{8}\Big(2(\bar{\rho} + \bar{p})\bar{v}^\mu\bar{v}^\nu + (\bar{\rho}-\bar{p})\bar{g}^{\mu\nu}\Big).
    \end{equation*}
    Substituting in the rescaled fluid variables $(\rho,p,v)$ and separating density and pressure terms, we have
    \begin{equation*}
        2g^{\mu\alpha}g^{\nu\beta}R_{\alpha\beta}[\bar{g}] = \frac{1}{\tau^2}\rho\big(2v^\mu v^\nu + g^{\mu\nu}) + \frac{1}{\tau^{2+\frac{6}{n}}}p\big(2v^\mu v^\nu - g^{\mu\nu}\big).
    \end{equation*}
    Inserting this into \eqref{sec:gauged-eqs.prop:ungauged-sys-proof-1} and decomposing the Ricci tensor $R^{\mu\nu}[g]$ as in Lemma~\ref{sec:gauged-eqs.lem:ricci-wave}, we obtain \eqref{sec:gauged-eqs.eq:ungauged-sys-1}.
   
    Next, we prove that the Euler equations \eqref{sec:intro.eq:euler} with polytropic equation of state \eqref{eq:polytropic-eos} are equivalent to \eqref{sec:gauged-eqs.eq:ungauged-sys-1}. First we change to the conformal metric $g = \tau^{-4}\bar{g}$, and conformal four-velocity $v^\mu = \tau^{2}\bar{v}^\mu$, which gives
    \begin{align*}
        v^\mu \pa_\mu \bar{\rho} + (\bar{\rho} + \bar{p})\pa_\mu v^\mu &=- \frac{6}{\tau}(\bar{\rho} + \bar{p})v^0 -(\bar{\rho}+\bar{p})v^\mu \Gamma_{\mu\lambda}^\lambda,\\
        v^\mu \pa_\mu v^\nu + \frac{\bar{p}'(\bar{\rho})}{\bar{\rho}+\bar{p}}\Pi^{\mu\nu}\pa_\mu \bar{\rho} +\frac{2}{\tau}\Pi^{0\nu} &= -v^\mu v^\lambda\Gamma_{\mu\lambda}^{\nu}\,.
    \end{align*}
    Then we transform to the rescaled fluid density and pressure $\rho = \tau^6\bar{\rho}$, $p = \tau^{6+\frac{6}{n}}\bar{p}$, from which we obtain the equations
    \begin{subequations}\label{sec:gauged-eqs.prop:ungauged-sys-proof-3}
        \begin{align}
            v^\mu \pa_\mu \rho + \chi\pa_\mu v^\mu = &-\chi v^\mu \Gamma_{\mu\lambda}^\lambda- \frac{6}{\tau^{1+\frac{6}{n}}}pv^0,\\
            v^\mu \pa_\mu v^\nu + \frac{p'(\rho)}{\tau^{\frac{6}{n}}\chi}\Pi^{\mu\nu}\pa_\mu \rho + \frac{2}{\tau}\Pi^{0\nu} =&-v^\mu v^\lambda\Gamma_{\mu\lambda}^\nu + \frac{6\rho p'(\rho)}{\tau^{1+\frac{6}{n}}\chi}\Pi^{0\nu},\qquad \mu=0,1,2,3,
        \end{align}
    \end{subequations}
    where $p'(\rho) := \frac{\de p}{\de \rho} = C(1+\frac{1}{n})\rho^{\frac{1}{n}}$. The system \eqref{sec:gauged-eqs.prop:ungauged-sys-proof-3} is overdetermined, in the sense that the $0$-component of the momentum equation is fully determined by the other equations in combination with the normalisation condition $g_{\alpha\beta} v^\alpha v^\beta = -1$. 
    In light of this, the Euler equations \eqref{sec:intro.eq:euler} are thus equivalent to \eqref{sec:gauged-eqs.eq:ungauged-sys-2}, completing the proof.
\end{proof}

\begin{remark}\label{sec:gauged-eqs.rmk:rescaling}
    The rescalings \eqref{sec:gauged-eqs.eq:conf-rescaling} normalise the (expected) asymptotics of the evolution variables, so that $(g^{\mu\nu},\rho,v) = O(1)$, roughly speaking. Performing the same rescalings to the the background FLRW solution $(\bar{g}_n,\bar{\rho}_n,\bar{v}_n)$ from Definition~\ref{def:aeds}, we find that
    \begin{equation}
       g_n^{\mu\nu} = \tau^4 \bar{g}_n^{\mu\nu} \ra m^{\mu\nu},\qquad \rho_n := \tau^6 \bar{\rho}_n \ra 12,\qquad v_n^i := \tau^2 \bar{v}^i = 0,
    \end{equation}
    In particular, the metric $\ol{g}_n^{\mu\nu}$ is asymptotically Einstein--de Sitter in nature, and so $g_n^{\mu\nu}$ is only a perturbation of Minkowski, which nonetheless converges to Minkowski as $\tau \ra \infty$.
\end{remark}

\subsection{Sourced wave gauge}
We are now ready to introduce the gauge-fixed Einstein--Euler system. 
\begin{definition}[Metric perturbation and gauge choice]\label{def:gauge-choice}
   We define the metric perturbation variables
\begin{equation}
    h^{\mu\nu} := g^{\mu\nu} - m^{\mu\nu}\,,
\end{equation}
and fix the following choice of wave gauge:
\begin{subequations}\label{sec:gauged-eqs.eq:gauge-choice}
    \begin{align}
        \Gamma^\mu(g) =&\, \mc{Y}^\mu ,\\
        \mc{Y}^\mu :=&\, \frac{\kappa}{\tau} (g^{0\mu}+\delta_0^\mu) + \delta_0^\mu \Big(\frac{3-\kappa}\tau h_{\av}^{00} + \frac{1}{4\tau}(\rho_\av-12)\Big). 
    \end{align}
\end{subequations}
where $\kappa > 0$ is some large parameter.  
\end{definition}

\begin{remark}
The above definition is a sourced wave gauge. For the standard choice of wave gauge, one would set $\Gamma^\mu=0$. Here, instead, we assume that the contracted Christoffel symbols satisfy an equation with lower-order gauge source functions. Our particular choice of gauge has two purposes. The first, which was pioneered in the influential work \cite{Cho:Einsteinlwp:52} by Choquet-Bruhat, is that any choice of wave gauge (or source wave gauge) transform the Einstein equations into a system of quasilinear wave equations. The second is that the gauge source function $\mc{Y}^\mu$ introduce important damping terms in the resulting wave equations. These damping terms are lower-order in terms of derivatives, but modify the leading-order decay rates of the evolution variables. In \eqref{sec:gauged-eqs.eq:gauge-choice}, we have included terms in the gauge source functions that modify the decay rates of the spatial averages of the evolution variables, as well as terms that modify the behaviour of higher-order spatial derivatives of the variables. Indeed, we observe that taking the average of the wave gauge condition gives
\begin{equation*}
    \Gamma_\av^0 = \frac{1}{4\tau}\big((\rho_\av-12) + 12h_\av^{00}\big), \qquad \Gamma_\av^i = \frac{\kappa}{\tau} h_\av^{0i},
\end{equation*}
while taking spatial derivatives of the condition yields
\begin{equation*}
    \pa_x^I \Gamma^\mu = \frac{\kappa}{\tau}\pa_x^I h^{0\mu}.
\end{equation*}
The quantity $(\rho_\av - 12) + 12h_\av^{00}$ will play an important role in our analysis of the spatial averages of the system. While the averaged density or lapse do not decay to their background values, we will be able to prove decay of $(\rho_\av - 12) + 12h_\av^{00}$; see already Proposition~\ref{sec:avg-ests.prop:avg-ests}.
\end{remark}

\begin{definition}[Gauge-fixed system]\label{def:gauge-fixed-system}
 We define the \emph{gauge-fixed system} to be the wave equations
\begin{subequations}\label{sec:gauged-eqs.eq:gauged-eqs-0}
    \begin{multline}\label{sec:gauged-eqs.eq:gauged-eqs-0-einstein}
        \wh{\square}_g g^{\mu\nu} = - 2g^{\lambda(\mu}\pa_\lambda \mc{Y}^{\nu)} -\frac{4}{\tau}g^{0\lambda}\pa_\lambda g^{\mu\nu} +\frac{8}{\tau}g^{\lambda(\mu}\pa_\lambda g^{\nu)0} +\frac{12}{\tau^2}g^{00}g^{\mu\nu}-\frac{4}{\tau}\mc{Y}^0 g^{\mu\nu}-\frac{24}{\tau^2}g^{0\mu}g^{0\nu}\\
        +\frac{1}{\tau^2}\rho\big(2v^\mu v^\nu + g^{\mu\nu}) + \frac{1}{\tau^{2+\frac{6}{n}}}p\big(2v^\mu v^\nu - g^{\mu\nu}\big) + \mc{F}^{\mu\nu},\qquad \mu,\nu=0,1,2,3,
    \end{multline}
    which are obtained by taking \eqref{sec:gauged-eqs.eq:ungauged-sys-1} and replacing all instances of $\Gamma^\mu$ with $\mc{Y}^\mu$, in conjunction with the Euler equations 
    \begin{align}
        \label{sec:gauged-eqs.eq:gauged-eqs-0-euler-1}
        v^\mu \pa_\mu \rho + \chi\pa_\mu v^\mu = &-\chi v^\mu \Gamma_{\mu\lambda}^\lambda- \frac{6}{\tau^{1+\frac{6}{n}}}pv^0,\\
        \label{sec:gauged-eqs.eq:gauged-eqs-0-euler-2}
        v^\mu \pa_\mu v^i + \frac{p'(\rho)}{\tau^{\frac{6}{n}}\chi}\Pi^{\mu i}\pa_\mu \rho + \frac{2}{\tau}\Pi^{0i} =&-v^\mu v^\lambda\Gamma_{\mu\lambda}^i + \frac{6\rho p'(\rho)}{\tau^{1+\frac{6}{n}}\chi}\Pi^{0i},\qquad i=1,2,3 \,.
    \end{align}
\end{subequations}
\end{definition}
 We discuss the local theory of the gauge-fixed system in more depth in Section \ref{subsec:gauge-initial-data}, but for now we say that the system \eqref{sec:gauged-eqs.eq:gauged-eqs-0} is a locally well-posed hyperbolic system of mixed order.

The next proposition is a decomposition of the gauged-fixed system that highlights the structure of the various equations.
\begin{proposition}[Structure of the gauge-fixed system]\label{sec:gauged-eqs.prop:gauged-eqs}
    The equations \eqref{sec:gauged-eqs.eq:gauged-eqs-0} for $(g,\rho,v)$ imply the following equations for the metric perturbations $h^{\mu\nu}$:
    \begin{subequations}\label{sec:gauged-eqs.eq:wave}
        \begin{align}
            \label{sec:gauged-eqs.eq:gauged-eqs-1}
            \Big(\wh{\square}_g +\frac{2(\kappa-2)}{\tau}g^{00}\pa_\tau\Big)h^{00} &= \mc{L}_h^{00}+\mc{G}_h^{00},\\
            \label{sec:gauged-eqs.eq:gauged-eqs-2}
            \Big(\wh{\square}_g + \frac{\kappa}{\tau}g^{00}\pa_\tau\Big) h^{0i} &= \mc{L}_h^{0i} + \mc{G}_h^{0i},\\
            \label{sec:gauged-eqs.eq:gauged-eqs-3}
            \Big(\wh{\square}_g + \frac{4}{\tau} g^{00} \pa_\tau\Big) h^{ij}&= \mc{L}_h^{ij} + \mc{G}_h^{ij},
        \end{align}
    \end{subequations}
    and the following equations for $(\rho,v)$:
    \begin{subequations}\label{sec:gauged-eqs.eq:fluid}
        \begin{align}
        \label{sec:gauged-eqs.eq:gauged-eqs-4}
        v^\mu \pa_\mu \rho + \chi\pa_\mu v^\mu &= \mc{L}_{(\rho,v)} +  \mc{G}_{(\rho,v)},\\
        \label{sec:gauged-eqs.eq:gauged-eqs-5}
        v^\mu \pa_\mu v^i + \frac{p'(\rho)}{\tau^{\frac{6}{n}}\chi}\Pi^{\mu i}\pa_\mu \rho +\frac{2}{\tau} v^0 v^i &= \mc{L}_{(\rho,v)}^i + \mc{G}_{(\rho,v)}^i,
        \end{align}
    \end{subequations}
     The quantities $\mc{L}_h^{\mu\nu}$, $\mc{L}_{(\rho,v)}$, $\mc{L}_{(\rho,v)}^i$ contain terms that are principal in terms of decay, but lower-order in terms of derivatives, and are given by
    \begin{subequations}
        \begin{align}
            \mc{L}_h^{00} :=&\,-\frac{1}{\tau^2}g^{00}(\rho-\rho_\av)-\frac{2(\kappa+6)}{\tau^2}g^{00}(h^{00}-h_\av^{00})\\
            &-\frac{2}{\tau}g^{00}\pa_\tau \big((3-\kappa)h_\av^{00}+ \frac{1}{4}(\rho_\av-12)\big) - \frac{3}{2\tau^2}g^{00}\big((\rho_\av-12)+12h_\av^{00}\big),\nonumber\\
            \mc{L}_h^{0i} :=& -\frac{\kappa-4}{\tau}g^{i\lambda}\pa_\lambda h^{00} + \frac{1}{\tau^2}\rho\big(2v^0v^i + h^{0i}\big),\\
            \mc{L}_h^{ij} :=& -\frac{2(\kappa - 4)}{\tau}g^{\lambda(i}\pa_\lambda h^{j)0}-\frac{4(\kappa - 3)}{\tau^2}g^{ij} (h^{00}-h_\av^{00}) + \frac{1}{\tau^2}(\rho-\rho_\av)g^{ij},\\
            \mc{L}_{(\rho,v)} :=& -\chi v^0 \Gamma_{0\lambda}^\lambda,\\
            \mc{L}_{(\rho,v)}^i :=& -\frac{2}{\tau} h^{0i} - (v^0)^2 \Gamma_{00}^i,
        \end{align}
    \end{subequations}
    while the quantities $\mc{G}_h^{\mu\nu}$, $\mc{G}_{(\rho,v)}$, $\mc{G}_{(\rho,v)}^i$ contain error terms that are lower-order both in terms of decay and derivatives, given by
    \begin{subequations}
        \begin{align}
            \mc{G}_h^{00} :=& -\frac{2(\kappa-2)}{\tau}h^{0a}\pa_a h^{00}+\frac{2}{\tau^2}\rho\big(g^{00}+(v^0)^2\big)+ \frac{1}{\tau^{2+\frac{6}{n}}}p\big(2(v^0)^2 - g^{00})+\mc{F}^{00},\\
            \mc{G}_h^{0i} :=&\, -\frac{\kappa}{\tau}h^{0a}\pa_a h^{0i} + \frac{\kappa-24}{\tau^2}g^{00}h^{0i}-\frac{3(\kappa -4)}{\tau^2}h^{00}h^{0i}+ \frac{1}{\tau^{2+\frac{6}{n}}}p\big(2v^0 v^i- h^{0i}\big)+ \mc{F}^{0i}\\
            &-\frac{1}{\tau} h^{0i}\pa_\tau\Big((3-\kappa) h_\av^{00}  + \frac{1}{4}(\rho_\av-12)\Big) - \frac{3}{\tau^2}h^{0i}\Big((3-\kappa) h_\av^{00}  + \frac{1}{4}(\rho_\av-12)\Big),\nonumber\\
            \mc{G}_h^{ij} :=&\,-\frac{4}{\tau} h^{0a} \pa_a h^{ij} + \frac{2(\kappa - 12)}{\tau^2}h^{0i}h^{0j} + \frac{2}{\tau^2}\chi v^i v^j - \frac{1}{\tau^{2+\frac{6}{n}}} g^{ij}p+ \mc{F}^{ij}, \\
            \mc{G}_{(\rho,v)} :=& -\chi v^a \Gamma_{a\lambda}^\lambda-\frac{6}{\tau^{1+\frac{6}{n}}}pv^0,\\
            \mc{G}_{(\rho,v)}^i :=& -2v^0 v^a \Gamma_{0a}^i - v^a v^b \Gamma_{ab}^i + \frac{6\rho p'(\rho)}{\tau^{1+\frac{6}{n}}\chi}\Pi^{0i}.
        \end{align}
    \end{subequations}
\end{proposition}
\begin{proof}
    These follow from straightforward computations. For example, from \eqref{sec:gauged-eqs.eq:ungauged-sys-1} we compute
    \begin{multline*}
        \wh{\square}_g h^{00} = - 2g^{0\lambda}\pa_\lambda \mc{Y}^{0} +\frac{4}{\tau}g^{0\lambda}\pa_\lambda h^{00} -\frac{12}{\tau^2}(g^{00})^2  -\frac{4}{\tau}\mc{Y}^0g^{00}\\+\frac{1}{\tau^2}\rho \big(2(v^0)^2 + g^{00}\big) +\frac{1}{\tau^{2+\frac{6}{n}}}p\big(2(v^0)^2 - g^{00}\big)+ \mc{F}^{00}.
    \end{multline*}
    The terms containing the gauge source functions $\mc{Y}^\mu$ combine with other terms on the right hand side of \eqref{sec:gauged-eqs.eq:gauged-eqs-1} like
    \begin{multline*}
        - 2g^{0\lambda}\pa_\lambda \mc{Y}^{0} +\frac{4}{\tau}g^{0\lambda}\pa_\lambda h^{00} = -\frac{2(\kappa-2)}{\tau}g^{00}\pa_\tau h^{00}-\frac{2(\kappa-2)}{\tau}h^{0a}\pa_a h^{00} + \frac{2\kappa}{\tau^2}g^{00}h^{00}\\
        -\frac{2}{\tau}g^{00}\pa_\tau\big((3-\kappa)h_\av^{00} + \frac{1}{4}(\rho_\av-12)\big)+ \frac{2}{\tau^2}g^{00}\big((3-\kappa)h_\av^{00} + \frac{1}{4}(\rho_\av-12)\big),
    \end{multline*}
    and
    \begin{equation*}
        -\frac{12}{\tau^2}(g^{00})^2 - \frac{4}{\tau}\mc{Y}^0 g^{00} = -\frac{12}{\tau^2} (g^{00})^2 - \frac{4\kappa}{\tau^2}g^{00} h^{00}-\frac{4}{\tau^2}g^{00}\big((3-\kappa)h_\av^{00}+\frac{1}{4}(\rho_\av-12)\big)\\
    \end{equation*}
    Noting that
    \begin{align*}
        -\frac{12}{\tau^2}(g^{00})^2 &= \frac{12}{\tau^2}g^{00} -\frac{12}{\tau^2}g^{00}h^{00},\\
        \frac{1}{\tau^2}\rho \big(2(v^0)^2 + g^{00}\big) &= \frac{2}{\tau^2}\rho\big((v^0)^2 + g^{00}\big) - \frac{1}{\tau^2}(\rho-12)g^{00}-\frac{12}{\tau^2}g^{00},
    \end{align*}
    we thus have
    \begin{multline}\label{sec:gauged-eqs.eq:gauged-eqs-proof-1}
        \Big(\wh{\square}_g + \frac{2(\kappa-2)}{\tau}g^{00}\pa_\tau \Big)h^{00}  = - \frac{1}{\tau^2}(\rho-12)g^{00}  - \frac{2(\kappa+6)}{\tau^2}g^{00}h^{00}\\
        -\frac{2}{\tau}g^{00}\pa_\tau\big((3-\kappa)h_\av^{00} + \frac{1}{4}(\rho_\av-12)\big) - \frac{2}{\tau^2}g^{00}\big((3-\kappa)h_\av^{00} + \frac{1}{4}(\rho_\av-12)\big)\\
        \underbrace{- \frac{2(\kappa-2)}{\tau}h^{0a}\pa_a h^{00} + \frac{2}{\tau^2}\rho\big((v^0)^2 + g^{00}\big) + \frac{1}{\tau^{2+\frac{6}{n}}}p\big(2(v^0)^2 - g^{00}\big)+ \mc{F}^{00}}_{=\mc{G}_h^{00}}.
    \end{multline}
    Rearranging several of the terms in \eqref{sec:gauged-eqs.eq:gauged-eqs-proof-1}, we get
    \begin{multline*}
        - \frac{2(\kappa+6)}{\tau^2}g^{00}h^{00}-\frac{2}{\tau^2}g^{00}\big((3-\kappa)h_\av^{00} + \frac{1}{4}(\rho_\av-12)\big) - \frac{1}{\tau^2}(\rho-12)g^{00}\\
        = -\frac{2(\kappa+6)}{\tau^2}g^{00}(h^{00}-h_\av^{00}) -\frac{1}{\tau^2} g^{00}(\rho-\rho_\av) - \frac{3}{2\tau^2}g^{00}\big((\rho_\av-12)+12h_\av^{00}\big).
    \end{multline*}
    Combining this with \eqref{sec:gauged-eqs.eq:gauged-eqs-proof-1} yields the equation \eqref{sec:gauged-eqs.eq:gauged-eqs-1}. The other equations follow in a similar manner; we leave their proofs to the reader.
\end{proof}

\subsection{Local theory for the Einstein--Euler system}\label{subsec:gauge-initial-data}
We briefly collect several results regarding the well-posedness theory for the Einstein--Euler system with polytropic equation of state. This includes propagation of the gauge condition \eqref{sec:gauged-eqs.eq:gauge-choice}, equivalence of initial data for the Einstein--Euler and gauge-fixed systems, and local well-posedness of the gauge-fixed system.

We introduce the quantities
\begin{equation}
    \mc{D}^\mu = \Gamma^\mu(g) - \mc{Y}^\mu,
\end{equation}
so that the wave gauge condition \eqref{sec:gauged-eqs.eq:gauge-choice} is equivalent to
\begin{equation}\label{sec:gauged-eqs.eq:gauge-d}
    \mc{D}^\mu \equiv 0,\qquad \mu=0,1,2,3.
\end{equation}
In the following lemma, we show that if $(g,\rho,v)$ is a solution to the gauge-fixed system \eqref{sec:gauged-eqs.eq:gauged-eqs-0}, then the functions $\mc{D}^\mu$ satisfy wave equations with respect to the physical metric $\bar{g}$.
\begin{lemma}\label{lem:evol-D}
Suppose we have a solution $(g, \rho, v)$ to the gauge-fixed Einstein--Euler system \eqref{sec:gauged-eqs.eq:gauged-eqs-0}. 
Then, there are smooth functions $A^{\mu\alpha\beta}$ and $B^{\mu\alpha}$ such that the scalar functions $\{\mc{D}^\mu:\mu =0,1,2,3\}$ satisfy
\[
\wh{\square}_g \mc{D}^\mu + A^{\mu\alpha\beta}\pa_\alpha \mc{D}_\beta + B^{\mu\alpha}\mc{D}_\alpha = 0 \,.
\]
\end{lemma}

\begin{proof}
We begin by writing out an alternative form of the gauge-fixed Einstein equations. We rearrange \eqref{sec:gauged-eqs.eq:gauged-eqs-0-einstein}, substituting $\mc{Y}^\mu = \Gamma^\mu - \mc{D}^\mu$ to get
\begin{multline*}
    \frac{1}{2}\wh{\square}_g g^{\mu\nu} + g^{\lambda(\mu}\pa_\lambda \Gamma^{\nu)}-\mc{F}^{\mu\nu} + \frac{2}{\tau}g^{0\lambda}\pa_\lambda g^{\mu\nu}-\frac{4}{\tau}g^{\lambda(\mu}\pa_\lambda g^{\nu)0}
    - \frac{6}{\tau^2}g^{00}g^{\mu\nu} -\frac{2}{\tau}\Gamma^0g^{\mu\nu} +\frac{12}{\tau^2}g^{0\mu}g^{0\nu}\\= g^{\alpha\mu}g^{\beta\nu}\Big(T_{\alpha\beta}-\frac{1}{2}g^{\lambda\delta}T_{\lambda\delta}g_{\alpha\beta}\Big) + g^{\lambda(\mu}\pa_\lambda \mc{D}^{\nu)}+ \frac{2}{\tau}\mc{D}^0 g^{\mu\nu}.
\end{multline*}
The left hand side is just $g^{\alpha\mu} g^{\beta\nu}R^{\alpha\beta}[\ol{g}]$ by Lemma~\ref{sec:gauged-eqs.lem:ricci-wave} and \eqref{sec:gauged-eqs.prop:ungauged-sys-proof-1}. Contracting twice with the lowered spacetime metric $\bar{g}$ and rearranging, it follows that the gauge-fixed Einstein equations are equivalent to the equations
\begin{equation}
    R_{\mu\nu}[\bar{g}] - \Big(T_{\mu\nu}-\frac{1}{2}\bar{g}_{\mu\nu}T\Big) = \frac{1}{\tau^2}\bar{g}_{\lambda(\mu}\pa_{\nu)} \mc{D}^{\lambda} + \frac{2}{\tau^3}\mc{D}^0 \bar{g}_{\mu\nu},
\end{equation}
defined with respect to the original spacetime variables (note that $T = \bar{g}^{\alpha\beta}T_{\alpha\beta}$). We take the trace with respect to $\bar{g}$:
\begin{equation}
    R[\bar{g}] + T =  \frac{1}{\tau^2}\pa_\lambda \mc{D}^\lambda + \frac{8}{\tau^3} \mc{D}^0,
\end{equation}
and so the gauge-fixed Einstein equations \eqref{sec:gauged-eqs.eq:gauged-eqs-0-einstein} imply\footnote{The functions $\mc{D}^\mu$ can also be written solely in terms of the coordinates $(\tau,x)$ and spacetime evolution variables $(\bar{g},\bar{\rho},\bar{p},\bar{v})$. We do not do this here, but point the interested reader to \protect{\cite[Remark 3.9]{Ber:decelerated:25}}.}
\begin{equation}\label{eq:reducedEin-D}
    R_{\mu\nu}[\bar{g}] - \frac{1}{2}\bar{g}_{\mu\nu}R[\bar{g}] - T_{\mu\nu} =  \frac{1}{\tau^2}\bar{g}_{\lambda(\mu}\pa_{\nu)} \mc{D}^{\lambda} -\frac{1}{2\tau^2}\bar{g}_{\mu\nu}\pa_\lambda \mc{D}^\lambda  - \frac{2}{\tau^3}\mc{D}^0 \bar{g}_{\mu\nu}\,.
\end{equation}
We take the divergence of both sides of this equation with respect to the covariant derivative $\bar{\nabla}$. The left hand side vanishes by a combination of the Bianchi identity and the Euler equations \eqref{sec:gauged-eqs.eq:gauged-eqs-0-euler-1}--\eqref{sec:gauged-eqs.eq:gauged-eqs-0-euler-2}, leaving us with the equation
\begin{equation}
    0 = \bar{\nabla}^\mu \Big(\frac{1}{\tau^2}\bar{g}_{\lambda(\mu}\pa_{\nu)} \mc{D}^{\lambda} -\frac{1}{\tau^2}\bar{g}_{\mu\nu}\pa_\lambda \mc{D}^\lambda  - \frac{2}{\tau^3}\mc{D}^0 \bar{g}_{\mu\nu}\Big).
\end{equation}
Treating $\bar{g}_{\lambda\mu} \mc{D}^{\lambda}$ as a one-form indexed by $\mu$, since $D^\lambda$ are now scalar functions, the result follows upon distributing the covariant derivative. 
\end{proof}

Initial data for the Einstein--Euler system is typically specified in a geometric manner, consisting of a $3$-dimensional Riemannian manifold $(\Sigma,\bar{g}_0)$, along with a covariant two-tensor $\bar{k}_0$, a function $\bar{\rho}_0$ and a vectorfield $\bar{v}_0$, all on $\Sigma$. In order to launch a solution to the Einstein--Euler system, the quadruple $(\bar{g}_0,\bar{k}_0,\bar{\rho}_0,\bar{v}_0)$ must also satisfy the \emph{Gauss-Codazzi constraints}
\begin{subequations}\label{sec:gauged-eqs.eq:constraints}
    \begin{align}
        R[\bar{g}_0] - (\text{Tr}_{\idmet} \bar{k}_0)^2 + |\bar{k}_0|^2 &= 2 T_{00}|_{\Sigma}, \\
        \nabla[\bar{g}_0]^i (\bar{k}_0)_{ij}- \nabla_j[\bar{g}_0] \text{Tr}_{\idmet} \idsff &= T_{0j}|_{\Sigma}\,.
    \end{align}
\end{subequations}

It is not a priori obvious that data for the gauge-fixed system of Definition~\ref{def:gauge-fixed-system}, given by
\begin{equation*}
    (g^{\mu\nu},\pa_\tau g^{\mu\nu},\rho,v^i)\big|_{\Sigma},\qquad \mu,\nu=0,1,2,3,\quad i=1,2,3,
\end{equation*}
can be specified in a manner that simultaneously guarantees that the induced geometric data $(\Sigma,\bar{g}_0,\bar{k}_0,\bar{\rho}_0,\bar{v}_0)$ satisfies the constraints \eqref{sec:gauged-eqs.eq:constraints}, and that the wave coordinate condition initially satisfies $\mc{D}^\mu|_{\Sigma} = \pa_\tau\mc{D}^\mu|_{\Sigma} = 0$. In view of Lemma~\ref{lem:evol-D}, this would guarantee that the solution to the gauge-fixed system launched by $(g^{\mu\nu},\pa_\tau g^{\mu\nu},\rho,v^i)\big|_{\Sigma}$ can be transformed back into a solution of the Einstein--Euler system. We now prove this assertion.
\begin{lemma}\label{sec:gauged-eqs.lem:initial-data}
The initial data $(g^{\mu\nu},\pa_\tau g^{\mu\nu}, \rho, v^i)|_\Sigma$ for the gauge-fixed equations \eqref{sec:gauged-eqs.eq:wave}--\eqref{sec:gauged-eqs.eq:fluid} can be chosen so that the emanating solution satisfies the Einstein--Euler system \eqref{intro:Einstein-Euler}.
\end{lemma}

\begin{proof}
Recall from Theorem \ref{main-theorem-rough}, our initial data set $((\bar{g}_0)_{ij}, (\bar{k}_0)_{ij}, \bar{\rho}_0, \bar{v}_0^i)$ specified on a spacelike hypersurface $\Sigma=\{\tau=\tau_0\}\times \mathbb{T}^3$ satisfies the constraints \eqref{sec:gauged-eqs.eq:constraints}. Clearly we should set $g^{ij}|_{\Sigma} = \tau_0^4 \bar{g}^{ij}$, $\rho|_{\Sigma} = \bar{\rho}_0$ and, after inverting the conformal transformation, set $v^i|_{\Sigma} = \tau_0^2 \bar{v}_0^i$. 
We pick the remaining data $g^{0\mu}|_{\Sigma}$ and $\partial_\tau  g^{0\mu}|_{\Sigma}$ in a way that guarantees that $\mc{D}^\mu$ and $\pa_\tau \mc{D}^\mu$ vanish on $\Sigma$, and so in combination with Lemma \ref{lem:evol-D} and classical properties of wave equations, we may conclude that $\mc{D}^\mu \equiv 0$. 

First, we choose
\[ g^{0i}|_\Sigma= 0, \quad g^{00}|_\Sigma = -1\,.\]
Due to this choice, the future directed unit normal to the hypersurface is $\pa_\tau$ and we also have $h^{0\mu}|_\Sigma = 0$.
Consequently, we prescribe $\partial_\tau g^{ij}$ by undoing the conformal transformation inside the identity $\partial_\tau \bar{g}^{ij} = 2\bar{k}^{ij}$ to find
\[ \pa_\tau g^{ij}|_\Sigma =\tau_0^4\Big(2\bar{k}_0^{ij}+\frac{4}{\tau_0^5}\bar{g}_0^{ij}\Big)\,.\]
Then we prescribe $\partial_\tau g^{00}$ and $\pa_\tau g^{0i}$ in a manner that $\mathcal{D}^\mu|_{\Sigma} = 0$ is satisifed. Expanding out $\Gamma^0(g)=\mc{Y}^0$ produces
\begin{align}
\mc{Y}^0 = -\pa_\tau g^{00} -\pa_a g^{0a} +\frac{1}{2} g_{\alpha\beta} g^{00}\pa_\tau g^{\alpha\beta}
+\frac{1}{2} g_{\alpha\beta} g^{0a}\pa_a g^{\alpha\beta}\,.
\end{align}
Several terms vanish since we have picked the shift to initially vanish, and so rearranging this yields
\[
\frac12 \pa_\tau g^{00}|_\Sigma  = - \frac{1}{4\tau_0}((\idrho)_\av - 12) - \frac12 \tau_0^{-4}(\idmet)_{ab} (\pa_\tau g^{ab}|_\Sigma) \,.
\]
Similarly 
\[
\mc{Y}^i= -\pa_\tau g^{0i}
-\pa_a g^{ia} +\frac{1}{2} g_{\alpha\beta} g^{0i}\pa_\tau g^{\alpha\beta} +\frac{1}{2} g_{\alpha\beta} g^{ia}\pa_a g^{\alpha\beta}\,,
\]
and so we pick
\[
\pa_\tau g^{0i}|_\Sigma = -\tau_0^4 \pa_j (\idmet)^{ij}+\frac{1}{2} \tau_0^4 (\idmet)_{ab} (\idmet)^{ij}\pa_j (\idmet)^{ab}\,.
\]
All together this ensures $\mathcal{D}^\mu = 0$ on the initial hypersurface. Then, evaluating the $0i-$components of \eqref{eq:reducedEin-D} and using the constraint equations \eqref{sec:gauged-eqs.eq:constraints} yields
\[\frac12 \bar{g}_{ia} \pa_\tau \mc{D}^a = 0\,,\]
while evaluating the $00-$component yields
\[ 0=  -\frac1{2} \pa_\tau \mc{D}^0\,.\]
In combination these imply that $\pa_\tau \mathcal{D}^\mu|_{\Sigma} = 0$  for $\mu = 0,1,2,3$.
\end{proof}

Finally, we state the local existence and uniqueness theorem satisfied by the gauge-fixed system of Definition~\ref{def:gauge-fixed-system}. We do not prove this result, but remark that it follows from a standard argument involving the contraction mapping theorem and energy estimates of the variety proved already in our global existence theorem. We remark that the inclusion of spatial averages in the gauge-fixed system does not affect such an argument, as the spatial averages can be viewed as functions that are smooth (in fact constant) along the level sets of $\tau$.
\begin{proposition}[Local well-posedness of the gauge-fixed system]\label{prop:local}
    Fix a polytropic index $n > 0$, an integer $N \geq 3$, and an initial time $\tau_0 > 0$. Letting $\Sigma_\tau = \{\tau\}\times \TT^3$, consider initial data
    \begin{gather*}
        g^{\mu\nu}|_{\Sigma_{\tau_0}} = g_0^{\mu\nu} \in H^{N+1}(\TT^3)\,,\qquad \pa_\tau g^{\mu\nu}|_{\Sigma_{\tau_0}} = g_1^{\mu\nu} \in H^{N}(\TT^3)\,, \qquad \mu,\nu=0,1,2,3\,, \\
        \rho|_{\Sigma_{\tau_0}} = \rho_0 \in H^N(\TT^3)\,, \qquad v^i|_{\Sigma_{\tau_0}} = v_0^i \in H^N(\TT^3)\,, \qquad i=1,2,3
    \end{gather*}
    for the gauge-fixed Einstein--Euler system \eqref{sec:gauged-eqs.eq:wave}--\eqref{sec:gauged-eqs.eq:fluid} (not necessarily satisfying the constraints \eqref{sec:gauged-eqs.eq:gauge-choice}, \eqref{sec:gauged-eqs.eq:constraints}). Suppose $g_0^{00} < 0$, $\inf_{\Sigma}|\rho_0|>0$, and the functions $g_{0}^{ij}$ constitute the components of (the inverse of) a Riemannian metric on $\Sigma$. Then there exists a unique classical solution $(g,\rho,v)$ to the gauge-fixed Einstein--Euler system \eqref{sec:gauged-eqs.eq:gauged-eqs-0} on $[\tau_0,T)\times \TT^3$, $T>\tau_0$ launched by
    $(g_0^{\mu\nu},g_1^{\mu\nu},\rho_0,v_0^i)$. The solution has regularity properties
    \begin{gather*}
        g^{\mu\nu} \in C^0\big([\tau_0,T), H^{N+1}(\TT^3)\big) \cap C^1\big([\tau_0,T), H^{N}(\TT^3)\big)\,, \qquad \mu,\nu=0,1,2,3,\\
        \rho \in C^0\big([\tau_0,T), H^{N}(\TT^3)\big)\,,\qquad  v^i \in C^0\big([\tau_0,T), H^{N}(\TT^3)\big)\,,\qquad i=1,2,3\,,
    \end{gather*}
    and the maximal time of existence $T$ and norms of the solution depend continuously on the initial data.

    If, in addition $(g_0^{\mu\nu},g_1^{\mu\nu},\rho_0,v_0^i)$ satisfy the constraints 
    \begin{subequations}\label{sec:gauged-eqs.eq:constraints2}
        \begin{align}
            R[\bar{g}_0] - (\text{Tr}_{\idmet} \bar{k}_0)^2 + |\bar{k}_0|^2 &= 2 T_{00}|_{\Sigma}\,, \\
        \nabla[\bar{g}_0]^i (\bar{k}_0)_{ij}- \nabla_j[\bar{g}_0] \text{Tr}_{\idmet} \idsff &= T_{0j}|_{\Sigma}\,,\\
        \mc{D}^\mu|_{\Sigma} &= 0\,, \qquad \mu=0,1,2,3\,,
        \end{align}
    \end{subequations}
    then the triple $(\bar{g},\bar{\rho},\bar{v})$ defined by \eqref{sec:gauged-eqs.eq:conf-rescaling} constitutes a solution to the Einstein--Euler system \eqref{intro:Einstein-Euler}.
\end{proposition}
\section{The nonlinear stability theorem}\label{sec:stability}
In this section, we state the nonlinear stability theorem for the Einstein--Euler system with polytropic equation of state, and begin its proof.
\paragraph{\textit{Bootstrap norms}}
We define several norms which will play a central role in proving Theorem~\ref{sec:stability.thm:main}.

We let $\delta > 0$ denote a small parameter satisfying the bounds
\begin{equation}\label{sec:stability.eq:delta-choice}
     \frac{1}{\sqrt{\kappa}} < \delta < \max\big\{\frac{3}{4n},1-\frac{3}{n}\big\}\,,
\end{equation}
where $n$ is the polytropic index of the fluid, and $\kappa$ is the damping parameter defined by the gauge condition \eqref{sec:gauged-eqs.eq:gauge-choice}. Note that in these bounds, we are already constraining the damping parameter $\kappa$ to be of sufficiently large size. Conversely, this means that $\delta$ can in principle be made arbitrarily small, provided $\kappa$ is large enough.

For a given integer $N \geq 3$, we define the following norms.
\begin{itemize}
    \item Metric averages:
    \begin{subequations}\label{sec:stability.eq:boot}
        \begin{multline}\label{sec:stability.eq:boot-1}
            S_{h;\av}(\tau):= \big|h_\av^{00}(\tau)\big| + \tau^{1+\frac{3}{n}}\sum_{i=1}^3 \big|h_\av^{0i}(\tau)\big|+ \sum_{i,j=1}^3 \big|h_\av^{ij}(\tau)\big|\\
            +  \tau^{1+\frac{3}{n}}\big|\pa_\tau h_\av^{00}(\tau)\big| + \tau^{2+\frac{3}{n}}\sum_{i=1}^3 \big|\pa_\tau h_{\av}^{0i}(\tau)\big| + \tau^{1+\frac{3}{n}}\sum_{i,j=1}^3 \big|\pa_\tau h_\av^{\mu\nu}(\tau)\big| \leq \ve,
        \end{multline}
        \item Fluid averages:
        \begin{equation}\label{sec:stability.eq:boot-2}
            S_{(\rho,v);\av}(\tau):= \big|\rho_\av(\tau)-12\big| + \tau^{1+\frac{3}{n}}\sum_{i=1}^3\big|v_\av^i(\tau)\big|\leq \ve,\\
        \end{equation}
        \item Higher-order metric derivatives:
        \begin{equation}\label{sec:stability.eq:boot-3}
             S_{h;N}(\tau):= \tau^{2+\frac{3}{2n}-\delta}\sum_{\mu=0}^3 \big\|\pa \pa_x h^{0\mu}\big\|_{H^{N-1}} + \tau^{2-\delta}\sum_{i,j=1}^3 \big\|\pa \pa_x h^{ij}\big\|_{H^{N-1}} \leq \ve,\\
        \end{equation}
         \begin{equation}\label{sec:stability.eq:boot-3.5}
             S_{h;N-1}^{\mathbf{v}}(\tau):= \tau^{2+\frac{3}{2n}-\delta}\sum_{\mu=0}^3 \big\|\pa \pav  h^{0\mu}\big\|_{H^{N-1}} + \tau^{2-\delta}\sum_{i,j=1}^3 \big\|\pa \pav  h^{ij}\big\|_{H^{N-1}} \leq \ve,\\
        \end{equation}
        \item Higher-order fluid derivatives:
        \begin{equation}\label{sec:stability.eq:boot-4}
            S_{(\rho,v);N}(\tau):= \tau^{1-\frac{3}{2n}-\delta}\big\|\pa_x \rho\big\|_{H^{N-1}} + \tau^{1+\frac{3}{2n}-\delta}\big\|\pa_x v\big\|_{H^{N-1}} \leq \ve,
        \end{equation}
        where $\| v\| := \sum_{i=1}^3 \|v^i\|$ and $\|\pa_x v\|:= \sum_{i=1}^3 \|\pa_x v^i\|$.
    \end{subequations}
\end{itemize}
We also define the total norm
\begin{equation*}
    S_N(\tau):= S_{h;\av}(\tau) + S_{(\rho,v);\av}(\tau) + S_{h;N}(\tau) + S_{h;N-1}^{\mathbf{v}}(\tau) + S_{(\rho,v);N}(\tau).
\end{equation*}
By the Poincar\'e inequality, the total norm $S_N$ controls the $H^{N+1}$ norm of metric components $g^{\mu\nu}$, and the $H^N$ norms of $\pa_\tau g^{\mu\nu}$, $\rho$, and $v$.

We are now ready to state the main theorem of this article.
\begin{theorem}[Future-stability of the Einstein--de Sitter universe]\label{sec:stability.thm:main}
    Fix a polytropic index $n > 3$, an integer $N \geq 3$, and an initial time $\tau_0 > 0$. Consider initial data
    \begin{gather*}
        g^{\mu\nu}|_{\Sigma_{\tau_0}} = g_0^{\mu\nu} \in H^{N+1}(\TT^3),\qquad \pa_\tau g^{\mu\nu}|_{\Sigma_{\tau_0}} = g_1^{\mu\nu} \in H^{N}(\TT^3), \qquad \mu,\nu=0,1,2,3,\\
        \rho|_{\Sigma_{\tau_0}} = \rho_0 \in H^N(\TT^3), \qquad v^i|_{\Sigma_{\tau_0}} = v_0^i \in H^N(\TT^3), \qquad i=1,2,3
    \end{gather*}
    for the gauge-fixed Einstein--Euler system \eqref{sec:gauged-eqs.eq:gauged-eqs-0}, satisfying the constraints \eqref{sec:gauged-eqs.eq:constraints2}. Then there exists a constant $\ve > 0$ sufficiently small that if $\ve < \ve_0$, $\tau_0$ is sufficiently large, and the initial data satisfies
    \begin{equation}\label{sec:stability.eq:initial-bound}
        \sum_{\mu,\nu=0}^3 \Big(\|g_0^{\mu\nu} - m^{\mu\nu}\|_{H^{N+1}(\Sigma_{\tau_0})} + \|g_1^{\mu\nu}\|_{H^N(\Sigma_{\tau_0})}\Big) + \|\rho_0-12\|_{H^{N}(\Sigma_{\tau_0})} + \sum_{i=1}^3 \|v_0^i\|_{H^N(\Sigma_{\tau_0})} \leq \ve^2,
    \end{equation}
    then the solution $(g,\rho,v)$ to the gauge-fixed system launched by $(g_0^{\mu\nu},g_1^{\mu\nu},\rho_0,v_0^i)$ is future-global and satisfies the bound
    \begin{equation}\label{sec:stability.eq:global-bound}
        S_N(\tau) \leq \ve\,,\qquad \forall\tau\in [\tau_0,\infty)\,.
    \end{equation}
    
    Moreover, the physical variables $(\bar{g},\bar{\rho},\bar{v})$ constitute a future-causal geodesically complete solution to the Einstein--Euler system with polytropic equation of state $\bar{p} = C\bar{\rho}^{1+\frac{1}{n}}$, that converges as $\tau \ra \infty$ to a spatially homogeneous solution to \eqref{intro:Einstein-Euler}. After transforming to the time coordinate $t = \tau^{1/3}$, the physical variables $(\bar{g},\bar{\rho},\bar{v})$ have the asymptotics
    \begin{subequations}\label{sec:stability.eq:asymptotics}
        \begin{align}
            \|\bar{g}_{00} - g_{00}^\infty\|_{W^{N-2,\infty}} &\lesa \tau^{-\frac{3}{n}},\\
            \sum_{i=1}^3\|t^{-\frac{2}{3}}\bar{g}_{0i}\|_{W^{N-2,\infty}} &\lesa \tau^{-(1+\frac{3}{n})},\\
            \sum_{i,j=1}^3\|t^{-\frac{4}{3}}\bar{g}_{ij} - g_{ij}^{\infty}\|_{W^{N-2,\infty}} &\lesa \tau^{-\frac{3}{n}},\\
           \|\tau^2\bar{\rho}-\rho^\infty\|_{W^{N-3,\infty}} &\lesa \tau^{-\min\{\frac{3}{n},1-\frac{3}{2n}-\delta\}},\\
           \sum_{i=1}^3\|t^{\frac{2}{3}}\bar{v}^i\|_{W^{N-3,\infty}} &\lesa \tau^{-(1+\frac{3}{2n}-\delta)}.
        \end{align}
    \end{subequations}
\end{theorem}
\begin{remark}[Initial time $\tau_0$]\label{sec:stability.rmk:initial-time}
    The assumption on $\tau_0$ being sufficiently large is made so that the quantities $(m^{\mu\nu},12,0)$ are close to the quantities induced data of the rescaled asymptotically Einstein--de Sitter solution to the Einstein--Euler system with polytropic equation of state at conformal time $\tau = \tau_0$. Recall from Remark~\ref{sec:gauged-eqs.rmk:rescaling} that $\bar{g}_{n}$, once conformally transformed by $\tau^{-4}$, is not exactly Minkowski, but only asymptotes to Minkowski. 
    
    One can recover a full nonlinear stability statement from Theorem~\ref{sec:stability.thm:main} for arbitrary initial time $\tau_0' > 0$ by the following argument. Take data for the Einstein--Euler system at some initial time $\tau_0' > 0$. Assume that the initial data on $\Sigma_{\tau_0'}$ is sufficiently close to the data of the asymptotically EdS solution,\footnote{The data at time $\tau_0'$ may need to be chosen to be even smaller than $O(\ve^2)$, but this does not affect our argument.} then apply the Cauchy stability of Einstein--Euler to obtain a solution which satisfies the bound \eqref{sec:stability.eq:initial-bound} at a later time $\tau_0 > \tau_0'$. This is always possible, since $(g_n^{\mu\nu},\pa_\tau g_n^{\mu\nu},\rho_n, v_n^i) \ra (m^{\mu\nu},0,12,0)$ as $\tau \ra \infty$.
\end{remark}
\begin{remark}[Asymptotics]
    The asymptotics \eqref{sec:stability.eq:asymptotics} can be derived in a straightforward manner after proving stability by using the global bound \eqref{sec:stability.eq:global-bound}, Lemma~\ref{sec:stability.lem:lowered-met} and the Sobolev embedding. We note that the bounds \eqref{sec:stability.eq:asymptotics} immediately imply the asymptotic description \eqref{sec:intro.eq:asymptotics} stated in the introduction.
\end{remark}
\begin{proof}
    Due to the propagation of the wave gauge condition \eqref{sec:gauged-eqs.eq:gauge-choice}, any solution $(g,\rho,v)$ to the gauge-fixed system  \eqref{sec:gauged-eqs.eq:wave}--\eqref{sec:gauged-eqs.eq:fluid} satisfying the constraints \eqref{sec:gauged-eqs.eq:constraints2} induces a solution $(\bar{g},\bar{\rho},\bar{v})$ to the Einstein--Euler system \eqref{intro:Einstein-Euler}. Additionally, the future-causal geodesic completeness of $\bar{g}$ follows from arguments identical to those given in \protect{\cite[Section 9.2]{Ber:decelerated:25}}. Thus it suffices to
    prove future-global existence of solutions $(g,\rho,v)$ to the gauged-fixed system. We do this by a bootstrap argument.

    Let $(g_0^{\mu\nu},g_1^{\mu\nu},\rho,v^i)$ be an initial data set satisfying the hypotheses of Theorem~\ref{sec:stability.thm:main}, and let $(g,\rho,v)$ be the unique solution launched by that data. We let $I \subset [\tau_0,\infty)$ denote the connected set where $T \in I$ if an only if $(g,\rho,v)$ exists for all $\tau \in [\tau_0,T]$, and satisfies the bound
    \begin{equation}\label{sec:globed.eq:boot}
        S_N(\tau) \leq \ve
    \end{equation}
    for all $\tau \in [\tau_0,T]$. Henceforth, we call the bound \eqref{sec:globed.eq:boot} the \emph{bootstrap assumption}. Clearly $I$ is nonempty by the smallness assumption \eqref{sec:stability.eq:initial-bound} on the initial data, and the local well-posedness result of Proposition~\ref{prop:local}. Moreover, $I$ is closed by continuity properties of the equations. To show that $I$ is open, observe that if $T \in I$, then by local well-posedness, the solution $(g,\rho,v)$ extends uniquely to a small neighbourhood around $T$. If we can show that for all $T \in I$, the bootstrap assumption \eqref{sec:globed.eq:boot} can be uniformly improved, then we are done. We do this in two parts. First, we improve the bound on the spatial average norms $S_{h;\av}$, $S_{(\rho,v);\av}$, which we prove in Proposition~\ref{sec:avg-ests.prop:avg-ests}. Then, we improve the bound on the higher-order norms $S_{h;N}$, $S_{(\rho,v);N}$ in Theorem~\ref{sec:boot-imp.thm:higher-ests}.
\end{proof}
In the remainder of this article, we assume the bootstrap assumption \eqref{sec:globed.eq:boot} always holds, whether explicitly stated or not. We also assume that the initial time $\tau_0$ is sufficiently large that
\begin{equation*}
    \tau_0 \geq \ve^{-4n},
\end{equation*}
where $n$ is the polytropic index. We will use this assumption only twice, in Proposition~\ref{sec:avg-ests.prop:avg-ests} and in Theorem~\ref{sec:boot-imp.thm:higher-ests}, to show that rapidly decaying error terms are in fact small for all $\tau \geq \tau_0$.

\subsection{Preliminary estimates}
We collect some important results that generally follow from the bootstrap assumptions. These allow us to estimate various quantities appearing in estimates that are not explicitly controlled by the norm $S_{N}$.
\begin{lemma}[Lowered metric component estimates]\label{sec:stability.lem:lowered-met}
    The following pointwise and $L^2$-based bounds hold for the components of the (lowered) metric $g_{\mu\nu}$:
    \begin{subequations}
        \begin{align}
            \label{sec:stability.eq:lowered-met-1}
            \Big|g_{00} - (g^{00})^{-1}\Big| &\lesa \sum_{i,j=1}^3 |h^{0i}|^2\,,\\
            \label{sec:stability.eq:lowered-met-2}
            \sum_{i=1}^3 |g_{0i}| &\lesa \sum_{i,j=1}^3 |h^{0i}|\,,\\
            \label{sec:stability.eq:lowered-met-3}
            \sum_{\mu,\nu=0}^3|g_{\mu\nu}-m_{\mu\nu}| &\lesa \sum_{\mu,\nu=0}^3 |g^{\mu\nu} - m^{\mu\nu}|\,.
        \end{align}
        \begin{multline}\label{sec:stability.eq:lowered-met-4}
            \big|(g_{00})_\av + 1\big| + \tau^{1+\frac{3}{n}}\sum_{i=1}^3\big|(g_{0i})_\av\big| + \sum_{i,j=1}\big|(g_{ij}-\delta_{ij})\big|\\
            +\tau^{1+\frac{3}{n}}\big|\pa_\tau (g_{00})_\av\big| + \tau^{2+\frac{3}{n}}\sum_{i=1}^3 \big|\pa_\tau (g_{0i})_\av\big| + \tau^{1+\frac{3}{n}}\sum_{i,j=1}^3 \big|\pa_\tau (g_{ij})_\av\big| \lesa \ve,
        \end{multline}
        \begin{equation}\label{sec:stability.eq:lowered-met-5}
            \tau^{2+\frac{3}{2n}-\delta}\sum_{\mu=0}^3 \big\|\pa \pa_x g_{0\mu}\big\|_{H^{N-1}} + \tau^{2-\delta}\sum_{i,j=1}^3 \big\|\pa\pa_x g_{ij}\big\|_{H^{N-1}} \lesa \ve.
        \end{equation}
    \end{subequations}
\end{lemma}
\begin{proof}
    The pointwise bounds for the lowered metric components are fairly standard--the rough bound \eqref{sec:stability.eq:lowered-met-3} follows from $g^{\mu\nu}$ being both the inverse of $g_{\mu\nu}$ and close to the Minkowski metric. This rough bound then implies \eqref{sec:stability.eq:lowered-met-1} and \eqref{sec:stability.eq:lowered-met-2} from rearranging the identities
    \begin{align*}
        1 &= g_{0\mu} g^{0\mu} = g_{00}g^{00} + g_{0a}g^{0a} ,\\
        0 &= g_{i\mu} g^{0\mu} = g_{0i}g^{00} + g_{ia}g^{0a},
    \end{align*}
    to obtain
    \begin{align*}
        g_{00} - (g^{00})^{-1} &= -(g^{00})^{-1}g_{0a}g^{0a},\\
        g_{0i} &= -(g^{00})^{-1}g_{ia} g^{0a}.
    \end{align*}
    The bound \eqref{sec:stability.eq:lowered-met-4} can be derived in a similar manner using the bootstrap assumptions, while \eqref{sec:stability.eq:lowered-met-5} follows from repeated applications of the identity
    \begin{equation*}
        \pa_\lambda g_{\mu\nu} = -g_{\mu\alpha}g_{\nu\beta}\pa_\lambda g^{\mu\alpha},
    \end{equation*}
    along with the bounds \eqref{sec:stability.eq:lowered-met-1} -  \eqref{sec:stability.eq:lowered-met-3}.
\end{proof}
\begin{lemma}[Miscellaneous Estimates for the fluid four-velocity components]\label{sec:boot:lem.misc-fluid}
    The following estimates hold for the zeroth components of the fluid four-velocity
    \begin{subequations}
        \begin{align}
            \label{sec:boot:eq.misc-fluid-1}
            \big|(v^0)^2 + g^{00}\big| &\lesa \sum_{i=1}^3 \big((v^i)^2 + (h^{0i})^2\big),\\
            \label{sec:boot:eq.misc-fluid-2}
            |(v_0)^2 + g_{00}| &\lesa \sum_{i=1}^3 \big((v^i)^2 + (h^{0i})^2\big),\\
            \label{sec:boot:eq.misc-fluid-3}
            \|\pa_x v^0\|_{H^{N-1}} &\lesa \frac{\ve}{\tau^{1+\frac{3}{2n}-\delta}},
        \end{align}
    \end{subequations}
    and for the lowered fluid four-velocity components:
    \begin{subequations}
        \begin{align}
            \label{sec:boot:eq.misc-fluid-4}
            \sum_{i=1}^3 (v_i)^2 &\lesa \sum_{i=1}^3 \big((v^i)^2 + (h^{0i})^2\big),\\
            \label{sec:boot:eq.misc-fluid-5}
            \sum_{i=1}|(v_i)_\av| &\lesa \frac{\ve}{\tau^{1+\frac{3}{n}}},\\
            \label{sec:boot:eq.misc-fluid-6}
            \sum_{i=1}^3\|\pa_x v_i\|_{H^{N-1}} &\lesa \frac{\ve}{\tau^{1+\frac{3}{2n}-\delta}}\,.
        \end{align}
    \end{subequations}
\end{lemma}
\begin{proof}
    We rearrange the constraint
    \begin{equation*}
        -1 = g_{\alpha\beta}v^\alpha v^\beta = g_{00} (v^{0})^2 + 2g_{0a}v^0 v^a + g_{ab}v^a v^b,
    \end{equation*}
    and solve for $v^0$, which by Lemma~\ref{sec:stability.lem:lowered-met} implies the bound
    \begin{equation*}
        \big|(v^0)^2 + g^{00}\big| \leq \big|(v^0)^2 + (g_{00})^{-1}\big| + \big|(g_{00})^{-1} - g^{00}\big| \lesa \sum_{i=1}^3 \big(|v^i|^2 +|h^{0i}|^2\big)\,.
    \end{equation*}
    This gives the pointwise bound \eqref{sec:boot:eq.misc-fluid-1} and the norm bound \eqref{sec:boot:eq.misc-fluid-3}. The remaining bounds all follow from the relation $v_\mu = g_{\mu\alpha}v^\alpha$ and the previous estimates.
\end{proof}

Next, we compute the Christoffel symbols that appear in the Euler equations. Using the gauge condition \eqref{sec:gauged-eqs.eq:gauge-choice}, we find that these Christoffel symbols depend (to leading order) on the lapse and shift perturbations, and not on the spatial metric perturbations.
\begin{lemma}[Christoffel symbol estimates]\label{sec:stability.lem:christoffel}
    The Christoffel symbols satisfy
    \begin{subequations}
        \begin{align}
            \label{sec:stability.eq:christoffel-1}
            \Gamma_{0\lambda}^\lambda &= -\frac{1}{g^{00}}\big(\pa_\lambda h^{0\lambda} + \Gamma^0\big) + \frac{1}{2}g^{00} g_{\alpha\beta} h^{0a}\pa_a h^{\alpha\beta},\\
            \label{sec:stability.eq:christoffel-2}
            \Gamma_{00}^i &= -g_{00}\pa_\tau h^{0i}-\frac{1}{2}g^{i\lambda}\pa_\lambda g_{00}  - g_{0a}\pa_\tau h^{ia}.
        \end{align}
    \end{subequations}
    Moreover, we have the bounds 
    \begin{subequations}\label{sec:stability.eq:christoffel-3}
        \begin{align}
            \big\|\chi v^\mu\Gamma_{\mu\lambda}^\lambda\big\|_{H^{N}} &\lesa \frac{1}{\tau^{1+\frac{3}{n}}},\\
            \big\|v^\mu v^\nu \Gamma_{\mu\nu}^i\big\|_{H^{N}} &\lesa \frac{1}{\tau^{2+\frac{3}{2n}-\delta}}\,,
        \end{align}
    \end{subequations}
    and
    
    \begin{subequations}\label{sec:stability.eq:christoffel-4}
        \begin{align}
            \big\|\pa_x \big(\chi v^\mu\Gamma_{\mu\lambda}^\lambda\big)\big\|_{H^{N-1}} &\lesa \frac{1}{\tau^{2+\frac{3}{2n}-\delta}}\,,\\
            \big\|\pa_x \big(v^\mu v^\nu \Gamma_{\mu\nu}^i\big)\big\|_{H^{N-1}} &\lesa \sum_{\mu=0}^3 \big\|\pa \pa_xh^{0\mu}\big\|_{H^{N-1}} + \frac{1}{\tau^{2+\frac{9}{2n}-\delta}}\,.
        \end{align}
    \end{subequations}
\end{lemma}
\begin{proof}
    We observe that the contracted Christoffel symbol $\Gamma^0$ decomposes like
    \begin{equation*}
        \Gamma^0 = \frac{1}{2}g^{\alpha\beta}g^{0\lambda}\big(2\pa_{(\alpha}g_{\beta)\lambda}-\pa_\lambda g_{\alpha\beta}\big)= -\pa_\lambda h^{0\lambda}-\frac{1}{2}g^{00} g^{\alpha\beta}\pa_\tau g_{\alpha\beta} + \frac{1}{2}h^{0a}g_{\alpha\beta}\pa_a h^{\alpha\beta}\,,
    \end{equation*}
    and so we have 
    \begin{equation*}
        \Gamma_{0\lambda}^\lambda = \frac{1}{2}g^{\alpha\beta}\pa_\tau g_{\alpha\beta} = -\frac{1}{g^{00}}\big(\pa_\lambda h^{0\lambda}+\Gamma^0\big) + \frac{1}{2g^{00}}g_{\alpha\beta}h^{0a}\pa_a h^{\alpha\beta}\,,
    \end{equation*}
    which is \eqref{sec:stability.eq:christoffel-1}. For \eqref{sec:stability.eq:christoffel-2}, we compute
    \begin{equation*}
        \Gamma_{00}^i = g^{i\lambda}\pa_\tau g_{0\lambda}-\frac{1}{2} g^{i\lambda}\pa_\lambda g_{00} = -g_{00}\pa_\tau h^{0i} - g_{0a}\pa_\tau h^{ia}-\frac{1}{2} g^{i\lambda}\pa_\lambda g_{00}\,.
    \end{equation*}
    These identities then imply the $H^N$-norm bounds \eqref{sec:stability.eq:christoffel-3}.

    In proving the remaining bounds \eqref{sec:stability.eq:christoffel-4} we denote all combinations of Christoffel symbols by
    $\Gamma$, and note that by the bootstrap assumptions, all Christoffel symbols obey the rough bounds
    \begin{equation*}
        \|\Gamma\|_{H^N} \lesa \frac{\ve}{\tau^{1+\frac{3}{n}}}\,,\qquad \big\|\pa_x \Gamma\big\|_{H^{N-1}} \lesa \frac{\ve}{\tau^{2-\delta}}\,.
    \end{equation*}
    We apply the Sobolev product Lemma and the Leibniz rule to obtain
    \begin{multline*}
        \Big\|\pa_x(\chi v^\mu \Gamma_{\mu\lambda}^\lambda )\Big\|_{H^{N-1}}\\
        \lesa \big\|\pa_x \Gamma_{0\lambda}^\lambda\big\|_{H^{N-1}} + \big(\|\pa_x\chi\|_{H^{N-1}} + \|\pa_x v\|_{H^{N-1}}\big)\|\Gamma\|_{H^{N-1}} + \|v\|_{H^{N-1}}\big\|\pa_x\Gamma\big\|_{H^{N-1}}\\
        \lesa \frac{1}{\tau^{2+\frac{3}{2n}-\delta}}\,,
    \end{multline*}
    where we have bounded $\|\pa_x \Gamma_{0\lambda}^\lambda\|_{H^{N-1}}$ using Lemma~\ref{sec:stability.eq:christoffel-1} and the wave gauge condition \eqref{sec:gauged-eqs.eq:gauge-choice}. We also have by \eqref{sec:stability.eq:christoffel-2} that
    \begin{multline*}
        \Big\|\pa_x \big(v^\mu v^\nu \Gamma_{\mu\nu}^i\big)\Big\|_{H^{N-1}} \lesa \big\|\pa_x \Gamma_{00}^i\big\|_{H^{N-1}} + \|\pa_x v\|_{H^{N-1}}\|\Gamma\|_{H^{N-1}} + \|v\|_{H^{N-1}}\big\|\pa_x \Gamma\big\|_{H^{N-1}}\\
        \lesa \sum_{\mu=0}^3 \big\|\pa \pa_xh^{0\mu}\big\|_{H^{N-1}} + \frac{1}{\tau^{2+\frac{9}{2n}-\delta}}\,.
    \end{multline*}
    This completes the proof.
\end{proof}

As a consequence of the Euler equations and the previous lemma, we also have improved bounds for certain derivatives of $(\rho,v)$.
\begin{proposition}[Improved fluid variable decay]\label{sec:stability.prop:fluid-time-deriv}
    The time derivatives of the fluid evolution variables obey the bounds 
    \begin{subequations}
        \begin{align}
            \|\pa_\tau \rho\|_{H^{N-1}} &\lesa \frac{1}{\tau^{1+\frac{3}{2n}-\delta}}S_{(\rho,v);N}(\tau) + \frac{1}{\tau^{1+\frac{3}{n}}}\,, \\
            \|\pa_\tau v\|_{H^{N-1}} &\lesa \frac{1}{\tau^{1+\frac{9}{2n}-\delta}}S_{(\rho,v);N}(\tau) + \frac{1}{\tau^{2+\frac{3}{2n}-\delta}}\,.
        \end{align}
    \end{subequations}
    The same bounds also hold if we replace $\pav$ with the coordinate time derivative $\pa_\tau$.
\end{proposition}
\begin{remark}
    These estimates show that certain derivatives of the fluid variables exhibit improved decay in comparison to their spatial derivatives. Using the bootstrap assumptions, we find that $\|\pa_\tau \rho\|_{H^{N-1}} \lesa \ve \tau^{-(1+\frac{3}{2n}-\delta)}$, and $\|\pa_\tau v\|_{H^{N-1}} \lesa \ve \tau^{-(1 + \frac{9}{2n}-\delta)}$, which is an improvement in decay of order $\tau^{-\frac{3}{n}}$ compared to the spatial derivatives. This observation will be crucial when deriving commuted energy estimates.
\end{remark}
\begin{proof}
    We rearrange the Euler equations to find 
    \begin{align}
        \label{sec:stability.eq:fluid-time-deriv-proof-1}
        \pa_\tau \rho &= \frac{\chi}{v_0v^0}v_a\pa_\tau v^a +\frac{\chi}{2v_0 v^0}v^\alpha v^\beta \pa_\tau g_{\alpha\beta}- \frac{\chi}{v^0}\pa_a v^a - \frac{\chi}{v^0}v^\mu \Gamma_{\mu\lambda}^\lambda - \frac{1}{v^0}\Big[v^a\pa_a \rho + \frac{6}{\tau^{1+\frac{6}{n}}}pv^0\Big]\,, \\
        \label{sec:stability.eq:fluid-time-deriv-proof-2}
        \pa_\tau v^i &= - \frac{p'(\rho)}{\tau^{\frac{6}{n}}v^0\chi}\Pi^{0i}\pa_\tau \rho - \frac{p'(\rho)}{\tau^{\frac{6}{n}}v^0\chi}\Pi^{ia}\pa_a \rho - \frac{1}{v^0}v^\mu v^\nu \Gamma_{\mu\lambda}^i - \frac{1}{v^0}\Big[v^a \pa_a v^i + \Big(\frac{2}{\tau} + \frac{6\rho p'(\rho)}{\tau^{1+\frac{6}{n}}\chi}\Big)\Pi^{0i}\Big]\,,
    \end{align}
    where we used the identity $v_\mu v^\mu= - 1$ to compute the time derivative of $v^0$ in the continuity equation. We use these equations, the Sobolev product inequality and Lemma~\ref{sec:stability.lem:christoffel} to derive the bounds
    \begin{align}
        \label{sec:stability.eq:fluid-time-deriv-proof-3}
        \|\pa_\tau\rho\|_{H^{N-1}} &\lesa \frac{\ve}{\tau^{1+\frac{3}{2n}-\delta}}\|\pa_\tau v\|_{H^{N-1}} + \|\pa_x v\|_{H^{N-1}} + \frac{1}{\tau^{1+\frac{3}{n}}}\,, \\
        \label{sec:stability.eq:fluid-time-deriv-proof-4}
        \|\pa_\tau v\|_{H^{N-1}} &\lesa \frac{\ve}{\tau^{1+\frac{15}{2n}-\delta}}\|\pa_\tau \rho\|_{H^{N-1}} + \frac{1}{\tau^{\frac{6}{n}}}\|\pa_x \rho\|_{H^{N-1}} + \frac{1}{\tau^{2+\frac{3}{2n}-\delta}}\,. 
    \end{align}
    We note that the worst-behaved error term in \eqref{sec:stability.eq:fluid-time-deriv-proof-1} is given by $\chi (v^0)^{-1}v^\mu \Gamma_{\mu\lambda}^\lambda$, whose $H^{N-1}$ norm decays like the error term in \eqref{sec:stability.eq:fluid-time-deriv-proof-3}. Several terms in \eqref{sec:stability.eq:fluid-time-deriv-proof-2} decay at the worst rate given in \eqref{sec:stability.eq:fluid-time-deriv-proof-4}.
    
    Combining the inequalities  \eqref{sec:stability.eq:fluid-time-deriv-proof-3} and  \eqref{sec:stability.eq:fluid-time-deriv-proof-4}, we find that
    \begin{multline*}
        \big\|\pa_\tau\rho\big\|_{H^{N-1}} + \tau^{\frac{3}{n}}\|\pa_\tau v\|_{H^{N-1}} \lesa \frac{\ve}{\tau^{1+\frac{9}{2n}-\delta}}\Big(\big\|\pa_\tau\rho\big\|_{H^{N-1}} + \tau^{\frac{3}{n}}\|\pa_\tau v\|_{H^{N-1}}\Big)\\+\frac{1}{\tau^{\frac{3}{n}}}\Big(\big\|\pa_x\rho\big\|_{H^{N-1}} + \tau^{\frac{3}{n}}\|\pa_x v\|_{H^{N-1}}\Big) + \frac{1}{\tau^{1+\frac{3}{n}}}+ \frac{1}{\tau^{2-\frac{3}{2n}-\delta}}\,.
    \end{multline*}
    The first collection of terms on the right hand side can be absorbed into the left hand side, and the terms $\|\pa_x \rho\|_{H^{N-1}}$, $\tau^{\frac{3}{n}}\|\pa_x v\|_{H^{N-1}}$ can be bounded in terms of the norm $S_{(\rho,v);N}$, giving the combined estimate
    \begin{equation*}
        \big\|\pa_\tau\rho\big\|_{H^{N-1}} + \tau^{\frac{3}{n}}\|\pa_\tau v\|_{H^{N-1}} \lesa \frac{1}{\tau^{1+\frac{3}{2n}-\delta}}S_{(\rho,v);N} + \frac{1}{\tau^{1+\frac{3}{n}}}+ \frac{1}{\tau^{2-\frac{3}{2n}-\delta}}\,. 
    \end{equation*}
    Plugging this estimates back into \eqref{sec:stability.eq:fluid-time-deriv-proof-3}, \eqref{sec:stability.eq:fluid-time-deriv-proof-4} yield the desired bound.
\end{proof}
\section{ODE Estimates for the spatial averages}\label{sec:avg-ests}
In this section, we analyse the behaviour of the spatial averages of the evolution variables, given by
\begin{equation}
    h_\av^{\mu\nu}(\tau):= \fint_{\TT^3} h^{\mu\nu}(\tau,x)\de^3 x\,, \qquad \rho_\av := \fint_{\TT^3} \rho(\tau,x)\de^3 x\,, \qquad v_\av^i(\tau):= \fint_{\TT^3} v^i(\tau,x)\de^3 x \,.
\end{equation}
An independent analysis of these quantities is necessary, as the dynamics of the spatial averages $(h_\av^{\mu\nu},\rho_\av,v_\av^i)$ largely decouple from dynamics of the remainders $(h^{\mu\nu}-h_\av^{\mu\nu},\rho-\rho_\av,v^i-v_\av^i)$.\footnote{These remainder quantities are controlled by the top-order derivatives of the evolution variables by virtue of the Poincar\'e inequality.} The main result of this section are estimates which improve on the bootstrap assumptions \eqref{sec:stability.eq:boot-1}--\eqref{sec:stability.eq:boot-2} for the spatial averages of the evolution variables, contained in Proposition~\ref{sec:avg-ests.prop:avg-ests}.
\subsection{The ODE system for the spatial averages}
From the various equations in the system, we derive ODE satisfied by the spatial averages of the evolution variables, beginning with the metric components.
\begin{lemma}[ODE for the averaged metric components]\label{sec:avg-ests.lem:avg-met}
	The spatial averages $h_\av^{\mu\nu}$ satisfy the system of ODE
    \begin{subequations}
	   \begin{align}
		      \label{sec:avg-ests.eq:avg-met-1}
            \pa_\tau h_\av^{00} + \frac{6}{\tau} h_\av^{00} &= \frac{1}{2}(g_{ab}g^{00})_\av\pa_\tau  h_\av^{ab} -\frac{1}{2\tau}(\rho-12)_\av + \mc{H}_h^{00}\,, \\
            \label{sec:avg-ests.eq:avg-met-2}
		    \pa_\tau h_\av^{0i} + \frac{\kappa}{\tau} h_\av^{0i} &= \mc{H}_h^{0i}\,,\\
            \label{sec:avg-ests.eq:avg-met-3}
		    \pa_\tau^2 h_\av^{ij} + \frac{4}{\tau} \pa_\tau h_\av^{ij} &= \mc{H}_h^{ij}\,,
        \end{align}
    \end{subequations}
	where $\mc{H}_h^{\mu\nu}$, $\mc{H}_\rho$, $\mc{H}_v^i$ are the error terms
		\begin{align*}
		    \mc{H}_h^{00} :=&\,\int_{\TT^3} \Big(\frac{1}{2}g_{\alpha\beta}h^{0a}\pa_a h^{\alpha\beta}+ h_{0a}g^{00}\pa_\tau h^{0a} -\frac{1}{2}h_{0a}h^{0a}\pa_\tau h^{00} \Big) \\&+ \frac{1}{2}\int_{\TT^3} \big[g_{ab}g^{00}-(g_{ab}g^{00})_\av\big]\pa_\tau (h^{ab}-h_\av^{ab})\\
            \mc{H}_h^{0i} :=&\, \frac{1}{2}\int_{\TT^3} \Big(g_{\alpha\beta}h^{0i}\pa_\tau h^{\alpha\beta} + g_{\alpha\beta}g^{ia}\pa_a h^{\alpha\beta}\Big),\\
            \mc{H}_h^{ij} :=&\,-\int_{\TT^3} (g^{00})^{-1}\Big(2h^{0a}\pa_a\pa_\tau h^{ij} + g^{ab}\pa_a\pa_b h^{ij}\Big) + \int_{\TT^3} \mc{L}_h^{ij} + \int_{\TT^3} \mc{G}_h^{ij},
		\end{align*}
		\end{lemma}
\begin{proof}
    The ODE \eqref{sec:avg-ests.eq:avg-met-1}, \eqref{sec:avg-ests.eq:avg-met-2} are derived from the wave gauge conditions \eqref{sec:gauged-eqs.eq:gauge-choice}, while the ODE \eqref{sec:avg-ests.eq:avg-met-3} are derived from the wave equations \eqref{sec:gauged-eqs.eq:gauged-eqs-3} for the spatial metric perturbations.

    From \cite{Ber:decelerated:25} (see in particular equations $(6.24) - (6.25)$), the contracted Christoffel symbols $\Gamma^\mu = g^{\alpha\beta}\Gamma_{\alpha\beta}^\mu(g)$ satisfy
    \begin{align*}
        \Gamma^0 &=-\frac{1}{2}\pa_\tau h^{00}-\pa_a h^{0a} +\frac{1}{2}g_{ab}g^{00}\pa_\tau h^{ab}+\frac{1}{2}g_{\alpha\beta}h^{0a}\pa_a h^{\alpha\beta}+ h_{0a}g^{00}\pa_\tau h^{0a} -\frac{1}{2}h_{0a}h^{0a}\pa_\tau h^{00},\\
        \Gamma^i &= -\pa_\tau h^{0i}-\pa_a h^{ia} + \frac{1}{2}g_{\alpha\beta}h^{0i} \pa_\tau h^{\alpha\beta} + \frac{1}{2}g_{\alpha\beta}g^{ia} \pa_a h^{\alpha\beta}.
    \end{align*}
    Combining these with the wave gauge equations \eqref{sec:gauged-eqs.eq:gauge-choice} and rearranging terms gives rise to the equations
    \begin{subequations}
        \begin{align}
            \label{sec:avg-ests.eq:avg-met-proof-1}
            \pa_\tau h^{00} + \frac{6}{\tau}h^{00} = & -2\pa_a h^{0a} + \frac{1}{2}g_{ab}g^{00}\pa_\tau h^{ab} + \frac{1}{2\tau}(\rho_\av-12)\\
            &+\frac{1}{2}g_{\alpha\beta}h^{0a}\pa_a h^{\alpha\beta}+ h_{0a}g^{00}\pa_\tau h^{0a} -\frac{1}{2}h_{0a}h^{0a}\pa_\tau h^{00} + \frac{2(\kappa-3)}{\tau}(h^{00}-h_\av^{00}),\nonumber\\
            \label{sec:avg-ests.eq:avg-met-proof-2}
            \pa_\tau h^{0i} + \frac{\kappa}{\tau}h^{0i} = &- \pa_a h^{ia} + \frac{1}{2}g_{\alpha\beta}h^{0i} \pa_\tau h^{\alpha\beta} + \frac{1}{2}g_{\alpha\beta}g^{ia} \pa_a h^{\alpha\beta}.
        \end{align}
    \end{subequations}
    We integrate these over the torus. For \eqref{sec:avg-ests.eq:avg-met-proof-1}, this gives
    \begin{multline*}
        \pa_\tau h_\av^{00} + \frac{6}{\tau}h_\av^{00} = - 2\underbrace{\int_{\TT^3} \pa_a h^{0a}}_{=0}+ \frac{1}{2}\int_{\TT^3} g_{ab} g^{00} \pa_\tau h^{ab} + \frac{1}{2\tau}(\rho_\av -12) \\
        + \int_{\TT^3} \Big(\frac{1}{2}g_{\alpha\beta}h^{0a}\pa_a h^{\alpha\beta}+ h_{0a}g^{00}\pa_\tau h^{0a} -\frac{1}{2}h_{0a}h^{0a}\pa_\tau h^{00}\Big) + \frac{2(\kappa-3)}{\tau}\underbrace{\int_{\TT^3} \big(h^{00}-h_\av^{00}\big)}_{=0}\,.
    \end{multline*}
    The underbraced integrals vanish, as their integrands clearly have vanishing spatial average. After writing
    \begin{equation*}
        \frac{1}{2}\int_{\TT^3} g_{ab} g^{00} \pa_\tau h^{ab} =\frac{1}{2}(g_{ab}g^{00})_\av\pa_\tau  h_\av^{ab} + \frac{1}{2}\int_{\TT^3} \big[g_{ab}g^{00}-(g_{ab}g^{00})_\av\big]\pa_\tau (h^{ab}-h_\av^{ab})\,,
    \end{equation*}
    we find that $h_{\av}^{00}$ obeys the equation
    \begin{multline*}
        \pa_\tau h_\av^{00} + \frac{6}{\tau}h_\av^{00} = \frac{1}{2}(g_{ab}g^{00})_\av\pa_\tau  h_\av^{ab} + \frac{1}{2\tau}(\rho_\av -12) \\
        + \int_{\TT^3} \Big(\frac{1}{2}g_{\alpha\beta}h^{0a}\pa_a h^{\alpha\beta}+ h_{0a}g^{00}\pa_\tau h^{0a} -\frac{1}{2}h_{0a}h^{0a}\pa_\tau h^{00} \Big) \\+ \frac{1}{2}\int_{\TT^3} \big[g_{ab}g^{00}-(g_{ab}g^{00})_\av\big]\pa_\tau (h^{ab}-h_\av^{ab})\,,
    \end{multline*}
    which is \eqref{sec:avg-ests.eq:avg-met-1}.
    Next we take the average of \eqref{sec:avg-ests.eq:avg-met-proof-2} and find that
    \begin{equation}
        \pa_\tau h_\av^{0i} + \frac{\kappa}{\tau}h_\av^{0i} = -\underbrace{\int_{\TT^3} \pa_a h^{ia}}_{=0} + \frac{1}{2}\int_{\TT^3} \Big(g_{\alpha\beta}h^{0i}\pa_\tau h^{\alpha\beta} + g_{\alpha\beta}g^{ia}\pa_a h^{\alpha\beta}\Big)\,,
    \end{equation}
    is just \eqref{sec:avg-ests.eq:avg-met-2}.
    Finally, to derive \eqref{sec:avg-ests.eq:avg-met-3}, we multiply the evolution equation \eqref{sec:gauged-eqs.eq:gauged-eqs-3} by $(g^{00})^{-1}$ and integrate over the torus. Since
    \begin{align*}
        \int_{\TT^3} (g^{00})^{-1}\Big(\wh{\square}_g h^{ij} + \frac{4}{\tau} g^{00}\pa_\tau h^{ij}\Big) &= \pa_\tau^2 h_\av^{ij} + \frac{4}{\tau}\pa_\tau h_\av^{ij} + \int_{\TT^3} (g^{00})^{-1}\Big(2h^{0a}\pa_a\pa_\tau h^{ij} + g^{ab}\pa_a\pa_b h^{ij}\Big)\,,
    \end{align*}
    we see that $h_\av^{ij}$ satisfies the ODE
    \begin{equation}
        \pa_\tau^2 h_\av^{ij} + \frac{4}{\tau}\pa_\tau h_\av^{ij} = -\int_{\TT^3} (g^{00})^{-1}\Big(2h^{0a}\pa_a\pa_\tau h^{ij} + g^{ab}\pa_a\pa_b h^{ij}\Big) + \int_{\TT^3} \mc{L}_h^{ij} + \int_{\TT^3} \mc{G}_h^{ij},
    \end{equation}
    which is \eqref{sec:avg-ests.eq:avg-met-3}.    
\end{proof}

\begin{lemma}[ODE for the fluid variables]\label{sec:avg-sys.lem:avg-fluid}
    The spatial averages $\rho_\av$, $v_\av$ satisfy the ODE
    \begin{subequations}\label{sec:avg-ests.eq:avg-fluid}
        \begin{gather}		
            \pa_\tau \Big((\rho_\av-12) + 12h_\av^{00}\Big) + \frac{3}{\tau}\Big((\rho_\av-12) + 12 h_\av^{00}\Big) = \mc{H}_\rho\,,\label{sec:avg-ests.eq:avg-fluid-1}\\
		      \pa_\tau \big(v^i_\av + h_\av^{0i}\big) + \frac{2}{\tau} \big(v_\av^i + h_\av^{0i}\big)= \mc{H}_v^i\,,\label{sec:avg-ests.eq:avg-fluid-2}
        \end{gather}
    \end{subequations}
    where the error terms $\mc{H}_\rho$, $\mc{H}_v^i$ are given by
        \begin{align*}
            \mc{H}_\rho :=&\, -\int_{\TT^3} \frac{1}{v^0}\Big[\chi \pa_\mu v^\mu + v^a\pa_a\rho +\mc{G}_{(\rho,v)}\Big]\\
            &+ \int_{\TT^3} \Big(\frac{\chi}{g^{00}} + 12\Big)\big(\pa_\lambda h^{0\lambda} + \Gamma^0\big)-\int_{\TT^3} \frac{\chi}{2g^{00}}g_{\alpha\beta}h^{0a}\pa_a h^{\alpha\beta}\,,\\
            \mc{H}_v^i :=&\,-\int_{\TT^3} \frac{1}{v^0}\Big[\frac{1}{\tau^{\frac{6}{n}} \chi}\Pi^{\mu i}\pa_\mu \rho +v^a\pa_a v^i +\mc{G}_{(\rho,v)}^i \Big]\\
            &- \frac{2}{\tau}\int_{\TT^3} \big[(v^0)^{-1}-1\big]h^{0i} +\int_{\TT^3} (v^0 g_{00} + 1)\pa_\tau h^{0i}+\int_{\TT^3} v^0\Big(g_{0a}\pa_\tau h^{ia} + \frac{1}{2}g^{i\lambda}\pa_\lambda g_{00}\Big)\,.
        \end{align*}
\end{lemma}
\begin{proof}
    The identities \eqref{sec:avg-ests.eq:avg-fluid-1}, \eqref{sec:avg-ests.eq:avg-fluid-2} are obtained by multiplying the continuity and momentum equations \eqref{sec:gauged-eqs.eq:fluid} by $(v^0)^{-1}$ and integrating over the torus. 
    This gives the equations
    \begin{subequations}
        \begin{align}
            \label{sec:avg-ests.eq:avg-fluid-proof-1}
            \pa_\tau (\rho_\av -12) &= - \int_{\TT^3} \chi \Gamma_{0\lambda}^\lambda -\int_{\TT^3} \frac{1}{v^0}\Big[\chi \pa_\mu v^\mu + v^a \pa_a \rho + \mc{G}_{(\rho,v)}\Big]\,,\\
            \label{sec:avg-ests.eq:avg-fluid-proof-2}
            \pa_\tau v_\av^i + \frac{2}{\tau} v_\av^i &= - \int_{\TT^3} v^0 \Gamma_{00}^i  -\frac{2}{\tau}\int_{\TT^3} \frac{h^{0i}}{v^0} -\int_{\TT^3} \frac{1}{v^0}\Big[\frac{1}{\tau^{\frac{6}{n}}\chi}\Pi^{\mu i}\pa_\mu \rho + v^a \pa_a v^i + \mc{G}_{(\rho,v)}^i\Big]\,.
        \end{align}
    \end{subequations}
    The Christoffel symbols in the first two equations can be rewritten with the identities from Lemma~\ref{sec:stability.lem:christoffel}. For \eqref{sec:avg-ests.eq:avg-fluid-proof-1}, this and the gauge condition imply
    \begin{align*}
        -\int_{\TT^3} \chi\Gamma_{0\lambda}^\lambda =&\, \int_{\TT^3} \frac{\chi}{g^{00}}\big[\pa_\lambda h^{0\lambda} + \Gamma^0\big] - \int_{\TT^3} \frac{\chi}{2g^{00}}g_{\alpha\beta}h^{0a}\pa_a h^{\alpha\beta}\\
        =&\, -12\int_{\TT^3} \big[\pa_\tau h^{00} + \Gamma^0\big] + \int_{\TT^3} \Big(\frac{\chi}{g^{00}}+12\Big)\big[\pa_\lambda h^{0\lambda} + \Gamma^0\big] -  \int_{\TT^3} \frac{\chi}{2g^{00}}g_{\alpha\beta}h^{0a}\pa_a h^{\alpha\beta}\\
        =&\,-12\pa_\tau h_\av^{00} -\frac{3}{\tau}\big((\rho_\av-12) + 12h_\av^{00}\big)\\
        &+ \int_{\TT^3} \Big(\frac{\chi}{g^{00}}+12\Big)\big[\pa_\lambda h^{0\lambda} + \Gamma^0\big] -  \int_{\TT^3} \frac{\chi}{2g^{00}}g_{\alpha\beta}h^{0a}\pa_a h^{\alpha\beta}.
    \end{align*}
    Plugging this into \eqref{sec:avg-ests.eq:avg-fluid-proof-1} and rearranging yields \eqref{sec:avg-ests.eq:avg-fluid-1}. Moreover, we have
    \begin{align*}
        - \int_{\TT^3} v^0 \Gamma_{00}^i  -\frac{2}{\tau}\int_{\TT^3} \frac{h^{0i}}{v^0} =& \int_{\TT^3} v^0g_{00}\pa_\tau h^{0i}-\frac{2}{\tau}\int_{\TT^3} \frac{h^{0i}}{v^0} + \int_{\TT^3} v^0\big[g_{0a}\pa_\tau h^{ia} + \frac{1}{2} g^{i\lambda}\pa_\lambda g_{00}\big]\\
        =&   -\pa_\tau h_\av^{0i} -\frac{2}{\tau}h_\av^{0i}  +\int_{\TT^3} (v^0 g_{00} + 1)\pa_\tau h^{0i}-\frac{2}{\tau}\int_{\TT^3} \big((v^0)^{-1}-1\big)h^{0i}\\
        &+\int_{\TT^3} v^0\big[g_{0a}\pa_\tau h^{ia} + \frac{1}{2}g^{i\lambda}\pa_\lambda g_{00}\big],
    \end{align*}
    which along with \eqref{sec:avg-ests.eq:avg-fluid-proof-2} implies the identity \eqref{sec:avg-ests.eq:avg-fluid-2}. This completes the proof.
\end{proof}
\subsection{Estimates for the ODE error terms}
\begin{lemma}\label{sec:avg-ests.lem:errors}
    Assume the bootstrap assumptions \eqref{sec:stability.eq:boot} hold. Then the error terms $\mc{H}_h^{\mu\nu}$, $\mc{H}_{\rho}$, $\mc{H}_v^i$ obey the bounds
    \begin{subequations}\label{sec:avg-ests.eq:errors}
        \begin{align}
            |\mc{H}_h^{00}| &\lesa \frac{1}{\tau^{3+\frac{3}{n}-\delta}}\,, \\
            \sum_{i=1}^3 |\mc{H}_h^{0i}| &\lesa \frac{1}{\tau^{2+\frac{6}{n}}}\,, \\
            \sum_{i,j=0}^3 |\mc{H}_h^{ij}| &\lesa \frac{1}{\tau^{2+\frac{6}{n}}}\,, \\
            |\mc{H}_\rho| &\lesa \frac{\ve}{\tau}\big|(\rho_\av-12)+12h_\av^{00}\big| + \ve\sum_{i,j=1}^3|\pa_\tau h_\av^{ij}|+ \frac{1}{\tau^{1+\frac{6}{n}}}\,, \\
            \sum_{i=1}^3 |\mc{H}_v^i|  &\leq\frac{\ve}{\tau}\sum_{i=1}^3 |h_\av^{0i}| +  \frac{1}{\tau^{2+\frac{6}{n}}}\,.
        \end{align}
    \end{subequations}
\end{lemma}

\begin{remark}
    At various points in the following proof, we rewrite a product in terms of its average and zero-average components, which implies the identity
    \begin{equation}
        \fint_{\TT^3} f\cdot g \,\textrm{d}^3 x = f_\av g_\av +  \fint_{\TT^3} (f-f_\av)(g-g_\av)\textrm{d}^3 x.
    \end{equation}
    Then, we apply the Poincar\'e inequality to estimate the final term. A useful consequence of this identity is that if one of the factors have zero average (for example, if one of the factors is a spatial derivative), then we can bound $\big|\fint_{\TT^3} f\cdot g \big|\lesa \|\pa_x f\|_{L^2}\|\pa_x g\|_{L^2}$, rendering such terms as quadratic errors.
\end{remark}
\begin{proof}
    This is a straightforward computation in which we freely use the Cauchy-Schwarz, Sobolev, and Poincar\'e inequalities. We bound the various integrals in $\mc{H}_{h}^{00}$ like
    \begin{align*}
        \Big|\int_{\TT^3} g_{\alpha\beta}h^{0a}\pa_a h^{\alpha\beta}\Big| &\lesa \sum_{\alpha,\beta=0}^3\|h^{0a}\|_{L^2} \|\pa_a h^{\alpha\beta}\|_{L^2} \lesa \frac{1}{\tau^{3+\frac{3}{n}-\delta}}\,,\\
        \Big|\int_{\TT^3} h_{0a}g^{00} \pa_\tau h^{0a}\Big| &\lesa \|h_{0a}\|_{L^2} \|\pa_\tau h^{0a}\|_{L^2} \lesa\frac{1}{\tau^{3+\frac{9}{2n}-\delta}}\,,\\
         \Big|\int_{\TT^3} h_{0a}h^{0a}\pa_\tau h^{00}\Big| &\lesa \|h_{0a}h^{0a}\|_{L^2} \|\pa_\tau h^{00}\|_{L^2} \lesa \frac{1}{\tau^{3+\frac{9}{n}}}\,,
    \end{align*}
    and
    \begin{multline*}
        \Big|\int_{\TT^3} \big[g_{ab}g^{00}-(g_{ab}g^{00})_\av\big]\pa_\tau (h^{ab}-h_\av^{ab})\Big|\\\lesa \big\|g_{ab} - (g_{ab})_\av\big\|_{L^2}\big\|\pa_\tau (h^{ab}-h_\av^{ab})\big\|_{L^2}
        \lesa \|\pa_x g_{ab}\|_{L^2} \|\pa_\tau \pa_x h^{ab}\|_{L^2} \lesa \frac{1}{\tau^{4-2\delta}}\,.
    \end{multline*}
    The slowest decay rate comes from the first of these four bounds, since we have assumed $\delta < 1-\frac{3}{n}$, which gives the stated estimate for $\mc{H}_h^{00}$. For $\mc{H}_h^{0i}$, we have
    \begin{align*}
        \Big|\int_{\TT^3} g_{\alpha\beta}h^{0i}\pa_\tau h^{\alpha\beta}\Big| &\lesa  \sum_{\alpha,\beta=0}^3\|h^{0i}\|_{L^2} \|\pa_\tau h^{\alpha\beta}\|_{L^2} \lesa \frac{1}{\tau^{2+\frac{6}{n}}}\,,\\
        \Big|\int_{\TT^3} g_{\alpha\beta}g^{ia}\pa_a h^{\alpha\beta}\Big| &\lesa \big|(g_{\alpha\beta} g^{ia})_\av\big|\cdot\underbrace{\Big|\int_{\TT^3} \pa_a h^{\alpha\beta}\big|}_{=0} + \Big| \int_{\TT^3} \big[g_{\alpha\beta}g^{ia} - (g_{\alpha\beta}g^{ia})_\av\big]\pa_a h^{\alpha\beta}\Big|\\
        &\lesa \big\|\pa_x (g_{\alpha\beta}g^{ia})\big\|_{L^2}\|\pa_a h^{\alpha\beta}\|_{L^2} \lesa \frac{1}{\tau^{4-2\delta}}\,,
    \end{align*}
    which gives the bound for $\mc{H}_h^{0i}$.
    
    Next, we estimate $\mc{H}_h^{ij}$. Integrating by parts, we find that
    \begin{multline*}
    \Big|\int_{\TT^3} (g^{00})^{-1}\Big(2h^{0a}\pa_a\pa_\tau h^{ij} + g^{ab}\pa_a\pa_b h^{ij}\Big) \Big| \\
    = \Big|\int_{\TT^3} 2\pa_a \big((g^{00})^{-1}h^{0a}\big)\pa_\tau h^{ij} + \pa_a \big((g^{00})^{-1}g^{ab}\big)\pa_b h^{ij}\Big|\\\lesa \big\|\pa_a \big((g^{00})^{-1}h^{0a}\big)\big\|_{L^2}\|\pa_\tau h^{ij}\|_{L^2} + \big\|\pa_a \big((g^{00})^{-1}g^{ab}\big)\|_{L^2}\|\pa_b h^{ij}\|_{L^2} \lesa \frac{1}{\tau^{3+\frac{9}{2n}-\delta}}\,.
    \end{multline*}
    The linear terms that comprise $\mc{L}_h^{ij}$ become error terms after integrating over the torus,
    \begin{multline*}
        \Big|\int_{\TT^3} (g^{00})^{-1}\mc{L}_h^{ij}\Big| \lesa \frac{1}{\tau}\Big|\int_{\TT^3} (g^{00})^{-1}h^{0(i}\pa_\tau h^{j)0} \Big| + \frac{1}{\tau}\Big|\int_{\TT^3} (g^{00})^{-1}g^{a(i}\pa_a h^{j)0}\Big|\\
        + \frac{1}{\tau^2}\Big|\int(g^{00})^{-1} g^{ij}\big[(h^{00}-h_\av^{00}) + (\rho - \rho_\av)\big]\Big|\,,
    \end{multline*}
    where for each term on the right hand side, we have
    \begin{align*}
        \frac{1}{\tau}\Big|\int_{\TT^3} h^{0(i}\pa_\tau h^{j)0} \Big| &\lesa \frac{1}{\tau}\|h^{0(i}\|_{L^2}\|\pa_\tau h^{j)0}\|_{L^2} \lesa \frac{1}{\tau^{4+\frac{9}{2n}-\delta}}\,,\\
        \frac{1}{\tau}\Big|\int_{\TT^3} (g^{00})^{-1}g^{a(i}\pa_a h^{j)0}\Big| & = \frac{1}{\tau}\Big|\int_{\TT^3} \pa_a \big((g^{00})^{-1}g^{a(i}\big)h^{j)0}\Big| \lesa \frac{1}{\tau}\big\|\pa_x h\big\|_{L^2}\|h^{j)0}\|_{L^2} \lesa \frac{1}{\tau^{4+\frac{3}{n}-\delta}}\,,
    \end{align*}
    and
    \begin{multline*}
        \frac{1}{\tau^2}\Big|\int(g^{00})^{-1} g^{ij}\big[(h^{00}-h_\av^{00}) + (\rho - \rho_\av)\big]\Big| \\\lesa \frac{1}{\tau^2}\Big|\big((g^{00})^{-1}g^{ij}\big)_\av\Big|\cdot \Big|\underbrace{\int_{\TT^3} \big[(h^{00}-h_\av^{00}) + (\rho - \rho_\av)\big]}_{=0}\Big|\\
        +\frac{1}{\tau^2}\Big|\int_{\TT^3} \Big[(g^{00})^{-1} g^{ij}-\big((g^{00})^{-1}g^{ij}\big)_\av\Big]\cdot\big[(h^{00}-h_\av^{00}) + (\rho - \rho_\av)\big]\Big|\\
        \lesa \frac{1}{\tau^2}\Big\|\pa_x\big[(g^{00})^{-1} g^{ij}\big]\Big\|_{L^2}\Big(\|\pa_x h^{00}\|_{L^2} + \|\pa_x \rho\|_{L^2}\Big)\lesa \frac{1}{\tau^{5-\frac{3}{2n}-2\delta}}.
    \end{multline*}
    Finally, the error term $\mc{G}_h^{ij}$ satisfies the bound
    \begin{equation*}
        \Big|\int_{\TT^3} (g^{00})^{-1}\mc{G}_h^{ij}\big| \lesa \frac{1}{\tau^{2+\frac{6}{n}}},
    \end{equation*}
    where the precise decay rate arises from the pressure term $\tau^{-(2+\frac{6}{n})}g^{ij} p$ in $\mc{G}_h^{ij}$. The slowest decay rate is from this final bound, and so these estimates in combination give the stated estimate for $\mc{H}_h^{ij}$. 
    
    The decay estimates for the $\mc{H}_\rho$, $\mc{H}_v^i$ follow by the same arguments, although we point out that the additional terms on the RHS of these equations arise from the estimates 
    \begin{gather*}
        \Big|\int_{\TT^3} \Big(\frac{\chi}{g^{00}} + 12\Big)\big(\pa_\lambda h^{0\lambda} + \Gamma^0\big)\Big| \lesa \frac{\ve}{\tau}\big|(\rho_\av-12)+12h_\av^{00}\big|+ \ve\sum_{i,j=1}^3|\pa_\tau h_\av^{ij}| + \frac{1}{\tau^{2+\frac{3}{2n}-\delta}}\,, \\
        \frac{1}{\tau}\Big|\int_{\TT^3} \big((v^0)^{-1}-1\big)h^{0i}\Big| + \Big|\int_{\TT^3} \big(v^0 g_{00} + 1\big)\pa_\tau h^{0i}\big| \lesa \frac{\ve}{\tau}|h_\av^{0i}| + \frac{1}{\tau^{2+\frac{9}{2n}-\delta}}\,,
    \end{gather*}
    where we used the wave gauge conditions to bound 
    \begin{align*}
        |\Gamma_\av^0| &\lesa \frac{1}{\tau}\big|(\rho_\av-12) + 12h_\av^{00} \big| + \frac{1}{\tau^{3+\frac{3}{n}-\delta}}\,, \\
        |\pa_\tau h_\av^{00}| &\lesa \frac{1}{\tau}\big|(\rho_\av-12) + 12 h_\av^{00}\big| + \sum_{i,j=1}^3|\pa_\tau h_\av^{ij}| +\frac{1}{\tau^{3+\frac{3}{n}-\delta}},\\
        |\pa_\tau h_\av^{0i}| &\lesa \frac{1}{\tau}|h_\av^{0i}| + \frac{1}{\tau^{2+\frac{6}{n}}}\,.
    \end{align*}
    We note that these additional terms decay at a critical rate, and so we must include them explicitly in the estimates, and track their smallness carefully.
\end{proof}
\subsection{Improvement of the bootstrap assumption for the spatial averages}
We are now ready to improve on the bounds from the bootstrap assumptions \eqref{sec:stability.eq:boot-1}--\eqref{sec:stability.eq:boot-2}.
\begin{proposition}\label{sec:avg-ests.prop:avg-ests}
    There exists a constant $C>0$ such that the norms $S_{h;\av}$, $S_{(\rho,v);\av}$ defined in \eqref{sec:stability.eq:boot-1}--\eqref{sec:stability.eq:boot-2} obey the bound
    \begin{equation}\label{sec:avg-ests.eq:avg-ests}
        S_{h;\av}(\tau) + S_{(\rho,v);\av}(\tau) \leq C\ve^{\frac{3}{2}},
    \end{equation}
    which is strictly smaller than $\ve$ (provided $\ve$ is sufficiently small). We additionally have the bound
    \begin{equation}\label{sec:avg-ests.omega-ests}
        \tau^{\frac3n}|(\rho_\av - 12) + 12 h^{00}_\av| \leq C\ve^{\frac{3}{2}}.
    \end{equation}
\end{proposition}
\begin{proof}
    Combining Lemmas~\ref{sec:avg-ests.lem:avg-met} and~\ref{sec:avg-ests.lem:errors}, we find that the spatial averages $(h_\av^{\mu\nu},\rho_\av,v_\av)$ satisfy
    \begin{subequations}
        \begin{align}
            |\pa_\tau h_\av^{00}| &\lesa \frac{1}{\tau}\Big|(\rho_\av-12) + 12h_\av^{00}\Big| + \sum_{i,j=1}^3 |\pa_\tau h_\av^{ij}|+ \frac{1}{\tau^{3+\frac{3}{n}-\delta}}\,, \label{sec:avg-ests.eq:avg-ests-proof-1}\\
		      \Big|\pa_\tau h_\av^{0i} + \frac{\kappa}{\tau} h_\av^{0i} \Big| &\lesa \frac{1}{\tau^{2+\frac{6}{n}}}\,, \label{sec:avg-ests.eq:avg-ests-proof-2}\\
		      \Big|\pa_\tau^2 h_\av^{ij} + \frac{4}{\tau} \pa_\tau h_\av^{ij}\Big| &\lesa \frac{1}{\tau^{2+\frac{6}{n}}}\,. \label{sec:avg-ests.eq:avg-ests-proof-3}
        \end{align}
        Similarly from Lemma~\ref{sec:avg-sys.lem:avg-fluid} and~\ref{sec:avg-ests.lem:errors} we find
        \begin{multline}\label{sec:avg-ests.eq:avg-ests-proof-4}
            \Big|\pa_\tau \Big((\rho_\av-12) + 12h_\av^{00}\Big) + \frac{3}{\tau}\Big((\rho_\av-12) + 12 h_\av^{00}\Big)\Big| \\
            \lesa \frac{\ve}{\tau}\big|(\rho_\av-12)+12h_\av^{00}\big| + \ve\sum_{i,j=1}^3|\pa_\tau h_\av^{ij}|+ \frac{1}{\tau^{1+\frac{6}{n}}}\,,
        \end{multline}
        \begin{equation}\label{sec:avg-ests.eq:avg-ests-proof-5}
		      \Big|\pa_\tau \big(v^i_\av + h_\av^{0i}\big) + \frac{2}{\tau} \big(v_\av^i + h_\av^{0i}\big)\Big| \lesa \frac{\ve}{\tau}\sum_{i=1}^3 |h_\av^{0i}| +  \frac{1}{\tau^{2+\frac{6}{n}}}\,.
        \end{equation}
    \end{subequations}
    Estimates can be derived from these ODE in a straightforward manner. The equations \eqref{sec:avg-ests.eq:avg-ests-proof-2}--\eqref{sec:avg-ests.eq:avg-ests-proof-5} all contain damping terms whose coefficients are all sufficiently large to account for the expected decay rates of the various quantities. More specifically, we have
    \begin{equation*}
        \kappa > 1+\frac{3}{n}\,,\qquad 4 > 1+\frac{3}{n}\,,\qquad 3 > \frac{3}{n},\qquad 2 > 1+\frac{3}{n}\,.
    \end{equation*}
    Define the rescaled averaged quantities
    \begin{gather*}
        \mu^i:= \tau^{1+\frac{3}{n}}h_\av^{0i}\,,\qquad \gamma^{ij}:= \tau^{1+\frac{3}{n}}\pa_\tau h_\av^{ij}\,,\\
        \omega:= \tau^{\frac{3}{n}}\big((\rho_\av-12) + 12h_\av^{00}\big)\,,\qquad \sigma^i:= \tau^{1+\frac{3}{n}}v^i\,.
    \end{gather*}
    The ODEs \eqref{sec:avg-ests.eq:avg-ests-proof-2}--\eqref{sec:avg-ests.eq:avg-ests-proof-5} therefore imply
    \begin{align*}
        \frac{1}{2}\sum_{i=1}^3\Big(\pa_\tau (\mu^i)^2 + \frac{B}{\tau}(\mu^i)^2\Big) &\leq \frac{C}{\tau^{1+\frac{3}{n}}}\sum_{i=1}^3|\mu^i|\,,\\
        \frac{1}{2}\sum_{i,j=1}^3\Big(\pa_\tau (\gamma^{ij})^2 + \frac{B}{\tau}(\gamma^{ij})^2\Big) &\leq \frac{C}{\tau^{1+\frac{3}{n}}}\sum_{i,j=1}^3|\gamma^{ij}|\,,\\
        \frac{1}{2}\big(\pa_\tau \omega^2 + \frac{B}{\tau}\omega^2\big) &\leq C\Big(\frac{\ve}{\tau}|\omega| + \frac{\ve}{\tau}\sum_{i,j=1}^3|\gamma^{ij}|+\frac{1}{\tau^{1+\frac{3}{n}}}\Big)|\omega|\,,\\
        \frac{1}{2}\sum_{i=1}^3\Big(\pa_\tau (\sigma^i)^2 + \frac{B}{\tau}(\sigma^i)^2\Big) &\leq C\Big(\frac{\ve}{\tau}\sum_{i=1}^3 |\mu^i| + \frac{1}{\tau^{1+\frac{3}{n}}}\Big)\sum_{i=1}^3|\sigma^i|\,,\\
    \end{align*}
    where $B$, $C$ are positive constants, whose sizes depend only on parameters like $N$, $n$, and crucially not on $\ve$. 
    Combining these estimates, we find that the energy-type quantity
    \begin{equation*}
        \alpha(\tau) := \frac{1}{2}\Big(\sum_{i=1}^3(\mu^i)^2 + \sum_{i,j=1}^3 (\gamma^{ij})^2 + \omega^2 + \sum_{i=1}^3 (\sigma^i)^2\Big)\,,
    \end{equation*}
    obeys the inequality
    \begin{multline*}
        \pa_\tau \alpha(\tau) + \frac{B}{\tau}\alpha(\tau) \leq \frac{C}{\tau} \sqrt{\alpha(\tau)}\Big(\sum_{i=1}^3 |\mu^i| + \sum_{i,j=1}^3|\gamma^{ij}| + |\omega| \Big) + \frac{C}{\tau^{1+\frac{3}{n}}}\sqrt{\alpha(\tau)}
        \leq \frac{C\ve}{\tau}\alpha(\tau) + \frac{C}{\tau^{1+\frac{3}{n}}}\,.
    \end{multline*}
    Note we are now using $C$ as a running constant which does not depend on $\ve$. Shrinking $\ve$ if necessary, we can absorb the first term on the right hand side into the left hand side. Moreover, we may use the fact that the initial time satisfies $\tau_0 \geq \ve^{-4n}$ to bound  $\tau^{-3/4n} \leq \ve^3$ which we apply to the error term to bound
    \begin{equation}\label{sec:avg-ests.eq:avg-ests-proof-6}
        \pa_\tau \alpha(\tau) \leq \frac{C\ve^3}{\tau^{1+\frac{9}{4n}}}\,.
    \end{equation}
    The right hand side is integrable in time, hence we integrate \eqref{sec:avg-ests.eq:avg-ests-proof-6} on the interval $[\tau_0,\tau]$ to obtain
    \begin{equation*}
        \alpha(\tau) \leq \alpha(\tau_0) + C\ve^3\,,
    \end{equation*}
    The initial bound \eqref{sec:stability.eq:initial-bound} implies that $\alpha(\tau_0) \leq C\ve^2$, thus we have the improved bound
    \begin{equation}
        \tau^{1+\frac{3}{n}}\sum_{i=1}^3|h_\av^{0i}(\tau)| + \tau^{1+\frac{3}{n}}\sum_{i,j=1}^3|\pa_\tau h_\av^{ij}(\tau)| + \tau^{1+\frac{3}{n}}\sum_{i=1}^3|v_\av^i(\tau)| \leq C\ve^{\frac{3}{2}}\,,
    \end{equation}
    moreover we find that the quantity $\omega$ decays like $|\omega| \leq C\ve^{3/2} \tau^{-\frac{3}{n}}$, implying \eqref{sec:avg-ests.omega-ests}. Since $\pa_\tau h_\av^{ij}$ decay at an integrable rate in time, one can simply integrate again to obtain improved estimates for $h_\av^{ij}$.
    
    The remaining equation \eqref{sec:avg-ests.eq:avg-ests-proof-1} lacks a damping term, so instead we simply have
    \begin{equation*}
        \frac{1}{2}\pa_\tau (h_\av^{00})^2 \leq \frac{C}{\tau^{1+\frac{3}{n}}}|h_\av^{00}|\sqrt{\alpha} + \frac{C\ve^3}{\tau^{3+\frac{3}{n}-\delta}} \leq \frac{C}{\tau^{1+\frac{3}{n}}} \leq \frac{C\ve^3}{\tau^{1+\frac{9}{4n}}}\,.
    \end{equation*}
    Integrating this in time and using the initial bound gives the improved estimate for $h_\av^{00}$. By the decay of $\omega$, this also implies an improved bound on $\rho_\av-12$. Finally, one obtains improved bounds for the quantities $\pa_\tau h_\av^{00}$, $\pa_\tau h_\av^{0i}$ from the wave gauge equations and the previous estimates. This completes the proof.
\end{proof}
\begin{remark}\label{sec:avg-ests.rmk:derivests}
    We can use Proposition~\ref{sec:avg-ests.prop:avg-ests} to prove decay estimates for other (higher) derivatives of the various quantities. The averaged Euler equations \eqref{sec:avg-ests.eq:avg-fluid} together with \eqref{sec:avg-ests.omega-ests} from Proposition~\ref{sec:avg-ests.prop:avg-ests} and Lemma \ref{sec:avg-ests.lem:errors} imply the bounds 
    \begin{subequations}\label{sec:avg-ests.eq:second-derivs}
    \begin{align}
        |\pa_\tau \rho_\av(\tau)| &\lesa \frac{\ve}{\tau^{1+\frac{3}{n}}}\,,\\
        \sum_{i=1}^3\big|\pa_\tau v_\av^i(\tau)\big| &\lesa \frac{\ve}{\tau^{2+\frac{3}{n}}}\,,
    \end{align}
    The equation \eqref{sec:avg-ests.eq:avg-met-3} for the $h_\av^{ij}$ implies the bound
    \begin{equation}
        \sum_{i,j=1}^3 \big|\pa_\tau^2 h_\av^{ij}(\tau)\big| \lesa \frac{\ve}{\tau^{2+\frac{3}{n}}}\,.
    \end{equation}
    One can then differentiate-in-time the spatially averaged Euler equations \eqref{sec:avg-ests.eq:avg-fluid} and the wave gauge equations \eqref{sec:avg-ests.eq:avg-met-1}-- \eqref{sec:avg-ests.eq:avg-met-2} to obtain the bounds
    \begin{equation}
        \sum_{\mu=0}^3 \big|\pa_\tau^2 h_\av^{0\mu}(\tau)\big| + |\pa_\tau^2 \rho_\av(\tau)| + \sum_{i=1}^3 |\pa_\tau^2 v_\av^i(\tau)| \lesa \frac{\ve}{\tau^{2+\frac{3}{n}}}\,.
    \end{equation}
    \end{subequations}
\end{remark}
\section{Energy estimates for the metric perturbations}\label{sec:energy-met}
In this section, we establish energy estimates for higher derivatives of the metric components. The main estimates are contained in Proposition~\ref{sec:energy-met.prop:commuted-energy} and Proposition~\ref{sec:energy-met.prop:higher-energy}.

In Proposition~\ref{sec:energy-met.prop:commuted-energy}, we prove energy estimates for a class of damped wave equations that are satisfied by the metric components. To prove these estimates, we construct a family of wave energies with two important features. The first is the inclusion of a lower-order correction term that leads to coercive damping terms in the energy estimates. Recalling our discussion from Section~\ref{intro:bded-sec}, these damping terms will enable us to establish decay of the metric perturbations. The second feature of the energies is the application of the material derivative $\pav = v^\mu \pa_\mu$ as a vectorfield commutator. In view of Proposition~\ref{sec:stability.prop:fluid-time-deriv}, material derivatives of the fluid components exhibit improved decay in comparison to their spatial derivatives, and so commuting the metric evolution equations \eqref{sec:gauged-eqs.eq:wave} with $\pav$ leads to improved decay of the source terms on the right hand side. 

In Proposition~\ref{sec:energy-met.prop:higher-energy}, we apply the energy estimates of Proposition~\ref{sec:energy-met.prop:commuted-energy} to the spatial derivatives of the metric components. The gauge-fixed Einstein equations \eqref{sec:gauged-eqs.eq:wave} have certain source terms, contained in the $\mc{L}_{h}^{\mu\nu}$ which decay at a critical rate. These must be tracked carefully through the estimates. Ultimately, we will control such terms by exploiting some remaining gauge freedom and fixing the parameter $\kappa$ which appears in the gauge condition \eqref{sec:gauged-eqs.eq:gauge-choice} to be some sufficiently large number. We  do this later in Section~\ref{sec:boot-imp}.

\subsection{Base wave energy}
\begin{definition}[Base wave energy and damped wave operator]\label{def:base-wave-energy}
   For a given function $u$ that is sufficiently regular, we introduce the base \textit{wave energy}
\begin{equation}
    E_{w}[u](\tau) = \frac{1}{2}\int_{\TT^3}|g^{00}|(\pa_\tau u)^2 + g^{ij}\pa_i u \pa_j u\,,
\end{equation} 
    and recall the damped wave operator 
    \begin{equation}\label{sec:energy-met.eq:base-energy-2}
    P_\eta := \wh{\square}_g + \frac{2\eta}{\tau}g^{00}\pa_\tau\,.
    \end{equation}
\end{definition}

By the bootstrap assumption, $E_{w}$ is equivalent to the norm
\begin{equation*}
    \frac{1}{C}\|\pa u\|_{L^2}^2 \leq E_{w}[u] \leq C\|\pa u\|_{L^2}^2\,.
\end{equation*}
We prove a bound for the time-evolution of this energy.
\begin{lemma}\label{sec:energy-met.lem:base-energy}
	The wave energy $E_{w}$ satisfies the differential inequality
	\begin{equation}\label{sec:energy-met.eq:base-energy-1}
		\pa_\tau E_{w}[u](\tau) + \frac{2a}{\tau}\int_{\TT^3} |g^{00}|(\pa_\tau u)^2 \lesa \|\pa_\tau u\|_{L^2}\|P_\eta u\|_{L^2} + \frac{1}{\tau^{1+\frac{3}{n}}} E_{w}[u]\,.
	\end{equation}
\end{lemma}
\begin{proof}
	We differentiate the energy $E_{w}[u]$ in time. For the first term, we have
	\begin{equation*}
		\frac{1}{2}\pa_\tau \Big(\int_{\TT^3} |g^{00}|(\pa_\tau u)^2\Big) =-\int_{\TT^3} g^{00}\pa_\tau u \pa_\tau^2 u -\frac{1}{2}\int_{\TT^3} \pa_\tau g^{00}(\pa_\tau u)^2\,.
	\end{equation*}
	For the spatial derivatives of $u$, we integrate by parts to obtain
	\begin{multline*}
		\frac{1}{2}\pa_\tau \Big(\int_{\TT^3} g^{ab}\pa_a u \pa_b u\Big) = \int_{\TT^3}g^{ab}\pa_a u \pa_b \pa_\tau u + \frac{1}{2}\int_{\TT^3} \pa_\tau g^{ab}\pa_a u \pa_b u\\
		= -\int_{\TT^3}g^{ab}\pa_\tau u \pa_a \pa_b u - \int_{\TT^3}\pa_a g^{ab}\pa_\tau u \pa_b u + \frac{1}{2}\int_{\TT^3}\pa_\tau g^{ab}\pa_a u \pa_b u\,.
	\end{multline*}
	It also follows from integration by parts that
	\begin{equation*}
		0 = \int_{\TT^3} \pa_a\big(g^{0a}(\pa_\tau u)^2\big)= 2\int_{\TT^3}g^{0a} \pa_\tau u \pa_\tau\pa_a u +\int_{\TT^3} \pa_a g^{0a}(\pa_\tau u)^2\,.
	\end{equation*}
	Combining these identities, we find that $E_{w}[u]$ satisfies
    \begin{equation}\label{sec:energy-met.eq:base-energy-proof-1}
        \pa_\tau E_{w}[u] =- \int_{\TT^3}\pa_\tau u \wh{\square}_g u - \frac{1}{2}\int_{\TT^3}\Big[\big(\pa_\tau g^{00}  +  2\pa_a g^{0a}\big) (\pa_\tau u)^2 +2\pa_a g^{ab}\pa_\tau u\pa_b u -\pa_\tau g^{ab}\pa_a u \pa_b u\Big]\,.
    \end{equation}
    The final integral consists solely of error terms, which can be bounded like
	\begin{equation}\label{sec:energy-met.eq:base-energy-proof-2}
		\Big|\int_{\TT^3}\Big[\big(\pa_\tau g^{00}  +  2\pa_a g^{0a}\big) (\pa_\tau u)^2 +2\pa_a g^{ab}\pa_\tau u\pa_b u -\pa_\tau g^{ab}\pa_a u \pa_b u\Big]\Big| \lesa \|\pa h\|_{L^\infty}E_{w}[u] \lesa \frac{1}{\tau^{1+\frac{3}{n}}}E_{w}[u]\,.
	\end{equation}
    Finally, we replace the wave operator $\wh{\square}_g$ with the damped wave operator $P_\eta$, and gain an additional damping term in the energy estimate:
    \begin{equation}\label{sec:energy-met.eq:base-energy-proof-3}
        - \int_{\TT^3}\pa_\tau u \wh{\square}_g u = -\int_{\TT^3} \pa_\tau u P_\eta u - \frac{2\eta}{\tau}\int_{\TT^3} |g^{00}|(\pa_\tau u)^2 \leq \|\pa_\tau u\|_{L^2} \|P_\eta u\|_{L^2} -  \frac{2\eta}{\tau}\int_{\TT^3} |g^{00}|(\pa_\tau u)^2\,,
    \end{equation}
    where the final line is a simple application of the Cauchy-Schwarz inequality. Combining the identity \eqref{sec:energy-met.eq:base-energy-proof-1} with the bounds \eqref{sec:energy-met.eq:base-energy-proof-2}, \eqref{sec:energy-met.eq:base-energy-proof-3} and rearranging, we get \eqref{sec:energy-met.eq:base-energy-1}.
\end{proof}
\subsection{The correction mechanism for energy decay}
The base wave energy does not, on its own, satisfy a differential inequality that leads to decay. The estimate from Lemma~\ref{sec:energy-met.lem:base-energy} does contain the bulk term
\begin{equation*}
    \frac{2\eta}{\tau}\int_{\TT^3} |g^{00}|(\pa_\tau u)^2\,,
\end{equation*}
which is inherited from the damping terms present in the wave operator $P_\eta$. However, this term does not control the entire energy, hence we cannot conclude that the energy $E_{w}$ decays for solutions to equations of the form $P_\eta u = F$.

To prove some form of energy decay, we modify the energy $E_{w}$ in two ways. The first is through the introduction of a lower-order ``corrective'' energy functional, which takes the form
\begin{equation}\label{sec:energy-met.eq:correction-term}
	\frac{1}{\tau}\int_{\TT^3} |g^{00}| u \pa_\tau u\,.
\end{equation}
We will introduce the second modification in the following subsection, but first we prove a bound on the time-evolution of the correction term \eqref{sec:energy-met.eq:correction-term}.
\begin{lemma}\label{sec:energy-met.lem:correction}
	For any function $u$, the functional \eqref{sec:energy-met.eq:correction-term} satisfies
	\begin{multline}\label{sec:energy-met.eq:correction}
		\pa_\tau \Big(\frac{1}{\tau}\int_{\TT^3} |g^{00}|u\pa_\tau u\Big) + \frac{1}{\tau}\int_{\TT^3} g^{ab}\pa_a u\pa_b u - \frac{1}{\tau}\int_{\TT^3} |g^{00}|(\pa_\tau u)^2 \\
		\lesa \frac{1}{\tau}\|u\|_{L^2}\|P_\eta u\|_{L^2} + \frac{1}{\tau^{2}}\big(\|u\|_{L^2}^2 + E_{w}[u]\big)\,.
	\end{multline}
\end{lemma}
\begin{proof}
	We compute
	\begin{equation*}
		\pa_\tau \Big(\frac{1}{\tau}\int_{\TT^3} |g^{00}|u\pa_\tau u\Big) =\frac{1}{\tau}\int_{\TT^3} |g^{00}|u\pa_\tau^2 u + \frac{1}{\tau}\int_{\TT^3} |g^{00}|(\pa_\tau u)^2\\
		+ \frac{1}{\tau^2}\int_{\TT^3} g^{00}u\pa_\tau u - \frac{1}{\tau}\int_{\TT^3} \pa_\tau g^{00} u\pa_\tau u\,.
	\end{equation*}
	For the first integral on the RHS, we substitute the wave operator $P_\eta$ and integrate by parts, giving
	\begin{align*}
		\frac{1}{\tau}\int_{\TT^3} |g^{00}|u\pa_\tau^2 u =& -\frac{1}{\tau}\int_{\TT^3}uP_\eta u+ \frac{1}{\tau}\int_{\TT^3} g^{ab}u\pa_a\pa_b u + \frac{2\eta}{\tau^2}\int_{\TT^3}g^{00}u\pa_\tau u + \frac2\tau\int_{\TT^3} g^{0a} u\pa_a\pa_\tau u\\
		=&-\frac{1}{\tau}\int_{\TT^3}uP_\eta u - \frac{1}{\tau}\int_{\TT^3} g^{ab}\pa_a u\pa_b u\\
        &+ \frac{2\eta}{\tau^2}\int_{\TT^3}g^{00} u\pa_\tau u -\frac{1}{\tau}\int_{\TT^3} \Big[\pa_a g^{ab}u \pa_b u 
        +2 g^{0a} \pa_a u \pa_\tau u +2 \pa_a g^{0a} u \pa_\tau u\Big]\,.
	\end{align*}
	Thus the correction term satisfies the identity
	\begin{multline*}
		\pa_\tau \Big(\frac{1}{\tau}\int_{\TT^3} |g^{00}|u\pa_\tau u\Big) + \frac{1}{\tau}\int_{\TT^3} g^{ab}\pa_a u\pa_b u - \frac{1}{\tau}\int_{\TT^3} |g^{00}|(\pa_\tau u)^2\\
		=-\frac{1}{\tau}\int_{\TT^3}uP_\eta u + \frac{1+2\eta}{\tau^2} \int_{\TT^3}g^{00}u\pa_\tau u\\
        -\frac{1}{\tau}\int_{\TT^3} \Big[\pa_\tau g^{00} u\pa_\tau u + \pa_a g^{ab}u \pa_b u + 2g^{0a}\pa_a u \pa_\tau u + 2\pa_a g^{0a}u \pa_\tau u\Big]\,.
	\end{multline*}
	We bound all terms on the right hand side using a combination of Young's inequality and the Cauchy-Schwarz inequality:
	\begin{align*}
		\frac{1}{\tau}\Big|\int_{\TT^3}uP_\eta u\Big| &\leq \frac{1}{\tau}\|u\|_{L^2}\|P_\eta u\|_{L^2},\\
        \frac{1}{\tau^2} \Big|\int_{\TT^3}g^{00}u\pa_\tau u\Big| &\leq \frac{1}{\tau^2}\|u\|_{L^2} \|\pa_\tau u\|_{L^2} \lesa \frac{1}{\tau^2}\big(\|u\|_{L^2}^2 + E_{w}[u]\big)\,,
    \end{align*}
    and
    \begin{multline*}
		\frac{1}{\tau}\Big|\int_{\TT^3} \Big[\pa_\tau g^{00} u\pa_\tau u + \pa_a g^{ab}u \pa_b u + 2g^{0a}\pa_a u \pa_\tau u + 2\pa_a g^{0a}u \pa_\tau u\Big]\Big|\\
        \lesa \frac{1}{\tau}\|\pa h\|_{L^\infty}\|u\|_{L^2}\|\pa u\|_{L^2} + \frac{1}{\tau}\sum_{i=1}^3\|h^{0i}\|_{L^\infty}\|\pa u\|_{L^2}^2 \lesa \frac{1}{\tau^{2+\frac{3}{n}}}\big(\|u\|_{L^2}^2 + E_{w}[u]\big)\,.
	\end{multline*}
	The result now follows.
\end{proof}
\begin{remark}
    The energy $E_{w}[u]$ does not control the $L^2$-norm of $u$, hence the additional terms that appear in \eqref{sec:energy-met.eq:correction}. However, the energy estimates of this section will be applied to the metric perturbations $h^{\mu\nu}$ \emph{after} taking at least one spatial derivative. Spatial derivatives of functions trivially have vanishing spatial average, and so we will control the $L^2$ norms of such quantities in a straightforward manner with the Poincar\'e inequality.
\end{remark}
\subsection{Commutation with the material derivative}
In light of the previous two lemmas, for any constant $\eta > 0$ one can show that the corrected energy
\begin{equation*}
    E_{w}[u] + \frac{\eta}{\tau}\int_{\TT^3} |g^{00}|u\pa_\tau u\,,
\end{equation*}
satisfies an energy inequality with a bulk term that controls the entire energy. We make here an essential modification of such an energy -- we commute once with the material derivative $\pav := v^\mu \pa_\mu$. We do this to derive improved estimates for the fluid variables that appear as source terms in the wave equations for the metric components.

\begin{definition}[Corrected wave energy and $\mathbf{v}$-adapted Riemannian metric]\label{def:corrected-pav-wave-energy}
Using the base wave energy defined in Definition \ref{def:base-wave-energy}, we introduce the corrected wave energy
 \begin{equation}\label{eq:corrected_wave_pav}
    E_{w, \eta}^{\mathbf{v}}[u](\tau): = E_{w}[\pav u](\tau) + \frac{\eta}{\tau}\int_{\TT^3} |g^{00}|(\pav u)(\pa_\tau \pav u)\,.
\end{equation}   
As in \cite{HadSpe:deSitterduststab:15}, define also the inverse Riemannian metric
    \begin{equation*}
        H^{ij} = g^{ij} + \frac{1}{(v^0)^2}g^{00}v^iv^j - \frac{2}{v^0}g^{0(i}v^{j)}\,.
    \end{equation*}
\end{definition}

In order to use $\pav$ as a vectorfield commutator in a way that controls our desired bootstrap quantities, the following Lemma is critical.

\begin{lemma}[Elliptic estimate]\label{sec:energy-met.lem:elliptic}
    We have the identity
    \begin{equation}\label{sec:energy-met.eq:elliptic-id}
        H^{ab}\pa_a \pa_b u = \wh{\square}_g u - \frac{1}{v^0}g^{00}\pav \pa_\tau u + \Big(\frac{1}{(v^0)^2}g^{00}v^a - \frac{2}{v^0}g^{0a}\Big) \pav \pa_a u\,.
    \end{equation}
    
    Additionally,  provided $u$ has zero spatial average, the following estimate holds,
    \begin{equation}\label{sec:energy-met.eq:elliptic-est}
        \|u\|_{H^2} \lesa \|P_\eta u\|_{L^2}  + \|\pa (\pa_{\mathbf{v}} u)\|_{L^2}\,.
    \end{equation}
\end{lemma}
\begin{proof}
    The identity \eqref{sec:energy-met.eq:elliptic-id} follows from the computation
    \begin{align*}
        \wh{\square}_g u &= g^{ab}\pa_a \pa_b u + \underbrace{2g^{0a}\pa_\tau \pa_a u}_{(i)} + \underbrace{g^{00} \pa_\tau^2 u}_{(ii)}\,,\\
        - \frac{1}{v^0}g^{00}\pav \pa_\tau u&= -\underbrace{g^{00}\pa_\tau ^2 u}_{(ii)}-\underbrace{\frac{1}{v^0}g^{00}v^a \pa_\tau \pa_a u}_{(iii)}\,,\\
        \frac{1}{(v^0)^2}g^{00}v^a \pav \pa_a u &= \underbrace{\frac{1}{v^0} g^{00}v^a \pa_\tau \pa_a u}_{(iii)} + \frac{1}{(v^0)^2}g^{00}v^a v^b \pa_a \pa_b u\,,\\
        - \frac{2}{v^0}g^{0a}\pav \pa_a u &= -\underbrace{2g^{0a}\pa_\tau \pa_a u}_{(i)}-\frac{2}{v^0}g^{0(a} v^{b)}\pa_a \pa_b u\,.
    \end{align*}
    The pairs of terms labelled $(i)$, $(ii)$, $(iii)$ all cancel, and the remaining three terms combine to make $H^{ab}\pa_a\pa_b u$, giving \eqref{sec:energy-met.eq:elliptic-id}. 
    
    For the estimate \eqref{sec:energy-met.eq:elliptic-est}, we take the $L^2$ norm of $H^{ab}\pa_a \pa_b u$ over the torus and integrate by parts twice to obtain
    \begin{multline}\label{sec:energy-met.eq:elliptic-proof-1}
        \int_{\TT^3} (H^{ab} \pa_a \pa_b u)^2 = \int_{\TT^3} H^{ab}H^{cd}\pa_a \pa_c u\,\pa_b \pa_d u - \int_{\TT^3}\pa_a(H^{ab}H^{cd})\pa_b u\pa_c\pa_d u \\+ \int_{\TT^3}\pa_d (H^{ab}H^{cd})\pa_b u\pa_a\pa_c u\,.
    \end{multline}
    By the bootstrap assumptions, the metric $H$ satisfies $\sum_{i,j=1}^3\|H^{ij}-\delta^{ij}\|_{L^\infty} \lesa \ve$, and so the first term on the right hand side of \eqref{sec:energy-met.eq:elliptic-proof-1} controls all second-order spatial derivatives of $u$:
    \begin{equation*}
        \sum_{H^2}\|\pa_x^I u\|_{L^2}^2 \lesa \int_{\TT^3} H^{ab}H^{cd}\pa_a \pa_c u\,\pa_b \pa_d u\,.
    \end{equation*}
    We insert this bound into \eqref{sec:energy-met.eq:elliptic-proof-1}, rearrange, absorbing all decaying error terms into the left hand side, and apply the Poincar\'e inequality to obtain the elliptic estimate
    \begin{equation*}
        \|u\|_{H^2} \lesa \|H^{ab}\pa_a\pa_b u\|_{L^2}\,.
    \end{equation*}
    
    To prove the stated estimate, we combine this with the identity \eqref{sec:energy-met.eq:elliptic-id}, as well as the Poincar\'e inequality to get
    \begin{multline*}
    	\|u\|_{H^2}  \lesa \Big(\|\wh{\square}_g u\|_{L^2} + \|\pav \pa_\tau u\|_{L^2} + \frac{1}{\tau^{1+\frac{3}{2n}-\delta}}\|\pav\pa_x u\|_{L^2}\Big)\\
        \lesa \|P_\eta u\|_{L^2} + \Big(\frac{1}{\tau}\|\pa_\tau u\|_{L^2} + \|\pav \pa_\tau u\|_{L^2} + \frac{1}{\tau^{1+\frac{3}{2n}-\delta}}\|\pav\pa_x u\|_{L^2}\Big)\,.
    \end{multline*}
    We can bound all remaining terms in the brackets in using the identities 
    \begin{subequations}\label{sec:energy-met.eq:elliptic-proof-2}
        \begin{align}
            \pa_i \pa_\tau u &= \frac{1}{v^0}\Big[\pa_i (\pav u) - v^a\pa_a\pa_i u - \pa_i v^a \pa_a u  - \pa_i v^0 \pa_\tau u\Big]\,,\\
            \pa_\tau^2 u &= \frac{1}{v^0}\Big[\pa_\tau (\pav u) - v^a\pa_a\pa_\tau u - \pa_\tau v^a \pa_a u  - \pa_\tau v^0 \pa_\tau u\Big]\,,
        \end{align}
    \end{subequations}
    along with the Poincar\'e inequality to estimate $\|\pa_\tau u\|_{L^2}$, giving
    \begin{align*}
        \frac{1}{\tau}\|\pa_\tau u\|_{L^2} + \|\pav \pa_\tau u\|_{L^2} + \frac{1}{\tau^{1+\frac{3}{2n}-\delta}}\|\pav\pa_x u\|_{L^2} \lesa \|\pa(\pav u)\|_{L^2} + \frac{1}{\tau^{1+\frac{3}{2n}-\delta}}\|u\|_{H^2}\,.
    \end{align*}
    The result now follows.
\end{proof}

\begin{remark}\label{rem:time-derivs-elliptic}
    We can also prove estimates on norms of the coordinate time derivative and material derivative. Using the previous lemma, the identities \eqref{sec:energy-met.eq:elliptic-proof-2}, and the Poincar\'e inequality, $\pa_\tau u$ and $\pa_\tau^2 u$ satisfy the bounds
    \begin{equation}\label{sec:energy-met.eq:eq:elliptice-rmk-1}
        \|\pa_\tau u\|_{H^1} + \|\pa_\tau^2 u\|_{L^2} \lesa \|\pa (\pav u)\|_{L^2} + \frac{1}{\tau^{1+\frac{3}{2n}-\delta}}\|P_\eta u\|_{L^2}\,.
    \end{equation}
    This also gives the following $L^2$-norm bound of $\pav u$:
    \begin{equation}\label{sec:energy-met.eq:eq:elliptice-rmk-2}
        \|\pav u\|_{L^2} \lesa \|\pa (\pav u)\|_{L^2} + \frac{1}{\tau^{1+\frac{3}{2n}-\delta}}\|P_\eta u\|_{L^2}\,.
    \end{equation}
\end{remark}

Using the previous lemma, we prove the following energy estimate for $E_{w, \eta}^{\mathbf{v}}$:
\begin{proposition}[Energy estimate for the commuted metric energy]\label{sec:energy-met.prop:commuted-energy}
    Let $u$ be a function with vanishing spatial average. Then the modified wave energy $E_{w, \eta}^{\mathbf{v}}[u]$ satisfies the differential inequality
    \begin{multline}\label{sec:energy-met.eq:commuted-energy}
        \pa_\tau E_{w, \eta}^{\mathbf{v}}[u](\tau) + \frac{2\eta}{\tau}E_{w, \eta}^{\mathbf{v}}[u] \lesa \big(E_{w,\eta}^{\mathbf{v}}[u]\big)^{1/2}\big\|\pav(P_\eta u)\big\|_{L^2} + \frac{1}{\tau^{1+\frac{3}{2n}-\delta}}E_{w, \eta}^{\mathbf{v}}[u](\tau)\\
        +\frac{1}{\tau^{2-\frac{3}{2n}-\delta}}\big\|\pav (P_\eta u)\big\|_{L^2}  +\frac{1}{\tau^{2+\frac{9}{2n}-\delta}}\|P_\eta u\|_{L^2}^2\,.
    \end{multline}
\end{proposition}
\begin{proof}
    By Lemmas~\ref{sec:energy-met.lem:base-energy} and~\ref{sec:energy-met.lem:correction}, we have for any $c \in \RR$ that
    \begin{multline*}
        \pa_\tau E_{w}[\pav u](\tau) + \pa_\tau \Big(\frac{c}{\tau}\int_{\TT^3} |g^{00}| (\pav u) \pa_\tau \pav u \Big)\\+ \frac{2\eta-c}{\tau}\int_{\TT^3}|g^{00}|(\pa_\tau \pav u)^2  + \frac{c}{\tau}\int_{\TT^3} g^{ab}\pa_a \pav u \pa_b \pav u\\
        \lesa \Big(\frac{1}{\tau}\|\pav u\|_{L^2} + \|\pa \pav u\|_{L^2}\Big)\big\|P_\eta (\pav u)\big\|_{L^2} + \frac{1}{\tau^{1+\frac{3}{n}}}\Big(\|\pav u\|_{L^2}^2 + E_{w}[\pav u]\Big)\,.
    \end{multline*}
    We see that the bulk term is maximised by setting $c = \eta$, and from this we obtain the bound
    \begin{multline}\label{sec:energy-met.eq:commuted-energy-proof-1}
         \pa_\tau E_{w, \eta}^{\mathbf{v}}[u](\tau) + \frac{2\eta}{\tau} E_{w, \eta}^{\mathbf{v}}[u] \lesa \Big(\frac{1}{\tau}\|\pav u\|_{L^2} + \|\pa (\pav u)\|_{L^2}\Big)\|P_\eta (\pav u)\|_{L^2}\\+ \frac{1}{\tau^{1+\frac{3}{n}}}\Big(\|\pav u\|_{L^2}^2 + E_{w}[\pav u]\Big)\,.
    \end{multline}
    To reach the stated estimate, we must bound the quantities $\|P_\eta (\pav u)\|_{L^2}$, $\|\pav u\|_{L^2}$. We observe that the formula for the commutator $[P_\eta,\pav]$ is
    \begin{equation*}
        [P_\eta,\pav]u = 2g^{\alpha\beta}\pa_\alpha v^\mu \pa_\beta \pa_\mu u + \wh{\square}_g v^\mu \cdot \pa_\mu u + \frac{2\eta}{\tau}g^{00}\pa_\tau v^\mu \pa_\mu u -v^\mu \pa_\mu g^{\alpha\beta}\cdot\pa_\alpha\pa_\beta u + \frac{2\eta}{\tau^2}v^0 g^{00}\pa_\tau u\,.
    \end{equation*}
    We apply this identity, Lemma~\ref{sec:energy-met.lem:elliptic}, and the subsequent remark to estimate
    \begin{multline}\label{sec:energy-met.eq:commuted-energy-proof-2}
        \big\|P_\eta (\pav u)\big\|_{L^2} \lesa \big\|\pav (P_\eta u)\big\|_{L^2} \\+ \big(\|\pa v\|_{L^\infty} + \|\pa^2v\|_{L^\infty} + \|\pa h\|_{L^\infty}+ \tau^{-2}\big)\big(\|u\|_{H^2} + \|\pa_\tau u\|_{H^1} + \|\pa_\tau^2 u\|_{L^2}\big)\\
        \leq \big\|\pav (P_\eta u)\big\|_{L^2} + \frac{1}{\tau^{1+\frac{3}{2n}-\delta}}\Big(\|P_\eta u\|_{L^2} + \|\pa(\pav u)\|_{L^2}\Big)\,.
    \end{multline}
    Finally, we control $\|\pav u\|_{L^2}$ with the estimate \eqref{sec:energy-met.eq:eq:elliptice-rmk-2} following Lemma~\ref{sec:energy-met.lem:elliptic}, which gives
    \begin{equation}\label{sec:energy-met.eq:commuted-energy-proof-3}
        \|\pav u\|_{L^2} \leq \frac{1}{\tau^{1+\frac{3}{2n}-\delta}}\|P_\eta u\|_{L^2} + \|\pa(\pav u)\|_{L^2}\,.
    \end{equation}
    The bounds \eqref{sec:energy-met.eq:commuted-energy-proof-2} and \eqref{sec:energy-met.eq:commuted-energy-proof-3} together imply by Young's inequality that
    \begin{align*}
         \Big(\frac{1}{\tau}\|\pav u\|_{L^2} + \|\pa (\pav u)\|_{L^2}\Big)\|P_\eta (\pav u)\|_{L^2} \lesa&\, \big(E_{w, \eta}^{\mathbf{v}}[u]\big)^{1/2} \big\|\pav(P_\eta u)\big\|_{L^2} + \frac{1}{\tau^{1+\frac{3}{2n}-\delta}}E_{w, \eta}^{\mathbf{v}}[u]\\
         &+\frac{1}{\tau^{2-\frac{3}{2n}-\delta}}\big\|\pav (P_\eta u)\big\|_{L^2}^2  +\frac{1}{\tau^{2+\frac{9}{2n}-\delta}}\|P_\eta u\|_{L^2}^2\,,\\
         \frac{1}{\tau^{1+\frac{3}{n}}}\Big(\|\pav u\|_{L^2}^2 + E_{w}[\pav u]\Big) \lesa&\, \frac{1}{\tau^{1+\frac{3}{n}}}E_{w, \eta}^{\mathbf{v}}[u] + \frac{1}{\tau^{3+\frac{6}{n}-2\delta}}\|P_\eta u\|_{L^2}^2\,,
    \end{align*}
    which we plug into \eqref{sec:energy-met.eq:commuted-energy-proof-1} to obtain the desired inequality.
\end{proof}

\subsection{Higher-order estimates of the metric perturbations}
Next, we derive energy estimates for the metric components at high order. The overall idea is to apply the energy estimate from Proposition~\ref{sec:energy-met.prop:commuted-energy} to the quantities $\pa_x^I h^{\mu\nu}$, $1 \leq |I| \leq N-1$. We recall the metric evolution equations \eqref{sec:gauged-eqs.eq:wave}, which can be written as
\begin{align*}
    P_{\kappa-2}h^{00} &= \mc{L}_h^{00} + \mc{G}_h^{00},\\
    P_{\kappa/2}h^{0i} &= \mc{L}_h^{0i} + \mc{G}_h^{0i},\\
    P_2 h^{ij} &= \mc{L}_h^{ij} + \mc{G}_h^{ij}.
\end{align*}
\begin{definition}[Top-order wave energies]
    Recalling Definition \ref{def:corrected-pav-wave-energy}
    for the $\pav-$boosted and corrected wave energy, we define the top-order wave energies for the metric components :
\begin{subequations}
    \begin{align}
        E_{h^{00};N-1}^{\mathbf{v}}(\tau) &:= \sum_{1 \leq|I|\leq N-1}E_{w,(\kappa-2)}^{\mathbf{v}}[\pa_x^I h^{00}](\tau)\,,\\ 
        E_{h^{0*};N-1}^{\mathbf{v}}(\tau)&:= \sum_{i=1}^3\sum_{1 \leq |I|\leq N-1}E_{w,(\kappa/2)}^{\mathbf{v}}[\pa_x^I h^{0i}](\tau)\,,
        \\ E_{h^{**};N-1}^{\mathbf{v}}(\tau)&:= \sum_{i,j=1}^3\sum_{1 \leq |I|\leq N-1}E_{w, 2}^{\mathbf{v}}[\pa_x^I h^{ij}](\tau)\,.
    \end{align}
\end{subequations}
\end{definition}

By the bootstrap assumption, these energies are bounded in the fashion
\begin{equation}\label{sec:energy-met.eq:higher-energy-bootstrap}
     E_{h^{00};N-1}^{\mathbf{v}}(\tau) \lesa \frac{\ve^2}{\tau^{4+\frac{3}{n}-2\delta}}\,, \qquad E_{h^{0*};N-1}^{\mathbf{v}}(\tau)\lesa \frac{\ve^2}{\tau^{4+\frac{3}{n}-2\delta}}\,, \qquad E_{h^{**};N-1}^{\mathbf{v}}(\tau) \lesa \frac{\ve^2}{\tau^{4-2\delta}}\,.
\end{equation}

\begin{remark}
Thanks to the elliptic estimate provided in Lemma \ref{sec:energy-met.lem:elliptic}, we can control (up to a decaying error) the higher-order metric norms $S_{h;N}$ and $ S_{h;N-1}^{\mathbf{v}}$ defined in the bootstrap as \eqref{sec:stability.eq:boot-3}--\eqref{sec:stability.eq:boot-3.5} using the modified energies defined above. Indeed using Lemma \ref{sec:energy-met.lem:elliptic}, Remark \ref{rem:time-derivs-elliptic}, and Proposition~\ref{sec:energy-met.prop:com-ests}, we have that
\begin{equation*}
    \|\pa \pa_x h^{00}\|_{H^{N-1}} + \|\pa\pav h^{00}\|_{H^{N-1}} \lesa (E_{h^{00};N-1}^{\mathbf{v}})^{1/2} + \frac{1}{\tau^{2+\frac{3}{n}-\delta}}
\end{equation*}
with analogous estimates for the norms of $h^{0i}$, $h^{ij}$. From this follows the norm bound
\begin{equation*}
    S_{h;N}(\tau) + S_{h;N-1}^{\mathbf{v}}(\tau) \lesa \tau^{2+\frac{3}{2n}-\delta}\Big[(E_{h^{00}; N-1}^{\mathbf{v}})^{1/2} + (E_{h^{0*}; N-1}^{\mathbf{v}})^{1/2}\big] + \tau^{2-\delta}(E_{h^{**}; N-1}^{\mathbf{v}})^{1/2} + \frac{1}{\tau^{\frac{3}{2n}}}.
\end{equation*}
\end{remark}

The top-order energy estimates for the three functionals $E_{h^{00};N-1}^{\mathbf{v}}, E_{h^{0*};N-1}^{\mathbf{v}}, E_{h^{**};N-1}^{\mathbf{v}}$ are the subject of the following Proposition. We state the result now, but note that several key estimates are required, whose proofs we delay until later in the section.
\begin{proposition}[Estimate for the higher-order metric energies]\label{sec:energy-met.prop:higher-energy}
    Suppose the metric perturbations $h^{\mu\nu}$, $\mu,\nu=0,1,2,3$ satisfy the gauge-fixed Einstein equations \eqref{sec:gauged-eqs.eq:wave}, and the bootstrap assumptions \eqref{sec:globed.eq:boot}. Then the top-order metric energies satisfy the following  inequalities:
    \begin{subequations}\label{sec:energy-met.eq:higher-energy}
        \begin{align}
            \pa_\tau E_{h^{00};N-1}^{\mathbf{v}}(\tau) + \frac{2(\kappa-2)}{\tau} E_{h^{00};N-1}^{\mathbf{v}} 
            &\lesa \frac{1}{\tau^{3+\frac{3}{2n}-\delta}}S_{(\rho,v);N}(E_{h^{00};N-1}^{\mathbf{v}})^{1/2} + \frac{1}{\tau^{5+\frac{9}{2n}-3\delta}}\,,\\
            \pa_\tau E_{h^{0*};N-1}^{\mathbf{v}}(\tau) + \frac{\kappa}{\tau} E_{h^{0*};N-1}^{\mathbf{v}} 
            &\lesa \Big[\frac{1}{\tau} (E_{h^{00};N-1}^{\mathbf{v}})^{1/2} + \frac{1}{\tau^{3+\frac{9}{2n}-\delta}}S_{(\rho,v);N} \Big](E_{h^{0*};N-1}^{\mathbf{v}})^{1/2}+ \frac{1}{\tau^{5+\frac{9}{2n}-3\delta}}\,,\\
            \pa_\tau E_{h^{**};N-1}^{\mathbf{v}}(\tau) + \frac{4}{\tau} E_{h^{**};N-1}^{\mathbf{v}} 
            &\lesa \Big[\frac{1}{\tau} (E_{h^{0*};N-1}^{\mathbf{v}})^{1/2} + \frac{1}{\tau^{3+\frac{3}{2n}-\delta}}S_{(\rho,v);N}\Big](E_{h^{**};N-1}^{\mathbf{v}})^{1/2} + \frac{1}{\tau^{5+\frac{9}{2n}-3\delta}}\,.
        \end{align}
    \end{subequations}
\end{proposition}
\begin{proof}
    The differentiated metric pertubations have zero spatial average, and so we may apply the energy inequality from Proposition~\ref{sec:energy-met.prop:commuted-energy} to the energies that comprise $E_{h^{00};N-1}^{\mathbf{v}}$, $E_{h^{0*};N-1}^{\mathbf{v}}$, $E_{h^{**};N-1}^{\mathbf{v}}$. It follows that these higher-order energies satisfy the inequalities
    \begin{subequations}\label{sec:energy-met.eq:higher-energy-proof-1}
        \begin{align}
            \pa_\tau E_{h^{00};N-1}^{\mathbf{v}}(\tau) + \frac{2(\kappa-2)}{\tau} E_{h^{00};N-1}^{\mathbf{v}} &\lesa P_{h^{00}}^{\mathbf{v}}(E_{h^{00};N-1}^{\mathbf{v}})^{1/2} + \frac{1}{\tau^{2-\frac{3}{2n}-\delta}}(P_{h^{00}}^{\mathbf{v}})^2 \notag\\&\quad + \frac{1}{\tau^{2+\frac{9}{2n}-\delta}}(P_{h^{00}}^{\mathbf{v}})^2  + \frac{1}{\tau^{1+\frac{3}{2n}-\delta}}E_{h^{00};N-1}^{\mathbf{v}}\,, \\
            \pa_\tau E_{h^{0*};N-1}^{\mathbf{v}}(\tau) + \frac{\kappa}{\tau} E_{h^{0*};N-1}^{\mathbf{v}} &\lesa P_{h^{0*}}^{\mathbf{v}}(E_{h^{0*};N-1}^{\mathbf{v}})^{1/2} + \frac{1}{\tau^{2-\frac{3}{2n}-\delta}}(P_{h^{0*}}^{\mathbf{v}})^2 \notag\\&\quad + \frac{1}{\tau^{2+\frac{9}{2n}-\delta}}(P_{h^{0*}})^2 + \frac{1}{\tau^{1+\frac{3}{2n}-\delta}}E_{h^{0*};N-1}^{\mathbf{v}}\,, \\
            \pa_\tau E_{h^{**};N-1}^{\mathbf{v}}(\tau) + \frac{4}{\tau} E_{h^{**};N-1}^{\mathbf{v}} &\lesa P_{h^{**}}^{\mathbf{v}}(E_{h^{**};N-1}^{\mathbf{v}})^{1/2} + \frac{1}{\tau^{2-\frac{3}{2n}-\delta}}(P_{h^{**}}^{\mathbf{v}})^2 \notag\\&\quad+ \frac{1}{\tau^{2+\frac{9}{2n}-\delta}}(P_{h^{**}})^2 + \frac{1}{\tau^{1+\frac{3}{2n}-\delta}}E_{h^{**};N-1}^{\mathbf{v}}\,,
        \end{align}
    \end{subequations}
    where the quantities $P_{h^{00}}^{\mathbf{v}}$, $P_{h^{0*}}^{\mathbf{v}}$, $P_{h^{**}}^{\mathbf{v}}$ are defined as
    \begin{align*}
        P_{h^{00}}^{\mathbf{v}} &:= \sum_{1\leq|I|\leq N-1}\Big\|\pav\big[P_{\kappa-2}(\pa_x^I h^{00})\big]\Big\|_{L^2}\,,\\
        P_{h^{0*}}^{\mathbf{v}} &:= \sum_{i=1}^3\sum_{1\leq|I|\leq N-1}\Big\|\pav\big[P_{\kappa/2}(\pa_x^I h^{0i})\big]\Big\|_{L^2}\,,\\
        P_{h^{**}}^{\mathbf{v}} &:= \sum_{i,j=1}^3\sum_{1\leq|I|\leq N-1}\Big\|\pav\big[P_{2}(\pa_x^I h^{ij})\big]\Big\|_{L^2}\,,
    \end{align*}
    and $P_{h^{00}}$, $P_{h^{0*}}$, $P_{h^{**}}$ as
    \begin{align*}
        P_{h^{00}} &:= \sum_{1\leq |I|\leq N-1}\big\|P_{\kappa-2}(\pa_x^I h^{00})\|_{L^2}\,,\\
        P_{h^{0*}} &:= \sum_{i=1}^3\sum_{1\leq|I|\leq N-1}\big\|P_{\kappa/2}(\pa_x^I h^{0i})\|_{L^2}\,,\\
        P_{h^{**}} &:= \sum_{i,j=1}^3\sum_{1\leq|I|\leq N-1}\big\|P_{2}(\pa_x^I h^{ij})\|_{L^2}\,.
    \end{align*}
     We claim that these quantities obey the bounds
        \begin{subequations}\label{sec:energy-met.eq:higher-energy-proof-2}
        \begin{align}
        P_{h^{00}}^{\mathbf{v}} &\lesa \frac{1}{\tau^{3+\frac{3}{2n}}}S_{(\rho,v);N}+ \frac{1}{\tau^{3+\frac{3}{n}-\delta}}\,,\\
        P_{h^{0*}}^{\mathbf{v}} &\lesa \frac{1}{\tau}(E_{h^{00};N-1}^{\mathbf{v}})^{1/2} + \frac{1}{\tau^{3+\frac{9}{2n}}}S_{(\rho,v);N} + \frac{1}{\tau^{3+\frac{3}{2n}}}\,,\\
        P_{h^{**}}^{\mathbf{v}} &\lesa \frac{1}{\tau}(E_{h^{0*};N-1}^{\mathbf{v}})^{1/2} + \frac{1}{\tau^{3+\frac{3}{2n}}}S_{(\rho,v);N} + \frac{1}{\tau^{3+\frac{3}{n}}}\,,
        \end{align}
    \end{subequations}
    and
    \begin{equation}\label{sec:energy-met.eq:higher-energy-proof-3}
        P_{h^{00}} + P_{h^{0*}} + P_{h^{**}} \lesa \frac{1}{\tau^{3-\frac{3}{2n}-\delta}}\,.
    \end{equation}
    We delay the proofs of these bounds to Proposition~\ref{sec:energy-met.prop:com-ests}. For now, we observe that inserting them and the bound \eqref{sec:energy-met.eq:higher-energy-bootstrap} into \eqref{sec:energy-met.eq:higher-energy-proof-1} yields the inequalities \eqref{sec:energy-met.eq:higher-energy}.
\end{proof}
In the final part of this section, we  prove the claimed bounds \eqref{sec:energy-met.eq:higher-energy-proof-2}, \eqref{sec:energy-met.eq:higher-energy-proof-3}. We do this in two steps. First, we have the following lemma.
\begin{lemma}\label{sec:energy-met.lem:source-ests}
    We have the rough bounds
    \begin{equation}\label{sec:energy-met.eq:source-ests-1}
        \|P_{\kappa-2}h^{00}\|_{H^N} + \sum_{i=1}^3\|P_{\kappa/2}h^{0i}\|_{H^N} + \sum_{i,j=1}^3 \|P_2 h^{ij}\|_{H^N}\lesa \frac{1}{\tau^2}\,,
    \end{equation}
    \begin{equation}\label{sec:energy-met.eq:source-ests-2}
        \big\|\pa_x(P_{\kappa-2}h^{00})\big\|_{H^{N-1}} +  \sum_{i=1}^3\big\|\pa_x(P_{\kappa/2}h^{0i})\big\|_{H^{N-1}} + \sum_{i,j=1}^3\big\|\pa_x(P_{\kappa/2}h^{ij})\big\|_{H^{N-1}} \lesa \frac{1}{\tau^{3-\frac{3}{2n}-\delta}}\,,
    \end{equation}
    and the sharp bounds
    \begin{subequations}\label{sec:energy-met.eq:source-ests-3}
        \begin{align}
            \big\|\pav(P_{\kappa-2}h^{00})\big\|_{H^{N-1}} &\lesa \frac{1}{\tau^{3+\frac{3}{2n}-\delta}}S_{(\rho,v);N} + \frac{1}{\tau^{3+\frac{3}{n}}}\,,\\
            \sum_{i=1}^3\big\|\pav(P_{\kappa/2}h^{0i})\big\|_{H^{N-1}} &\lesa \frac{1}{\tau}(E_{h^{00};N-1}^{\mathbf{v}})^{1/2} + \frac{1}{\tau^{3+\frac{9}{2n}-\delta}}S_{(\rho,v);N} + \frac{1}{\tau^{3+\frac{3}{n}}}\,,\\
            \sum_{i,j=1}^3 \big\|\pav(P_2 h^{ij})\big\|_{H^{N-1}} &\lesa \frac{1}{\tau}(E_{h^{0*};N-1}^{\mathbf{v}})^{1/2} + \frac{1}{\tau^{3+\frac{3}{2n}-\delta}}S_{(\rho,v);N} + \frac{1}{\tau^{3+\frac{3}{n}}}\,.
        \end{align}
    \end{subequations}
\end{lemma}
\begin{proof}
    First we prove \eqref{sec:energy-met.eq:source-ests-1}. We observe that the error terms $\mc{G}_h^{\mu\nu}$ satisfy
    \begin{equation}\label{sec:energy-met.eq:source-ests-proof-1}
        \|\mc{G}_h^{\mu\nu}\|_{H^{N}} \lesa \frac{1}{\tau^{2+\frac{6}{n}}}\,.
    \end{equation}
    To see why this is true, one can go through the terms that comprise the $\mc{G}_h^{\mu\nu}$, and quickly verify that the $H^N$ norm of all terms are bounded by $C\tau^{-(2+\frac{6}{n})}$. In particular, for the semilinear terms $\mc{F}^{\mu\nu}$, we have
    \begin{equation*}
        \|\mc{F}^{\mu\nu}\|_{H^N} \leq \|\pa h\|_{H^N}^2 \lesa \frac{1}{\tau^{2+\frac{6}{n}}}\,.
    \end{equation*}
    Using \eqref{sec:energy-met.eq:source-ests-proof-1}, rough bounds on the source terms $\mc{L}_h^{\mu\nu}$, and the bootstrap assumptions, we find that the quantities $P_{\kappa-2}h^{00}$, $P_{\kappa/2}h^{0i}$, $P_2 h^{ij}$ satisfy the bounds 
        \begin{multline*}
        \|P_{\kappa-2}h^{00}\|_{H^N} + \sum_{i=1}^3\|P_{\kappa/2}h^{0i}\|_{H^N} + \sum_{i,j=1}^3 \|P_2 h^{ij}\|_{H^N}\\
        \lesa \frac{1}{\tau^2}\Big(\|h\|_{H^N} + \|\rho\|_{H^N} + \|v\|_{H^N}\Big) + \frac{1}{\tau}\Big(\|\pa h\|_{H^{N}} + |\pa_\tau \rho_\av|\Big) + \|\pa h\|_{H^N}^2  + \frac{1}{\tau^{2+\frac{6}{n}}} \lesa \frac{1}{\tau^2}\,,
    \end{multline*}
    which proves \eqref{sec:energy-met.eq:source-ests-1}. We note that this also gives bounds on the quantities $\pa_\tau^2 h^{\mu\nu}$, as we have
    \begin{equation}\label{sec:energy-met.eq:source-ests-proof-2}
        \|\pa_\tau^2 h^{00}\|_{H^{N-1}} \lesa \|P_{\kappa-2} h^{00}\|_{H^{N-1}} + \frac{1}{\tau}\|\pa_\tau h^{00}\|_{H^{N-1}} + \|\pa \pa_x h^{00}\|_{H^{N-1}} \leq \frac{1}{\tau^2}\,,
    \end{equation}
    and analogous bounds for the $h^{0i}$, $h^{ij}$.

    Next, we prove \eqref{sec:energy-met.eq:source-ests-1}. Repeating the previous argument, we have the following bound for the $ \mc{G}_h^{\mu\nu}$:
    \begin{equation}\label{sec:energy-met.eq:source-ests-proof-3}
        \|\pa \mc{G}^{\mu\nu}\|_{H^{N-1}} \lesa \|\pa \pa h\|_{H^{N-1}} \|\pa h \|_{H^{N-1}} +\frac{1}{\tau^{3+\frac{3}{n}-\delta}} \lesa \frac{1}{\tau^{3+\frac{3}{n}-\delta}}\,,
    \end{equation}
    where we used \eqref{sec:energy-met.eq:source-ests-proof-2} to estimate any second-order time derivatives of the metric perturbations. Once again, using \eqref{sec:energy-met.eq:source-ests-proof-3}, rough bounds on the source terms $\mc{L}_h^{\mu\nu}$, and the bootstrap assumptions, we have
    \begin{multline*}
        \big\|\pa_x(P_{\kappa-2}h^{00})\big\|_{H^{N-1}} +  \sum_{i=1}^3\big\|\pa_x(P_{\kappa/2}h^{0i})\big\|_{H^{N-1}} + \sum_{i,j=1}^3\big\|\pa_x(P_{\kappa/2}h^{ij})\big\|_{H^{N-1}}\\
        \lesa \frac{1}{\tau^2}\Big(\|\pa_x \rho\|_{H^{N-1}} + \|\pa_x v\|_{H^{N-1}} + \|\pa_x h\|_{H^{N-1}}\Big)+ \frac{1}{\tau}\|\pa_x \pa h\|_{H^{N-1}}+ \frac{1}{\tau^{3+\frac{3}{n}-\delta
        }} \lesa \frac{1}{\tau^{3-\frac{3}{2n}-\delta}}\,,
    \end{multline*}
    where the slowest decay arises from the term $\tau^{-2} \|\pa_x \rho\|_{H^{N-1}}$. Hence the bound \eqref{sec:energy-met.eq:source-ests-2} holds.

    It remains to prove \eqref{sec:energy-met.eq:source-ests-3}. Here we must be slightly more careful, and precisely analyse the structures of the source terms $\mc{L}_h^{\mu\nu}$ in order to derive sharp bounds. For $\mc{L}_h^{00}$, we have 
    \begin{multline}\label{sec:energy-met.eq:source-ests-proof-4}
        \big\|\pav \mc{L}_h^{00}\big\|_{H^{N-1}} \lesa \frac{1}{\tau^3}\Big(\|\rho-\rho_\av\|_{H^{N-1}} + \|h^{00}-h_{\av}^{00}\|_{H^{N-2}} + \big|(\rho_\av - 12) + 12h^{00}_\av\big|\Big)\\+\frac{1}{\tau^2}\Big(\|\pav\rho\|_{H^{N-1}} + \|\pav h^{00}\|_{H^{N-1}} + |\pa_\tau \rho_\av| + |\pa_\tau h_\av^{00}|\Big)  + \frac{1}{\tau}\Big(|\pa_\tau^2 \rho_\av| + |\pa_\tau^2 h_\av^{00}|\Big)\\
        \leq \frac{1}{\tau^2}\|\pav \rho\|_{H^N} + \frac{C}{\tau^3}\|\pa_x \rho_\av\|_{H^{N-1}} + \frac{1}{\tau^{3+\frac{3}{n}}}\,,
    \end{multline}
    where the bounds on the various spatial averages follow from Proposition~\ref{sec:avg-ests.prop:avg-ests} and Remark~\ref{sec:avg-ests.rmk:derivests}, while the estimates for the quantities $\|\rho-\rho_\av\|_{H^{N-1}}$, $\|h^{00}-h_\av^{00}\|_{H^{N-1}}$ follow from the Poincar\'e inequality. We then use Proposition~\ref{sec:stability.prop:fluid-time-deriv} to estimate the material derivative of $\rho$, giving
    \begin{equation*}
        \frac{1}{\tau^2} \|\pav \rho\|_{H^{N-1}} +\frac{1}{\tau^3}\|\pa_x \rho\|_{H^{N-1}}\lesa \frac{1}{\tau^{3+\frac{3}{2n}-\delta}}S_{(\rho,v);N} + \frac{1}{\tau^{3+\frac{3}{n}}}\,.
    \end{equation*}
    Combining this with \eqref{sec:energy-met.eq:source-ests-proof-3} and \eqref{sec:energy-met.eq:source-ests-proof-4}, we obtain the first of the bounds \eqref{sec:energy-met.eq:source-ests-3}. We repeat this argument for the source terms $\mc{L}_h^{0i}$, $\mc{L}_h^{ij}$, and get
    \begin{align}
        \sum_{i=1}^3\big\|\pav \mc{L}_h^{0i} \big\|_{H^{N-1}} &\lesa \frac{1}{\tau}(E_{h^{00};N-1}^{\mathbf{v}})^{1/2} + \frac{1}{\tau^{3+\frac{9}{2n}-\delta}}S_{(\rho,v);N} + \frac{1}{\tau^{3+\frac{3}{n}}}\,,\\
        \sum_{i,j=1}^3\big\|\pav \mc{L}_h^{ij} \big\|_{H^{N-1}} &\lesa \frac{1}{\tau}(E_{h^{0*};N-1}^{\mathbf{v}})^{1/2} + \frac{1}{\tau^{3+\frac{3}{2n}-\delta}}S_{(\rho,v);N} + \frac{1}{\tau^{3+\frac{3}{n}}}\,.
    \end{align}
    By the above bounds and \eqref{sec:energy-met.eq:source-ests-proof-3}, we recover the remaining estimates.
\end{proof}

\begin{proposition}[$\pav-$commuted Einstein equations]\label{sec:energy-met.prop:com-ests}
    The estimates \eqref{sec:energy-met.eq:higher-energy-proof-2}, \eqref{sec:energy-met.eq:higher-energy-proof-3} hold. That is, we have
    \begin{subequations}
        \begin{align}
            \label{sec:energy-met.eq:com-ests-1}
            \sum_{1\leq|I|\leq N-1}\Big\|\pav\big[P_{\kappa-2}(\pa_x^I h^{00})\big]\Big\|_{L^2} &\lesa  \frac{1}{\tau^{3+\frac{3}{2n}-\delta}}S_{(\rho,v);N}+ \frac{1}{\tau^{3+\frac{3}{n}}}\,,\\
            \label{sec:energy-met.eq:com-ests-2}
            \sum_{i=1}^3\sum_{1\leq|I|\leq N-1}\Big\|\pav\big[P_{\kappa/2}(\pa_x^I h^{0i})\big]\Big\|_{L^2} &\lesa \frac{1}{\tau}(E_{h^{00};N-1}^{\mathbf{v}})^{1/2} + \frac{1}{\tau^{3+\frac{9}{2n}-\delta}}S_{(\rho,v);N} + \frac{1}{\tau^{3+\frac{3}{n}}}\,,\\
            \label{sec:energy-met.eq:com-ests-3}
            \sum_{i,j=1}^3\sum_{1\leq|I|\leq N-1}\Big\|\pav\big[P_{2}(\pa_x^I h^{ij})\big]\Big\|_{L^2} &\lesa \frac{1}{\tau}(E_{h^{0*};N-1}^{\mathbf{v}})^{1/2} + \frac{1}{\tau^{3+\frac{3}{2n}-\delta}}S_{(\rho,v);N} + \frac{1}{\tau^{3+\frac{3}{n}}}\,,
        \end{align}
    \end{subequations}
    and
    \begin{multline}\label{sec:energy-met.eq:com-ests-4}
        \sum_{1\leq|I|\leq N-1}\big\|P_{\kappa-2}(\pa_x^I h^{00})\big\|_{L^2}+
        \sum_{i=1}^3\sum_{1\leq|I|\leq N-1}\big\|P_{\kappa/2}(\pa_x^I h^{0i})\big\|_{L^2}\\+  \sum_{i,j=1}^3\sum_{1\leq|I|\leq N-1}\big\|P_{2}(\pa_x^I h^{ij})\big\|_{L^2} \lesa \frac{1}{\tau^{3-\frac{3}{2n}-\delta}}\,.
    \end{multline}
\end{proposition}
\begin{proof}
    In view of Lemma~\ref{sec:energy-met.lem:source-ests}, we are done if we can prove the commutator estimates 
    \begin{subequations}
        \begin{align}
            \label{sec:energy-met.eq:com-ests-proof-1}
            \sum_{\mu,\nu=0}^3\sum_{1\leq |I| \leq N-1}\Big\|\pav\big(\big[P_\eta,\pa_x^I] h^{\mu\nu}\big)\Big\|_{L^2} \lesa \frac{1}{\tau^{4-2\delta}}\,,\\
            \label{sec:energy-met.eq:com-ests-proof-2}
            \sum_{\mu,\nu=0}^3\sum_{1\leq |I| \leq N}\Big\|\big[P_\eta,\pa_x^I] h^{\mu\nu}\Big\|_{L^2} \lesa \frac{1}{\tau^{4-2\delta}}\,,
        \end{align}
    \end{subequations}
     since $\delta \leq \frac{1}{3}\big(1-\frac{3}{n}\big)$ implies $3+\frac{3}{n} < 4-2\delta$. Note that these estimates hold for all $\eta \in \RR$, and for all $h^{\mu\nu}$,
    Repeating the argument from \protect{\cite[Lemma 7.4]{Ber:decelerated:25}}, we write out the commutator $[P_\eta,\pa_x^I]$ in full, which gives
    \begin{equation*}
        [P_\eta,\pa_x^I]h^{\mu\nu} = \sum_{1\leq J \leq I}\sum_{\alpha,\beta=0}^3 c_{I,J}^{(1)}(\pa_x^J h^{\alpha\beta}) (\pa_\alpha\pa_\beta \pa_x^{I-J} h^{\mu\nu}\big)+ \frac{1}{\tau}\sum_{1\leq J \leq I}c_{I,J}^{(2)}(\pa_x^J h^{00})(\pa_\tau \pa_x^{I-J} h^{\mu\nu})\,,
    \end{equation*}
    with constants $c_{I,J}^{(1)}$, $c_{I,J}^{(2)}$ computed via the Leibniz rule. For each of the products
    \begin{equation*}
        (\pa_x^J h^{\alpha\beta})(\pa_\alpha \pa_\beta \pa_x^{I-J}h^{\mu\nu})\,,\qquad (\pa_x^J h^{00})( \pa_\tau \pa_x^{I-J}h^{\mu\nu})\,,
    \end{equation*}
    if we assume that $|I| \leq N$, $|J| \leq N-1$ then at least one factor must contain $\lfloor (N+2) /2 \rfloor$ or fewer derivatives. We estimate this factor in $L^\infty$, and the other in $L^2$. Since $N \geq 3$, we have $\lfloor (N+2)/2 \rfloor \leq N-1$, and so all resulting terms are controlled by the total energy. This argument yields the bound
    \begin{multline}\label{sec:energy-met.eq:com-ests-proof-3}
        \sum_{1\leq |I| \leq N}\Big\|\big[P_\eta,\pa_x^I] h^{\mu\nu}\Big\|_{L^2} \lesa \big\|\pa_x h\big\|_{H^N}\Big(\sum_{\alpha,\beta=0}^3\big\|\pa_\alpha\pa_\beta h^{\mu\nu}\big\|_{H^{N-1}} + \frac{1}{\tau}\big\|\pa_\tau h^{\mu\nu}\big\|_{H^{N}}\Big)\\
        \lesa \frac{1}{\tau^{2-\delta}}\big\|\pa_\tau^2 h^{\mu\nu}\big\|_{H^{N-1}} + \frac{1}{\tau^{4-2\delta}}\,,
    \end{multline}
    where in the second inequality we separated out the second-order time derivatives of $h^{\mu\nu}$, and applied the bootstrap bounds to all other terms. We estimate $\|\pa_\tau^2 h^{\mu\nu}\|$ using the bound \eqref{sec:energy-met.eq:source-ests-proof-2} that was stated in the proof of the previous lemma, giving
    \begin{equation*}
        \|\pa_\tau^2 h^{\mu\nu}\|_{H^{N-1}} \leq \frac{C}{\tau^2}\,.
    \end{equation*}
    Plugging this into \eqref{sec:energy-met.eq:com-ests-proof-3}, we obtain the estimates \eqref{sec:energy-met.eq:com-ests-proof-2}. 
    
    For the estimates \eqref{sec:energy-met.eq:com-ests-proof-1}, we have the expression
    \begin{multline*}
        \pav\big([P_\eta,\pa_x^I]h^{\mu\nu}\big) = \sum_{1\leq J \leq I}\sum_{\alpha,\beta=0}^3 c_{I,J}^{(1)}\Big[\big(v^\lambda\pa_\lambda\pa_x^J h^{\alpha\beta}\big) \big(\pa_\alpha\pa_\beta \pa_x^{I-J} h^{\mu\nu}\big) + (\pa_x^J h^{\alpha\beta}) \big(v^\lambda \pa_\lambda\pa_\alpha\pa_\beta \pa_x^{I-J} h^{\mu\nu}\big)\Big]\\+ \frac{1}{\tau}\sum_{1\leq J \leq I} c_{I,J}^{(2)}\Big[-\frac{1}{\tau}(\pa_x^J h^{00})\big(\pa_\tau \pa_x^{I-J} h^{\mu\nu}\big) + \big(v^\lambda\pa_\lambda\pa_x^J h^{00}\big)\big(\pa_\tau \pa_x^{I-J} h^{\mu\nu}\big) + (\pa_x^J h^{00})\big(v^\lambda\pa_\lambda\pa_\tau \pa_x^{I-J} h^{\mu\nu}\big)\Big]\,.
    \end{multline*}
    This arises from distributing at most $\partial^N$ derivatives. Since $1 \leq |J| \leq |I| \leq N-1$, each product has a factor with $\lfloor (N+2)/2 \rfloor$ or fewer derivatives of $h^{\mu\nu}$, and so we estimate these as before with the Sobolev embedding, to arrive at the bound 
    \begin{multline}\label{sec:energy-met.eq:com-ests-proof-4}
        \sum_{1\leq |I| \leq N - 1}\Big\|\pav\big(\big[P_\eta,\pa_x^I] h^{\mu\nu}\big)\Big\|_{L^2}\\\lesa \|\pa_x h\|_{H^N}\Big(\sum_{\alpha,\beta,\lambda=0}^3\big\|\pa_\alpha\pa_\beta\pa_\lambda h^{\mu\nu}\big\|_{H^{N-2}} + \frac{1}{\tau}\sum_{\lambda=0}^3 \big\|\pa_\lambda \pa_\tau h^{\mu\nu}\big\|_{H^{N-1}} + \frac{1}{\tau^2}\|\pa_\tau h^{\mu\nu}\|_{H^{N}}\Big)\\
        + \|\pa\pa_x h\|_{H^{N-1}}\big(\sum_{\alpha,\beta=0}^3\big\|\pa_\alpha\pa_\beta h^{\mu\nu}\big\|_{H^{N-1}} + \frac{1}{\tau}\big\|\pa_\tau h^{\mu\nu}\big\|_{H^{N}}\Big)\\
        \lesa \frac{1}{\tau^{2-\delta}}\Big(\big\|\pa_\tau^3 h^{\mu\nu}\big\|_{H^{N-2}} + \|\pa_\tau^2 h^{\mu\nu}\big\|_{H^{N-1}}\Big) + \frac{1}{\tau^{4-2\delta}}\,.
    \end{multline}
    Now we have separated out the second and third time derivatives. We bound the second time derivatives using \eqref{sec:energy-met.eq:source-ests-proof-2}. To estimate the third time derivatives, we rearranging the identity
    \begin{multline*}
        \pav \big(P_{\kappa-2}h^{00}\big) = v^0 g^{00}\pa_\tau^3 h^{00} + v^a g^{00}\pa_a \pa_\tau^2 h^{00}+ 2h^{0a}v^\lambda \pa_a \pa_\lambda\pa_\tau h^{00} + g^{ab}v^\lambda\pa_a\pa_b \pa_\lambda h^{00}\\
        + \pav h^{\alpha\beta}\pa_\alpha \pa_\beta h^{00} + \frac{2(\kappa-2)}{\tau}\pav h^{00}\pa_\tau h^{00} -\frac{2(\kappa-2)}{\tau^2}v^0g^{00}\pa_\tau h^{00} + \frac{2(\kappa-2)}{\tau}g^{00}v^\lambda\pa_\lambda\pa_\tau h^{00}\,,
    \end{multline*}
    and using Lemma~\ref{sec:energy-met.lem:source-ests} to estimate $\pav(P_{\kappa-2} h^{00})$, and \eqref{sec:energy-met.eq:source-ests-proof-2} to estimate $\pa_x \pa_\tau^2 h^{00}$, we obtain
    \begin{equation*}
        \|\pa_\tau^3 h^{00}\|_{H^{N-2}} \lesa \big\|\pav (P_{\kappa-2} h^{00})\big\|_{H^{N-2}} + \frac{1}{\tau^{2-\delta}} \lesa \frac{1}{\tau^{2-\delta}}\,.
    \end{equation*}
    Combining this with \eqref{sec:energy-met.eq:com-ests-proof-4} gives the bound \eqref{sec:energy-met.eq:com-ests-proof-1} for $h^{00}$; the estimates for $h^{0i}$, $h^{ij}$ follow by an identical argument.
\end{proof}
\section{Energy estimates for the fluid}\label{sec:energy-fluid}
Next, we prove energy estimates for higher-order derivatives of the fluid variables $(\rho,v)$. The main estimates are those in Proposition~\ref{sec:energy-fluid.prop:corrected-energy} and Proposition~\ref{sec:energy-fluid.prop:higher-energy}.
In Proposition~\ref{sec:energy-fluid.prop:corrected-energy}, we prove an estimate for the fundamental energy which we use to estimate the fluid variables. This energy builds on the canonical $L^2$-energy for the Euler equations, but with two important modifications. Like we did for the metric energy, we include certain lower-order correction terms which create coercive bulk terms in the energy estimate, leading to decay. We also commute the canonical energy with spatial derivatives, which allows us to employ an important div-curl identity to precisely control the resulting damping terms.

In Proposition~\ref{sec:energy-fluid.prop:higher-energy}, we apply the energy estimate derived in Proposition~\ref{sec:energy-fluid.prop:corrected-energy} to derivatives of the fluid variables. Similar to the gauge-fixed Einstein equations \eqref{sec:gauged-eqs.eq:wave}, the Euler equations \eqref{sec:gauged-eqs.eq:fluid} contain various source terms which decay at a critical rate. Here too we must keep careful track of these terms in our analysis of the higher-order fluid energy.

\begin{definition}[Fluid operators $Q, R^\nu$ and their weighted norms]\label{def:fluid-ops-QR}
    We define the following first-order differential operators acting on a scalar function $\varphi$ and vectorfield $\xi$:
\begin{align*}
	Q(\varphi,\xi) &:= v^\mu \pa_\mu \varphi + \chi \pa_\mu \xi^\mu\,,\\
	R^\nu(\varphi,\xi) &:= v^\mu \pa_\mu \xi^\nu + \frac{p'(\rho)}{\tau^{\frac{6}{n}}\chi}\Pi^{\mu\nu}\pa_\mu \varphi + \frac{2}{\tau}v^0 \xi^\nu\,.
\end{align*}
For later use, we also introduce certain norms of $Q, R$
    \begin{align*}
        f_{Q,R}[\varphi,\xi]&:= \big\|Q(\varphi,\xi)\big\|_{L^2} + \tau^{\frac{3}{n}}\sum_{i=1}^3 \big\|R^i(\varphi,\xi)\big\|_{L^2}\,, \\
        f_{Q,R;1}[\varphi,\xi](\tau)&:= \big\|\pa_x \big(Q(\varphi,\xi)\big)\big\|_{L^2} + \tau^{\frac{3}{n}}\sum_{\mu=0}^3\big\|\pa_x \big(R(\varphi,\xi)\big)\big\|_{L^2} + f_{Q,R}[\varphi,\xi](\tau)\,.
    \end{align*}
\end{definition}

Using the above definition, the Euler equations \eqref{sec:gauged-eqs.eq:fluid} can be expressed simply as
\begin{align*}
	Q(\rho,v) = \mc{L}_{(\rho,v)} + \mc{G}_{(\rho,v)}\,,\\
	R^i(\rho,v) = \mc{L}_{(\rho,v)}^i + \mc{G}_{(\rho,v)}^i\,.
\end{align*}
\subsection{Identities for the operators \texorpdfstring{$Q$, $R$}{Q, R}}
We briefly collect several results regarding the differential operators $Q$, $R$ which will be useful energy estimates that follow. The first shows that $R^0(\varphi,\xi)$ can be expressed purely in terms of the $R^i(\varphi,\xi)$ and lower-order differential operators.
\begin{lemma}\label{sec:energy-fluid.lem:R0-identity}
    Let $\varphi$ be a scalar function, and $\xi$ a vectorfield satisfying $g(\xi,v) = 0$, such that $\varphi$ and the components of $\xi$ have vanishing spatial average. Then $R^0(\varphi,\xi)$ satisfies the identity
    \begin{equation}\label{sec:energy-fluid.eq:R0-identity}
        R^0(\varphi,\xi) = -\frac{v_a}{v_0}R^a(\varphi,\xi) - v^\mu \pa_\mu \big(\frac{v_a}{v_0}\Big)\xi^a\,,
    \end{equation}
    moreover for all $k\in \{0,\dots,N-1\}$, we have
    \begin{equation}\label{sec:fluid-est.eq:R0-norm}
        \big\|R^0(\varphi,\xi)\big\|_{H^k} \lesa \frac{1}{\tau^{1+\frac{3}{2n}-\delta}}\sum_{i=1}^3\|R^i(\varphi,\xi)\|_{H^k} + \frac{1}{\tau^{1+\frac{9}{2n}-\delta}}\|\pa_x\xi\|_{H^k}\,.
    \end{equation}
\end{lemma}
\begin{proof}
    We rearrange $g(\xi,v) = 0$, and solve for $\xi^0$ to obtain
    \begin{equation*}
        \xi^0 = -\frac{1}{v_0}v_a \xi^a,
    \end{equation*}
    where $v_\mu = g_{\mu\alpha} v^\alpha$. It follows that
    \begin{equation*}
        R^0(\varphi,\xi) = - \frac{v_a}{v_0} v^\mu \pa_\mu \xi^a + \frac{p'(\rho)}{\tau^{\frac{6}{n}}\chi}\Pi^{0\mu}\pa_\mu \varphi -\frac{2}{\tau}\frac{v_a}{v_0}v^0\xi^a -v^\mu \pa_\mu \Big(\frac{v_a}{v_0}\Big)  \xi^a.
    \end{equation*}
    Then we have $v_\alpha \Pi^{\alpha \mu} = 0$, which implies
    \begin{equation*}
        \Pi^{0\mu}\pa_\mu \varphi = -\frac{v_a}{v_0}\Pi^{a\mu}\pa_\mu \varphi\,,
    \end{equation*}
    Inserting this into the previous equation yields the identity \eqref{sec:energy-fluid.lem:R0-identity}. Using Proposition~\ref{sec:stability.prop:fluid-time-deriv} to estimate the time derivative of $v$ as well as the Poincar\'e inequality, we have the norm bound
    \begin{multline*}
        \|R^0(\varphi,\xi)\|_{H^k} \lesa \|v\|_{H^N}\sum_{i=1}^3\|R^i(\varphi,\xi)\|_{H^k}\\+\sum_{a=1}^3\Big(\|\pa_\tau (v_a/v_0)\|_{H^{N-1}} + \|v\|_{H^N}\|\pa_x(v_a/v_0)\|_{H^{N-1}}\Big)\|\xi\|_{H^k}\\\lesa \frac{1}{\tau^{1+\frac{3}{2n}-\delta}}\sum_{i=1}^3\|R^i(\varphi,\xi)\|_{H^k} + \frac{1}{\tau^{1+\frac{9}{2n}-\delta}}\|\pa_x\xi\|_{H^k}\,,
    \end{multline*}
    which is \eqref{sec:fluid-est.eq:R0-norm}.
\end{proof}

The next Lemma allows us to rewrite time derivatives of $\varphi$ and $\xi^i$ in terms of $Q(\varphi,\xi)$, $R^i(\varphi,\xi)$.
\begin{lemma}\label{sec:energy-fluid.lem:time-deriv}
    Let $\varphi$ be a scalar function, and $\xi$ a vectorfield satisfying $g(\xi,v) = 0$, such that $\varphi$ and the components of $\xi$ have vanishing spatial average. Then the quantities $\pa_\tau \varphi$, $\pa_\tau \xi$ satisfy the identities
    \begin{subequations}
        \begin{align}
            \label{sec:energy-fluid.eq:time-deriv-1}
            \pa_\tau \varphi &=\frac{1}{(v^0)^2-\tau^{-\frac{6}{n}}p'(\rho)\Pi^{00}}\Big(v^0\mc{A}-\chi\mc{B}^0\Big)\,,\\
            \label{sec:energy-fluid.eq:time-deriv-2}
            \pa_\tau  \xi^i &= \frac{1}{v^0}\mc{B}^i -\frac{1}{(v^0)^2-\tau^{-\frac{6}{n}}p'(\rho)\Pi^{00}}\frac{p'(\rho)}{\tau^{\frac{6}{n}}\chi v^0}\Pi^{0i}\Big(v^0\mc{A}-\chi\mc{B}^0\Big)\,,
        \end{align}
    \end{subequations}
    where $\mc{A}$, $\mc{B}^\mu$, are defined by
    \begin{align*}
        \mc{A}&:=  Q(\varphi,\xi)-v^a \pa_a \varphi - \chi \pa_a \xi^a\,,\\
        \mc{B}^\mu &:= R^\mu(\varphi,\xi) - v^a\pa_a \xi^\mu - \frac{p'(\rho)}{\tau^{\frac{6}{n}}\chi}\Pi^{\mu a}\pa_a \varphi -\frac2\tau v^0 \xi^\mu\,.
    \end{align*}
    Moreover, the following norm bounds hold for $k \in \{0,\dots,N-1\}$:
    \begin{subequations}\label{sec:energy-fluid.eq:time-deriv-3}
        \begin{align}
            \|\pa_\tau \varphi\|_{H^k} \lesa&\, \|Q(\varphi,\xi)\|_{H^k} + \frac{1}{\tau^{1+\frac{3}{2n}-\delta}}\sum_{i=1}^3\|R^i(\varphi,\xi)\|_{H^k} + \frac{1}{\tau^{1+\frac{3}{2n}-\delta}}\|\pa_x \varphi\|_{H^k}+\|\pa_x \xi\|_{H^k}\,,\\
            \sum_{i=1}^3\|\pa_\tau \xi^i\|_{H^k} \lesa&\, \frac{1}{\tau^{1+\frac{9}{2n}-\delta}}\|Q(\varphi,\xi)\|_{H^k} + \sum_{i=1}^3\|R^i(\varphi,\xi)\|_{H^k}+ \frac{1}{\tau^{\frac{6}{n}}}\|\pa_x \varphi\|_{H^k}+\frac{1}{\tau}\|\pa_x \xi\|_{H^k}\,.
        \end{align}
    \end{subequations}
\end{lemma}

\begin{proof}
    We rearrange the expressions defining $Q(\varphi,\xi)$ and $R^0(\varphi,\xi)$ to find that the time derivatives of $\varphi$ and $\xi^\mu$ satisfy the system
    \begin{align*}
        v^0\pa_\tau \varphi + \chi \pa_\tau \xi^0 &= \mc{A}\,,\\
        \frac{p'(\rho)}{\tau^{\frac{6}{n}}\chi} \Pi^{0\mu}\pa_\tau \varphi + v^0 \pa_\tau \xi^\mu &= \mc{B}^\mu\,.
    \end{align*}
    The equations for $\pa_\tau \varphi$ and $\pa_\tau \xi^0$ decouple, and solving these first give
    \begin{align*}
        \pa_\tau \varphi =&\frac{1}{(v^0)^2 - \tau^{-\frac{6}{n}}p'(\rho)\Pi^{00}}\Big(v^0 \mc{A}-\chi \mc{B}^0\Big)\,,\\
        \pa_\tau \xi^0 =&\frac{1}{(v^0)^2 - \tau^{-\frac{6}{n}}p'(\rho)\Pi^{00}}\Big(v^0\mc{B}^0 - \frac{p'(\rho)}{\tau^{\frac{6}{n}}\chi}\Pi^{00}\mc{A}\Big)\,,
    \end{align*}
    the first of which is \eqref{sec:energy-fluid.eq:time-deriv-1}. Then we substitute this into the remaining equations for the $\xi^i$ to obtain \eqref{sec:energy-fluid.eq:time-deriv-2}. Under the assumption $g(\xi,v) = 0$, we can express $\mc{B}^0$ purely in terms of $\varphi$ and the spatial components of $\xi$. Using Lemma~\ref{sec:energy-fluid.lem:R0-identity} to substitute out $R^0$, and the identity $\xi^0 = -\frac{v_a}{v_0}\xi^a$, we have
    \begin{equation*}
        \mc{B}^0 = -\frac{v_a}{v_0}R^a(\varphi,\xi) + \frac{v_a v^b}{v_0} \pa_b \xi^a - \frac{p'(\rho)}{\tau^{\frac{6}{n}}\chi}\Pi^{0a}\pa_a \varphi - v^0 \pa_\tau\Big(\frac{v_a}{v_0}\big)\xi^a + \frac2\tau v^0 \frac{v_a}{v_0} \xi^a.\,.
    \end{equation*}
    
    The norm bounds \eqref{sec:energy-fluid.eq:time-deriv-3} follow from
    \begin{align*}
        \|\mc{A}\|_{H^k} &\lesa \|Q(\varphi,\xi)\|_{H^k} + \frac{1}{\tau^{1+\frac{3}{2n}-\delta}}\|\pa_x \varphi\|_{H^k} + \|\pa_x \xi\|_{H^k}\,,\\
        \|\mc{B}^0\|_{H^k} &\lesa \frac{1}{\tau^{1+\frac{3}{2n}-\delta}}\Big(\sum_{i=1}^3\|R^i(\varphi,\xi)\|_{H^k} + \frac{1}{\tau}\|\pa_x\xi\|_{H^k} + \frac{1}{\tau^{\frac{6}{n}}}\|\pa_x \varphi\|_{H^k}\Big)\,,\\
        \sum_{i=1}^3\|\mc{B}^i\|_{H^k} &\lesa \sum_{i=1}^3\|R^i(\varphi,\xi)\|_{H^k} + \frac{1}{\tau}\|\pa_x\xi\|_{H^k} + \frac{1}{\tau^{\frac{6}{n}}}\|\pa_x \varphi\|_{H^k}\,,\
    \end{align*}
    where we used the Poincar\'e inequality in the second and third bounds, recalling also that $\varphi$ and the components of $\xi$ have vanishing spatial average.
\end{proof}
\subsection{The base fluid energy}
\begin{definition}[Base fluid energy and norm of $\xi$]\label{def:base-fluid-energy}
   For a given scalar function $\varphi$ and vectorfield $\xi$, we define the base \textit{fluid energy}
\begin{equation}
	E_{f}[\varphi,\xi](\tau):= \frac{1}{2}\int_{\TT^3} p'(\rho) v^0 \varphi^2 + 2\chi p'(\rho) \varphi \xi^0 +  \tau^{\frac{6}{n}} \chi^2 v^0 g_{\alpha\beta} \xi^\alpha \xi^\beta,
\end{equation} 
where we recall that $p'(\rho) := \frac{\de p}{\de \rho} = C(1+\frac{1}{n})\rho^{\frac{1}{n}}$. Similar to $v^i$, we define $\|\xi\|_{L^2}^2 = \sum_{i=1}^3\|\xi^i\|_{L^2}^2$. 
\end{definition}

We show that this energy controls certain time-weighted norms of $\varphi$, $\xi$.
\begin{lemma}\label{sec:energy-fluid.lem:coercivity}
	Let $\varphi$ be a scalar function, and $\xi$ a vectorfield satisfying $g(\xi,v) = 0$. Then we have the equivalence
	\begin{equation}
		E_{f}[\varphi,\xi] \sim \|\varphi\|_{L^2}^2 + \tau^{\frac{6}{n}}\|\xi\|_{L^2}^2\,.
	\end{equation}
\end{lemma}
\begin{proof}
	The identity $g(\xi,v) = 0$ can be rearranged and solved for $\xi^0$ like
	\begin{equation*}
		\xi^0 = -\frac{1}{v_0}v_a \xi^a\,,
	\end{equation*}
	where $v_\mu = g_{\mu\alpha}v^\alpha$. Using the bootstrap assumptions, it follows that
	\begin{equation*}
		\|\xi^0\|_{L^2} \lesa \frac{1}{\tau^{1+\frac{3}{2n}-\delta}}\|\xi\|_{L^2}\,,
	\end{equation*}
    therefore we have the following bound for the mixed term in $E_f$:
	\begin{equation*}
		\Big|\int_{\TT^3} p'(\rho)\chi \varphi \xi^0\Big| \lesa \frac{1}{\tau^{1+\frac{3}{2n}-\delta}}\|\varphi\|_{L^2}\|\xi\|_{L^2} \lesa \frac{1}{\tau^{1+\frac{9}{2n}-\delta}}\Big(\|\varphi\|_{L^2}^2 + \tau^{\frac{6}{n}}\|\xi\|_{L^2}^2\Big)\,,
	\end{equation*}
    Meanwhile, the quantities $p'(\rho)$ and $\chi$ are nonzero to leading order in $\tau$, and so the first and last term comprising $E_f$ satisfy
    \begin{align*}
        \int_{\TT^3} p'(\rho) v^0 \varphi^2 \sim \|\varphi\|_{L^2}^2,\\\int_{\TT^3}\tau^{\frac{6}{n}} \chi^2 v^0 g_{\alpha\beta} \xi^\alpha \xi^\beta \sim \tau^{\frac{6}{n}}\|\xi\|_{L^2}^2\,.
    \end{align*}
	The result now follows.
\end{proof}

Next, we prove the basic energy estimate satisfied by the fluid energy $E_f$.
\begin{lemma}\label{sec:energy-fluid.lem:base-energy}
	Let $\varphi$ be a scalar function, and $\xi$ a vectorfield satisfying $g(\xi,v) = 0$. Then the base fluid energy $E_f$ satisfies the differential inequality
	\begin{equation}\label{sec:energy-fluid.eq:base-energy}
		\pa_\tau E_f[\varphi,\xi](\tau) + \frac{2-\frac{3}{n}}{\tau}\int_{\TT^3}\tau^{\frac{6}{n}} \chi^2 v^0 g_{\alpha\beta}\xi^\alpha \xi^\beta
		\lesa \big(E_f[\varphi,\xi]\big)^{1/2}f_{Q,R}[\varphi,\xi]+\frac{1}{\tau^{1+\delta}}E_f[\varphi,\xi]\,.
	\end{equation}
    where $f_{Q,R}$ was given in Definition \ref{def:fluid-ops-QR}.
\end{lemma}
\begin{proof}
	We differentiate the base fluid energy $E_f$ term-by term. First, we compute using integration by parts that
	\begin{align*}
		\frac{1}{2}\pa_\tau \Big(\int_{\TT^3} p'(\rho) v^0 \varphi^2\Big) =&\, \int_{\TT^3} p'(\rho) v^0 \varphi \pa_\tau \varphi + \frac{1}{2}\int_{\TT^3} \pa_\tau (p'(\rho) v^0)\varphi^2\\
		=&\,\int_{\TT^3} p'(\rho) \varphi v^\mu \pa_\mu \varphi - \int_{\TT^3} p'(\rho)\varphi v^a \pa_a \varphi + \frac{1}{2}\int_{\TT^3} \pa_\tau \big(p'(\rho) v^0\big)\varphi^2\\
		=&\, \int_{\TT^3} p'(\rho) \varphi v^\mu \pa_\mu \varphi  + \frac{1}{2}\int_{\TT^3} \pa_\mu \big(p'(\rho) v^\mu\big)\varphi^2\,.
	\end{align*}
	We then substitute in the operator $Q(\varphi,\xi)$ to obtain
	\begin{equation}\label{sec:energy-fluid.eq:base-energy-prf-1}
		\frac{1}{2}\pa_\tau \Big(\int_{\TT^3} p'(\rho) v^0 \varphi^2\Big) = \int_{\TT^3} p'(\rho) \varphi Q(\varphi,\xi) - \int_{\TT^3} p'(\rho)\chi \varphi \pa_\mu \xi^\mu + \frac{1}{2}\int_{\TT^3} \pa_\mu (p'(\rho) v^\mu)\varphi^2.
	\end{equation}
	For the second term that makes up $E_f$, integrating by parts gives
	\begin{align*}
		\pa_\tau \Big(\int_{\TT^3} \chi p'(\rho) \varphi \xi^0\Big) 
        =&\, \int_{\TT^3} \chi p'(\rho)\big(\xi^0 \pa_\tau \varphi + \varphi \pa_\tau \xi^0\big) + \int_{\TT^3} \pa_\tau  \big(\chi p'(\rho)\big) \varphi \xi^0\\
		=&\, \int_{\TT^3} \chi p'(\rho)\big(\xi^\mu \pa_\mu \varphi + \varphi \pa_\mu \xi^\mu\big) - \int_{\TT^3} \chi p'(\rho)\big(\xi^a \pa_a \varphi+ \varphi \pa_a \xi^a \big)\\
        &+ \int_{\TT^3} \pa_\tau  \big(\chi p'(\rho)\big) \varphi \xi^0\\
		=&  \int_{\TT^3} \chi p'(\rho)\big(\xi^\mu \pa_\mu \varphi + \varphi \pa_\mu \xi^\mu\big) + \int_{\TT^3} \pa_\mu  \big(\chi p'(\rho)\big) \varphi \xi^\mu.
	\end{align*}
	We differentiate the final term in $E_f$, integrate by parts, then substitute the operators $R^\mu(\varphi,\xi)$:
	\begin{align}
        \label{sec:energy-fluid.eq:base-energy-prf-3}
		\frac{1}{2}\pa_\tau&\Big(\int_{\TT^3} \tau^{\frac{6}{n}}\chi^2 v^0 g_{\alpha\beta}\xi^\alpha \xi^\beta\Big) \\
        =&\, \int_{\TT^3} \tau^{\frac{6}{n}} \chi^2 g_{\alpha\beta}\xi^\alpha v^0\pa_\tau \xi^\beta + \frac{3}{n\tau}\int_{\TT^3} \tau^{\frac{6}{n}} \chi^2 v^0 g_{\alpha\beta}\xi^\alpha \xi^\beta+\int_{\TT^3} \tau^{\frac{6}{n}} \pa_\tau \big(\chi^2 g_{\alpha\beta}v^0\big)\xi^\alpha \xi^\beta\nonumber\\
		=&\,\int_{\TT^3} \tau^{\frac{6}{n}} \chi^2 g_{\alpha\beta}\xi^\alpha v^\mu\pa_\mu \xi^\beta + \frac{3}{n\tau}\int_{\TT^3} \tau^{\frac{6}{n}} \chi^2 v^0 g_{\alpha\beta}\xi^\alpha \xi^\beta
        +\int_{\TT^3} \tau^{\frac{6}{n}} \pa_\mu \big(\chi^2 g_{\alpha\beta}v^\mu\big)\xi^\alpha \xi^\beta\nonumber\\
		=&\,\int_{\TT^3} \tau^{\frac{6}{n}} \chi^2 g_{\alpha\beta}\xi^\alpha R^\beta(\varphi,\xi)  - \int_{\TT^3} p'(\rho)\chi g_{\alpha\beta}\xi^\alpha \Pi^{\mu\beta}\pa_\mu \varphi\nonumber\\
		&-  \frac{2-\frac{3}{n}}{\tau}\int_{\TT^3} \tau^{\frac{6}{n}} \chi^2 v^0 g_{\alpha\beta}\xi^\alpha \xi^\beta
		+\int_{\TT^3} \tau^{\frac{6}{n}} \pa_\mu \big(\chi^2 g_{\alpha\beta}v^\mu\big)\xi^\alpha \xi^\beta\,.\nonumber
	\end{align}
	Combining \eqref{sec:energy-fluid.eq:base-energy-prf-1}--\eqref{sec:energy-fluid.eq:base-energy-prf-3}, we arrive at the identity
	\begin{multline}\label{sec:energy-fluid.eq:base-energy-prf-4}
		\pa_\tau E_f + \frac{2-\frac{3}{n}}{\tau}\int_{\TT^3} \tau^{\frac{6}{n}}\chi^2 v^0 g_{\alpha\beta}\xi^\alpha \xi^\beta\\
		= \int_{\TT^3}p'(\rho)\varphi Q(\varphi,\xi) + \int_{\TT^3}\tau^{\frac{6}{n}}\chi^2 g_{\alpha\beta}\xi^\alpha R^\beta(\varphi,\xi) + \int_{\TT^3} \chi p'(\rho)\big(\xi^\mu - g_{\alpha\beta}\xi^\alpha \Pi^{\mu\beta}\big)\pa_\mu \varphi\\
		+\frac{1}{2}\int_{\TT^3}\pa_\mu \big(p'(\rho)v^\mu)\varphi^2 + \int_{\TT^3} \pa_\mu \big(\chi p'(\rho) \big) \varphi\xi^\mu + \int_{\TT^3} \tau^{\frac{6}{n}}\pa_\mu \big(\chi^2 g_{\alpha\beta}v^\mu\big)\xi^\alpha \xi^\beta\,.
	\end{multline}
	The third integral on the RHS of \eqref{sec:energy-fluid.eq:base-energy-prf-4} vanishes due to the condition $g(\xi,v) = 0$, since
	\begin{equation*}
		\xi^\mu - g_{\alpha\beta}\xi^\alpha \Pi^{\mu\beta} = -g_{\alpha\beta}\xi^\alpha v^\beta v^\mu = 0\,.
	\end{equation*}
	For the first and second integrals on the RHS of \eqref{sec:energy-fluid.eq:base-energy-prf-4}, we bound with the Cauchy-Schwarz inequality to obtain
	\begin{align}
        \label{sec:energy-fluid.eq:base-energy-prf-5}
		\Big|\int_{\TT^3}p'(\rho)\varphi Q(\varphi,\xi) \Big| 
        &\lesa \|\varphi\|_{L^2}\|Q(\varphi,\xi)\|_{L^2}\lesa \big(E_f[\varphi,\xi]\big)^{1/2}f_{Q,R}[\varphi,\xi]\,,\\
		\Big|\int_{\TT^3}\tau^{\frac{6}{n}}\chi^2 g_{\alpha\beta}\xi^\alpha R^\beta(\varphi,\xi)\Big| 
        &\lesa \big(\tau^{\frac{3}{n}}\|\xi\|_{L^2}\big)\Big(\tau^{\frac{3}{n}}\sum_{\mu=0}^3\|R^\mu(\varphi,\xi)\|_{L^2}\Big)\\
        &\lesa \big(E_f[\varphi,\xi]\big)^{1/2}f_{Q,R}[\varphi,\xi] + \frac{1}{\tau^{1+\frac{3}{2n}-\delta}}E_f[\varphi,\xi]\,,\nonumber
	\end{align}
    where in the final inequality, we used the identity \eqref{sec:energy-fluid.eq:R0-identity} to bound $R^0$ in $L^2$.
    
	The final three integrals in \eqref{sec:energy-fluid.eq:base-energy-prf-4} are all error terms that can be controlled by $\tau^{-(1+\delta)}E_f$. Using a combination of Cauchy-Schwarz, the Sobolev embedding, and Proposition~\ref{sec:stability.prop:fluid-time-deriv} to estimate material derivatives of $(\rho,v)$, we find that 
    \begin{align}
        \Big|\int_{\TT^3} \pa_\mu \big(p'(\rho) v^\mu\big)\varphi^2 \Big| &\lesa \big(\|\pav \rho\|_{L^\infty} + \|\pa_\mu v^\mu\|_{L^\infty}\big)\|\varphi\|_{L^2}^2 \lesa \frac{1}{\tau^{1+\frac{3}{2n}-\delta}}E_f[\varphi,\xi]\,,\\
        \Big|\int_{\TT^3} \pa_\mu \big(\chi p'(\rho) \big) \varphi\xi^\mu\Big| 
        &\lesa \frac{1}{\tau^{\frac{3}{n}}}\|\pa \rho\|_{L^\infty}\|\varphi\|_{L^2}\big(\tau^{\frac{3}{n}}\|\xi\|_{L^2}\big) \lesa \frac{1}{\tau^{1+\frac{3}{2n}-\delta}}E_f[\varphi,\xi]\,,\\
        \label{sec:energy-fluid.eq:base-energy-prf-6}
        \Big|\int_{\TT^3} \tau^{\frac{6}{n}}\pa_\mu \big(\chi^2 g_{\alpha\beta}v^\mu\big)\xi^\alpha \xi^\beta 
        \Big| 
        &\lesa \big(\|\pav \rho\|_{L^\infty} + \|\pa_\mu v^\mu\|_{L^\infty} + \|\pav g\|_{L^\infty}\big)\tau^{\frac{6}{n}}\|\xi\|_{L^2}^2 \lesa \frac{1}{\tau^{1+\frac{3}{2n}-\delta}}E_f[\varphi,\xi]\,.
    \end{align}
    In the first and third error term above, we use the improved decay of $\pav \rho$ to derive sufficient decay. For the second error term, the bad decay of the term $\|\pa \rho\|_{L^\infty}$ is compensated for by the additional $\tau^{\frac{6}{n}}$ weight possessed by the kinetic part of the energy $E_f[\varphi,\xi]$ (the part that controls $\|\xi\|_{L^2}^2$). Inserting the estimates \eqref{sec:energy-fluid.eq:base-energy-prf-5}--\eqref{sec:energy-fluid.eq:base-energy-prf-6} into the identity \eqref{sec:energy-fluid.eq:base-energy-prf-4} implies \eqref{sec:energy-fluid.eq:base-energy}.
	\end{proof}
\subsection{Fluid energy correction mechanisn}
In analogy to the wave energy estimates of the previous section, the base fluid energy $E_f$ does not, on its own, satisfy an energy inequality leading to decay. This is because the term
    \begin{equation*}
        \frac{2-\frac{3}{n}}{\tau}\int_{\TT^3} \tau^{\frac{6}{n}}\chi^2 v^0 g_{\alpha\beta} \xi^\alpha \xi^\beta\,,
    \end{equation*}
on the left hand side of the differential inequality \eqref{sec:energy-fluid.eq:base-energy}, does not control the entire energy. Like with the wave energy used to estimate the metric perturbations, we overcome this with the combination of lower-order energy correction terms, as well as a commutator argument. For the fluid energy, the commutator argument is simpler, as we only need to commute with spatial derivatives not material derivatives. We delay discussion of this until the next subsection, and focus now on the energy correction mechanism.
In a similar spirit to the correction mechanism introduced in \cite{Fajetal:decelEuler1:25, Fajetal:decelEuler2:25}, we introduce the correction term
\begin{equation}\label{sec:energy-fluid.eq:correction-def}
    \frac{1}{\tau}\int_{\TT^3}\tau^{\frac{6}{n}}\chi (v^0)^2 \xi^a \pa_a \varphi\,,
\end{equation}
and first prove a differential inequality for this integral.
\begin{lemma}\label{sec:energy-fluid.lem:correction-term}
    Suppose the scalar function $\varphi$, and the spatial components of the vectorfield $\xi$ have vanishing spatial average. Then the correction term satisfies the differential inequality
    \begin{multline}\label{sec:energy-fluid.eq:correction-term}
        \pa_\tau \Big(\frac{1}{\tau}\int_{\TT^3}\tau^{\frac{6}{n}}\chi (v^0)^2 \xi^a \pa_a \varphi \Big) +\frac{1}{\tau}\int_{\TT^3} g^{ab}p'(\rho)v^0 \pa_a \varphi \pa_b \varphi - \frac{1}{\tau}\int_{\TT^3}\tau^{\frac{6}{n}}\chi^2 v^0 (\pa_a \xi^a)^2\\
        \lesa \frac{1}{\tau^{1-\frac{3}{n}}}\big(E_{\rho,v,1}[\varphi,\xi]\big)^{1/2}f_{Q,R;1}[\varphi,\xi]+\frac{1}{\tau^{1+\frac{3}{2n}-\delta}}E_f[\pa_x\varphi,\pa_x\xi]\,,
    \end{multline}
    where $f_{Q,R;1}$ was given in Definition \ref{def:fluid-ops-QR}.
\end{lemma}
\begin{proof}
	To expedite some of the computations that follow, we will label any terms that are obvious error terms controlled by $\tau^{-(1++\frac{3}{2n}-\delta)}E_f[\pa_x\varphi,\pa_x \xi]$ with $(i)$.
    
    Differentiating the correction term in time gives
		\begin{multline}\label{sec:energy-fluid.eq:correction-term-proof-1}
		\pa_\tau \Big(\frac{1}{\tau}\int_{\TT^3} \tau^{\frac{6}{n}}\chi (v^0)^2 \xi^a \pa_a \varphi \Big)\\
		=  \frac{1}{\tau}\int_{\TT^3} \tau^{\frac{6}{n}} \chi v^0 \xi^a \pa_a (v^0\pa_\tau \varphi) + \frac{1}{\tau} \int_{\TT^3} \tau^{\frac{6}{n}} \chi (v^0)^2 \pa_\tau \xi^a \pa_a \varphi - \frac{1}{\tau}\int_{\TT^3} \tau^{\frac{6}{n}} \chi \xi^a \pa_a v^0 v^0\pa_\tau \varphi\\ - \underbrace{\frac{1-\frac{6}{n}}{\tau^{2-\frac{6}{n}}}\int_{\TT^3} \chi (v^0)^2 \xi^a \pa_a \varphi}_{(i)} + \underbrace{\frac{1}{\tau}\int_{\TT^3} \tau^{\frac{6}{n}}\pa_\tau \big(\chi (v^0)^2\big) \xi^a\pa_a \varphi}_{(i)}\,.
	\end{multline}
    For integrals that are not error terms, we use the definitions of the operators $Q$, $R^i$, and integrate by parts to write
    \begin{subequations}
	\begin{align}
        \label{sec:energy-fluid.eq:correction-term-proof-2}\
		\frac{1}{\tau}\int_{\TT^3} \tau^{\frac{6}{n}} \chi v^0 \xi^a \pa_a (v^0 \pa_\tau \varphi) 
        =&\, \frac{1}{\tau}\int_{\TT^3} \tau^{\frac{6}{n}}\chi v^0\xi^a \pa_a \big(Q(\varphi,\xi)\big)-\frac{1}{\tau}\int_{\TT^3} \tau^{\frac{6}{n}} \chi v^0 \xi^a \pa_a (\chi\pa_\mu\xi^\mu)\\
        &-\underbrace{\frac{1}{\tau}\int_{\TT^3} \tau^{\frac{6}{n}} \chi v^0 \xi^a \pa_a v^b \pa_b \varphi}_{(i)} + \underbrace{\frac{1}{\tau}\int_{\TT^3} \tau^{\frac{6}{n}}\pa_a\big(\chi v^0 v^b \xi^a)\pa_b \varphi}_{(i)}\,,\nonumber\\
        \frac{1}{\tau}\int_{\TT^3} \tau^{\frac{6}{n}} \chi (v^0)^2 \pa_\tau \xi^a \pa_a \varphi 
        =&\, \frac{1}{\tau}\int_{\TT^3} \tau^{\frac{6}{n}} \chi v^0 \pa_a \varphi R^a(\varphi,\xi) -\frac{1}{\tau}\int_{\TT^3} p'(\rho)v^0\pa_a \varphi \Pi^{a\mu}\pa_\mu \varphi\\
        &- \frac{1}{\tau}\underbrace{\int_{\TT^3} \tau^{\frac{6}{n}} \chi v^0 v^b\pa_b \xi^a \pa_a \varphi}_{(i)}-\underbrace{\frac{2}{\tau^2} \int_{\TT^3} \tau^{\frac6n}\chi v^0 \pa_a \varphi v^0 \xi^a}_{(i)}\,,\nonumber\\
        \frac{1}{\tau}\int_{\TT^3} \tau^{\frac{6}{n}} \chi \xi^a \pa_a v^0 v^0\pa_\tau \varphi 
        =&\, \frac{1}{\tau}\int_{\TT^3} \tau^{\frac{6}{n}} \chi \xi^a \pa_a v^0 Q(\varphi,\xi)-\underbrace{\frac{1}{\tau}\int_{\TT^3} \tau^{\frac{6}{n}} \chi^2 \xi^a \pa_a v^0 \pa_\mu \xi^\mu}_{(i)}\\
        &-\underbrace{\frac{1}{\tau}\int_{\TT^3} \tau^{\frac{6}{n}} \chi \xi^a \pa_a v^0 v^b\pa_b \varphi}_{(i)}\,.\nonumber
	\end{align}
    \end{subequations}
    For all terms that contain $Q(\varphi,\xi)$, $R^i(\varphi,\xi)$, we bound using Cauchy-Schwarz and the Poincar\'e inequality to get 
    \begin{align}
        \frac{1}{\tau}\Big|\int_{\TT^3} \tau^{\frac{6}{n}}\chi v^0\xi^a \pa_a \big(Q(\varphi,\xi)\big)\Big| &\lesa \frac{1}{\tau^{1-\frac{3}{n}}}\big(E_f[\pa_x\varphi,\pa_x\xi]\big)^{1/2}f_{Q,R;1}[\varphi,\xi]\,,\\
        \frac{1}{\tau}\Big|\int_{\TT^3} \tau^{\frac{6}{n}} \chi v^0 \pa_a \varphi R^a(\varphi,\xi)\Big| &\lesa \frac{1}{\tau^{1-\frac{3}{n}}}\big(E_f[\pa_x\varphi,\pa_x\xi]\big)^{1/2}f_{Q,R;1}[\varphi,\xi]\,,\\
        \frac{1}{\tau}\Big|\int_{\TT^3} \tau^{\frac{6}{n}} \chi \xi^a \pa_a v^0 Q(\varphi,\xi)\Big| &\lesa \frac{1}{\tau^{2-\frac{3}{2n}-\delta}}\big(E_f[\pa_x\varphi,\pa_x\xi]\big)^{1/2}f_{Q,R;1}[\varphi,\xi]\,.
    \end{align}
    For the remaining two terms, we integrate by parts: 
    \begin{align}
        -\frac{1}{\tau}\int_{\TT^3} \tau^{\frac{6}{n}} \chi v^0 \xi^a \pa_a (\chi\pa_\mu\xi^\mu) =&\, \frac{1}{\tau}\int_{\TT^3} \tau^{\frac{6}{n}} \chi^2 v^0 (\pa_a \xi^a)^2\\
        &+ \frac{1}{\tau}\int_{\TT^3} \tau^{\frac{6}{n}} \chi^2 v^0 \pa_a \xi^a \pa_\tau \xi^0 + \frac{1}{\tau}\int_{\TT^3} \tau^{\frac{6}{n}} \pa_a \big(\chi v^0\big) \chi \xi^a \pa_\mu \xi^\mu\,, \nonumber\\
        \label{sec:energy-fluid.eq:correction-term-proof-3}
        -\frac{1}{\tau}\int_{\TT^3} p'(\rho)v^0\pa_a \varphi \Pi^{a\mu}\pa_\mu \varphi =& - \frac{1}{\tau}\int_{\TT^3} p'(\rho) v^0 g^{ab}\pa_a \varphi \pa_b \varphi \\
        &-\frac{1}{\tau}\int_{\TT^3} p'(\rho) v^0 v^a \pa_a \varphi v^\mu \pa_\mu \varphi - \frac{1}{\tau}\int_{\TT^3} p'(\rho) v^0 g^{a0}\pa_a \varphi \pa_\tau \varphi\,.\nonumber
    \end{align}
    The first integral on the right hand side of these identities are precisely the extra terms on the left hand side of \eqref{sec:energy-fluid.eq:correction-term}. The other integrals can all be bounded by the right hand side of \eqref{sec:energy-fluid.eq:correction-term}, using Lemma~\ref{sec:energy-fluid.lem:time-deriv} if necessary to estimate time derivatives of $(\varphi,\xi)$. Thus, substituting in the various estimates and bounds \eqref{sec:energy-fluid.eq:correction-term-proof-2}--\eqref{sec:energy-fluid.eq:correction-term-proof-3} into the identity \eqref{sec:energy-fluid.eq:correction-term-proof-1} and rearranging, we arrive at the estimate \eqref{sec:energy-fluid.eq:correction-term}.
\end{proof}

\subsection{The corrected base fluid energy}
We are now ready to introduce the corrected energy which we use to estimate higher-order derivatives of $\rho$ and $v$.
\begin{definition}[Corrected base fluid energy]\label{def:corrected-base-fluid-energy}
    For fixed $\eta \in \RR$, we define the corrected fluid energy
\begin{multline}\label{sec:energy-fluid.eq:corrected-energy-def}
    E_{f, \eta; 1}[\varphi,\xi](\tau) := \frac{1}{2}\int_{\TT^3} g^{ij}\Big(p'(\rho) v^0 \pa_i\varphi\pa_j \varphi + 2\chi p'(\rho) \pa_i\varphi \pa_j\xi^0 +  \tau^{\frac{6}{n}} \chi^2 v^0 g_{\alpha\beta} \pa_i\xi^\alpha \pa_j\xi^\beta\Big)\\
    +\frac{\eta}{\tau} \int_{\TT^3}\tau^{\frac{6}{n}}\chi (v^0)^2 \xi^a \pa_a \varphi\,.
\end{multline}
\end{definition}

One should think of this energy as the base fluid energy applied to the quantities $\pa_x\varphi$, $\pa_x \varphi$, \emph{plus} a multiple of the fluid energy correction term. The functional $E_{f, \eta; 1}$ is still coercive, subject to several conditions. If $\varphi$ and the spatial components of $\xi$ have vanishing spatial average, and if $n > 3$, then by Poincar\'e, the correction term obeys the bound
\begin{equation*}
    \Big|\frac{1}{\tau}\int_{\TT^3} \tau^{\frac{6}{n}}\chi (v^0)^2 \xi^ i \pa_i \varphi\Big| \lesa \frac{1}{\tau^{1-\frac{3}{n}}}\|\pa_x \varphi\|_{L^2} \cdot \tau^{\frac{3}{n}}\|\pa_x \xi\|_{L^2} \lesa \frac{1}{\tau^{1-\frac{3}{n}}}\big(\|\pa_x \varphi\|_{L^2}^2 + \|\pa_x \xi\|_{L^2}^2\big)\,,
\end{equation*}
If in addition $\xi$ satisfies $g(\xi,v) = 0$, then by Lemma~\ref{sec:energy-fluid.lem:coercivity}, we have
\begin{equation*}
    E_{f, \eta; 1}[\varphi,\xi](\tau) \sim \|\pa_x \varphi\|_{L^2}^2 + \tau^{\frac{6}{n}}\|\pa_x \xi\|_{L^2}^2\,.
\end{equation*}
We note that this is the primary place in our argument where the assumption $n > 3$ on the polytropic index is invoked; see Remark~\ref{sec:intro.rmk:polytropic-index}.

We now prove the energy estimate satisfied by $E_{f, \eta; 1}$.
\begin{proposition}[Estimate for corrected fluid energy]\label{sec:energy-fluid.prop:corrected-energy}
    Suppose that $\varphi$ and the spatial components of $\xi$ have vanishing spatial average, and that $g(\xi,v) = 0$. If $\eta = 1 - \frac{3}{2n}$, then the corrected fluid energy $E_{f, \eta; 1}[\varphi,\xi]$ satisfies the differential inequality
    \begin{equation}\label{sec:energy-fluid.eq:corrected-energy}
        \pa_\tau E_{f, \eta; 1}[\varphi,\xi](\tau) + \frac{2\big(1-\frac{3}{2n}\big)}{\tau}E_{f, \eta; 1}[\varphi,\xi]
        \lesa \big(E_{f, \eta; 1}[\varphi,\xi]\big)^{1/2}f_{Q,R;1}[\varphi,\xi] +\frac{1}{\tau^{1+\frac{3}{2n}-\delta}}E_{f, \eta; 1}[\varphi,\xi]\,.
    \end{equation}
\end{proposition}
\begin{proof}
    We begin by letting $\eta\in \RR$ be an arbitrary real number. Repeating the argument from Lemma~\ref{sec:energy-fluid.lem:base-energy}, we find that the integral
    \begin{equation}
        E_{f, 0; 1}[\varphi,\xi] = \frac{1}{2}\int_{\TT^3} g^{ij}\Big(p'(\rho) v^0 \pa_i\varphi\pa_j \varphi + 2\chi p'(\rho) \pa_i\varphi \pa_j\xi^0 +  \tau^{\frac{6}{n}} \chi^2 v^0 g_{\alpha\beta} \pa_i\xi^\alpha \pa_j\xi^\beta\Big)\,,
    \end{equation}
    satisfies the differential inequality
    \begin{multline}\label{sec:energy-fluid.eq:corrected-energy-proof-1}
         \frac{1}{2}\pa_\tau E_{f, 0; 1}[\varphi,\xi](\tau)
         + \frac{2-\frac{3}{n}}{\tau}\int_{\TT^3}\tau^{\frac{6}{n}}g^{ij} \chi^2 v^0 g_{\alpha\beta}\pa_i \xi^\alpha \pa_j\xi^\beta \\
         \lesa \big(E_{f, \eta; 1}[\varphi,\xi](\tau)\big)^{1/2}f_{Q,R}[\pa_x\varphi,\pa_x\xi]+\frac{1}{\tau^{1+\frac{3}{2n}-\delta}}E_{f, \eta; 1}[\varphi,\xi](\tau)\,.
    \end{multline}
    Any additional error terms which may arise from derivatives hitting the $g^{ij}$ quantities in the integrand of $E_{f, 0; 1}$ will be integrable, as we have $\|\pa g\|_{L^\infty} \lesa \tau^{-(1+\frac{3}{n})}$. Combining \eqref{sec:energy-fluid.eq:corrected-energy-proof-1} with the estimate for the correction term from Lemma~\ref{sec:energy-fluid.lem:correction-term} implies the inequality 
    \begin{multline}\label{sec:energy-fluid.eq:corrected-energy-proof-2}
        \pa_\tau E_{f, \eta; 1}[\varphi,\xi] + \frac{\eta}{\tau}\int_{\TT^3} g^{ab}p'(\rho)v^0 \pa_a \varphi \pa_b \varphi \\
        + \frac{1}{\tau}\int_{\TT^3} \tau^{\frac{6}{n}}\chi^2 v^0 \Big[\big(2-\frac{3}{n}\big)g^{ab}g_{\alpha\beta}\pa_a \xi^\alpha \pa_b \xi^\beta - \eta(\pa_a\xi^a)^2\Big]\\
        \lesa \big(E_{f, \eta; 1}[\varphi,\xi]\big)^{1/2}\Big(f_{Q,R}[\pa_x\varphi,\pa_x\xi] + f_{Q,R;1}[\varphi,\xi]\Big)+\frac{1}{\tau^{1+\frac{3}{2n}-\delta}}E_{f, \eta; 1}[\varphi,\xi]\,.
    \end{multline}
    We focus on the third term in this inequality, the final integral on the left hand side. Since $g(\xi,v) = 0$, the $0$-component of $\xi$ is given by $\xi^0 = -\frac{v_a}{v_0}\xi^a$, therefore we have the upper bound
    \begin{equation*}
        \frac{1}{\tau}\int_{\TT^3} \tau^{\frac{6}{n}}\chi^2 v^0 g^{ab}g_{\alpha\beta}\pa_a \xi^\alpha \pa_b \xi^\beta  \leq \frac{1}{\tau}\int_{\TT^3} \tau^{\frac{6}{n}}\chi^2 v^0 g^{ab}g_{cd}\pa_a \xi^c \pa_b \xi^d  + \frac{C}{\tau^{1+\frac{3}{2n}-\delta}}E_{f, \eta; 1}[\varphi,\xi]\,.
    \end{equation*}
    Hence we may replace the final term on the left hand side of the inequality \eqref{sec:energy-fluid.eq:corrected-energy-proof-1} with
    \begin{equation}\label{sec:energy-fluid.eq:corrected-energy-proof-3}
        \frac{1}{\tau}\int_{\TT^3} \tau^{\frac{6}{n}}\chi^2 v^0 \Big[\big(2-\frac{3}{n}\big)g^{ab}g_{cd}\pa_a \xi^c \pa_b \xi^d - \eta(\pa_a\xi^a)^2\Big]\,.
    \end{equation}
    It is clear that this term will be nonnegative and control $\tau^{-1}\|\pa_x \xi\|_{L^2}^2$ for sufficiently small $\eta$. We can make this observation precise by a div-curl identity. We introduce the quantities $\mathrm{div}\,\xi$, $\mathrm{curl}\,\xi$\footnote{These are not quite the standard divergence and curl operators, however we label them in the same manner, as we use them for the same purpose.}
    \begin{equation*}
        \mathrm{div}\,\xi := \pa_a \xi^a\,, \qquad \mathrm{curl}_{ij} \,\xi:= \frac{1}{2}\Big(g_{ai}\pa_j \xi^a - g_{aj}\pa_i \xi^a\Big)\,.
    \end{equation*}
    We claim that if the components of $\xi$ have vanishing spatial average, then $\mathrm{div}\, \xi$, $\mathrm{curl}\,\xi$ satisfy the identity
    \begin{equation}\label{sec:energy-fluid.eq:corrected-energy-proof-4}
        \int_{\TT^3} g^{ab}g_{cd} (\pa_a \xi^c)(\pa_b \xi^d)= \int_{\TT^3} (\mathrm{div}\,\xi)^2 + 2\int_{\TT^3} g^{ab}g^{cd} (\mathrm{curl}_{ac} \,\xi)(\mathrm{curl}_{bd} \,\xi) + \mc{J}\,,
    \end{equation}
    where $\mc{J}$ is an error term defined in \eqref{app:divcurl-proof-5} that satisfies 
    \begin{equation*}
        |\mc{J}| \lesa \Big(\|\pa_x h\|_{L^\infty} + \sum_{i=1}^3\|h^{0i}\|_{L^\infty}^2\Big)\|\pa_x \xi\|_{L^2}^2\,.
    \end{equation*}
    We defer the proof of the identity \eqref{sec:energy-fluid.eq:corrected-energy-proof-4} to Appendix~\ref{app:divcurl}. Observing the nonnegativity of the $\mathrm{curl}$ integral \eqref{sec:energy-fluid.eq:corrected-energy-proof-4}, as well as the aforementioned bound on $\mc{J}$, we conclude that
    \begin{multline*}
        \frac{1}{\tau}\int_{\TT^3} \tau^{\frac{6}{n}}\chi^2 v^0(\mathrm{div}\,\xi)^2 \\
        \leq \frac{1}{\tau}\int_{\TT^3} \tau^{\frac{6}{n}}\chi^2 v^0 g^{ab}g_{cd} (\pa_a \xi^c)(\pa_b \xi^d)+ \frac{C}{\tau^{1-\frac{6}{n}}}\Big(\|\pa_x \rho\|_{H^2} + \|\pa_x v\|_{H^2} + \|\pa_x g\|_{H^2}\Big)\|\pa_x \xi\|_{L^2}^2\\
        \leq \frac{1}{\tau}\int_{\TT^3} \tau^{\frac{6}{n}}\chi^2 v^0 g^{ab}g_{cd} (\pa_a \xi^c)(\pa_b \xi^d) + \frac{C\ve}{\tau^{2-\frac{3}{2n}-\delta}}E_{f, \eta; 1}[\varphi,\xi]\,,
    \end{multline*}
    where we dealt with the $\chi^2v^0$ terms by decomposing like $\chi^2v^0 = (\chi^2 v^0)_\av + \big[\chi^2v^0- (\chi^2 v^0)_\av \big]$, then bounding the remainder term in $L^\infty$ by $\|\pa_x\rho\|_{H^2} + \|\pa_x v\|_{H^2}$ with the Sobolev embedding and Poincar\'e inequality. Thus the term \eqref{sec:energy-fluid.eq:corrected-energy-proof-4} satisfies the lower bound
    \begin{multline*}
        \frac{1}{\tau}\int_{\TT^3} \tau^{\frac{6}{n}}\chi^2 v^0 \Big[\big(2-\frac{3}{n}\big)g^{ab}g_{cd}\pa_a \xi^c \pa_b \xi^d - \eta(\pa_a\xi^a)^2\Big] \\
        \geq \frac{2 - \frac{3}{n} -\eta}{\tau}\int_{\TT^3} \tau^{\frac{6}{n}}\chi^2 v^0 g_{\alpha\beta}g^{ab}\pa_a \xi^\alpha \pa_b \xi^\beta - \frac{C}{\tau^{2-\frac{3}{2n}-\delta}}E_{f, \eta; 1}[\varphi,\xi]\,.
    \end{multline*}
    Inserting this into \eqref{sec:energy-fluid.eq:corrected-energy-proof-2}, we arrive at the estimate
    \begin{multline}\label{sec:energy-fluid.eq:corrected-energy-proof-6}
        \pa_\tau E_{f, \eta; 1}[\varphi,\xi] + \frac{2\min\Big\{\eta,2-\frac{3}{n}-\eta\Big\}}{\tau}E_{f, \eta; 1}[\varphi,\xi] \\
        \lesa \big(E_{f, \eta; 1}[\varphi,\xi]\big)^{1/2}\Big(f_{Q,R}[\pa_x\varphi,\pa_x\xi] + f_{Q,R;1}[\varphi,\xi]\Big)+\frac{1}{\tau^{1+\frac{3}{2n}-\delta}}E_{f, \eta; 1}[\varphi,\xi]\,.
    \end{multline}
    It is now clear that the coefficient of the damping term is maximised by setting $\eta = 1-\frac{3}{2n}$. 
    
    It remains to bound the quantity $f_{Q,R}[\pa_x\varphi,\pa_x\xi]$ in terms of $f_{Q,R;1}[\varphi,\xi]$. The operators $Q$, $R$ satisfy the commutator relations
    \begin{align*}
        Q(\pa_i\varphi,\pa_i\xi) =&\, \pa_i Q(\varphi,\xi) -\pa_i v^\mu \pa_\mu \varphi - \pa_i \chi \pa_\mu \xi^\mu\,,\\
        R^j(\pa_i\varphi,\pa_i\xi) =&\, \pa_i R^j(\varphi,\xi) -\pa_i v^\mu \pa_\mu \xi^j - \pa_i\Big[\frac{p'(\rho)}{\tau^{\frac{6}{n}}\chi}\Pi^{\mu i}\Big]\pa_\mu\varphi - \frac{2}{\tau}\pa_i v^0 \xi^j\,,
    \end{align*}
    and so using Lemma~\ref{sec:energy-fluid.lem:time-deriv} to estimate the time derivatives of $\varphi,\xi$, we have
    \begin{multline*}
        \sum_{i=1}^3\Big(\big\|Q(\pa_i\varphi,\pa_i\xi)\big\|_{L^2} + \tau^{\frac{3}{n}}\sum_{j=1}^3\big\|R^j(\pa_i \varphi,\pa_i \xi)\big\|_{L^2}\Big) \\\lesa f_{Q,R;1}[\varphi,\xi] + \frac{1}{\tau^{1+\frac{3}{2n}-\delta}}\|\pa_x\varphi\|_{L^2} + \frac{1}{\tau^{1-\frac{3}{2n}-\delta}}\|\pa_x \xi\|_{L^2} \lesa f_{Q,R;1}[\varphi,\xi] + \frac{1}{\tau^{1+\frac{3}{2n}-\delta}}E_{f, \eta; 1}[\varphi,\xi]\,.
    \end{multline*}
    We insert this estimate into \eqref{sec:energy-fluid.eq:corrected-energy-proof-6} to obtain the stated bound.
\end{proof}

\subsection{Higher-order estimates for the fluid variables}
Next, we apply the estimate from Proposition~\ref{sec:energy-fluid.prop:corrected-energy} to the higher-order spatial derivatives of the fluid variables
\begin{equation*}
    (\pa_x^I \rho,\pa_x^I v)\,,\qquad 1\leq |I|\leq N\,.
\end{equation*}
\begin{definition}[Top-order fluid energy]
    Let $v^I$ denotes the top-order spatial derivatives of the fluid four-velocity, with modified $0$-component
\begin{equation}
    v^I = \Big((v^I)^0,(v^I)^1,(v^I)^2,(v^I)^3\Big) := \Big(-\frac{v_a}{v_0}\pa_x^I v^a,\pa_x^I v^1,\pa_x^I v^2,\pa_x^I v^3\Big)\,.
\end{equation}
Then we use the corrected base fluid energy from Definition \ref{def:corrected-base-fluid-energy} to  define the higher-order fluid energy 
\begin{equation}
    E_{(\rho, v); N}(\tau) := \sum_{0\leq|I|\leq N-1} E_{f,(1-\frac{3}{2n}) ;1}\big[\pa_x^I \rho,v^I\big](\tau)\,.
\end{equation}
\end{definition}

The above choice of $v^I$ ensures that $g(v^I,v) = 0$, so that the previous results in this section hold. Since the spatial components of the $v^I$ have vanishing spatial average for $|I| \geq 1$, it follows from the Poincar\'e inequality that the energy $E_{(\rho,v);N}$ is equivalent (up to a weight in $\tau$) to the higher-order norm of the fluid:
\begin{equation}\label{sec:energy-fluid.eq:higher-energy-bound} E_{(\rho,v);N}(\tau) \sim\tau^{-\big(2-\frac{3}{n}-2\delta\big)}S_{(\rho,v);N}^2(\tau)\,.
\end{equation}

We have the following energy estimate for $E_{(\rho,v);N}$. Just as in Proposition~\ref{sec:energy-met.prop:higher-energy}, we rely on several commutator estimates which we will prove later in the section.
\begin{proposition}[Estimate for the higher-order fluid energy]\label{sec:energy-fluid.prop:higher-energy}
    Suppose $(\rho,v)$ solve the Euler equations \eqref{sec:gauged-eqs.eq:fluid}, and that the bootstrap assumptions \eqref{sec:stability.eq:boot} hold. Then the time derivative of the energy $E_{(\rho,v);N}$ satisfies
    \begin{multline}\label{sec:energy-fluid.eq:higher-energy}
        \pa_\tau E_{(\rho,v);N}(\tau) + \frac{2-\frac{3}{n}}{\tau}E_{(\rho,v);N}\\\lesa \tau^{\frac{3}{n}}\big[(E_{h^{00};N-1}^{\mathbf{v}})^{1/2} + (E_{h^{0*};N-1}^{\mathbf{v}})^{1/2}\big](E_{(\rho,v);N})^{1/2} + \frac{1}{\tau^{3-\frac{3}{2n}-3\delta}}\,.
    \end{multline}
\end{proposition}
\begin{proof}
    The quantities $\pa_x^I \rho$, $(v^I)^i$ have vanishing spatial average for $1 \leq |I| \leq N$. Moreover, $g(v^I,v) = 0$. Thus we may apply Proposition~\ref{sec:energy-fluid.prop:corrected-energy} to estimate the quantities $(\pa_x^I\rho,v^I)$. Using the definition of the energy $E_{(\rho,v);N}$, this implies the bound
    \begin{equation}\label{sec:energy-fluid.eq:higher-energy-proof-1}
        \pa_\tau E_{(\rho,v);N}+ \frac{2-\frac{3}{n}}{\tau}E_{(\rho,v);N}\lesa \sum_{1\leq |I|\leq N-1}f_{Q,R;1}[\pa_x^I\rho,v^I](E_{(\rho,v);N})^{1/2} + \frac{1}{\tau^{3-\frac{3}{2n}-3\delta}}\,.
    \end{equation}
    We claim that the following bound holds:
    \begin{equation}\label{sec:energy-fluid.eq:higher-energy-proof-2}
        \sum_{1\leq |I|\leq N-1}f_{Q,R;1}[\pa_x^I \varphi,v^I] \lesa \tau^{\frac{3}{n}}\big[(E_{h^{00};N-1}^{\mathbf{v}})^{1/2} + (E_{h^{0*};N-1}^{\mathbf{v}})^{1/2}\big] + \frac{1}{\tau^{2-2\delta}}\,.
    \end{equation}
    We will prove this in Proposition~\ref{sec:energy-fluid.prop:com-ests}. For now, we plug \eqref{sec:energy-fluid.eq:higher-energy-proof-2} into \eqref{sec:energy-fluid.eq:higher-energy-proof-1} and use the bound \eqref{sec:energy-fluid.eq:higher-energy-bound} to estimate
    \begin{align*}
        \pa_\tau &E_{(\rho,v);N} + \frac{2-\frac{3}{n}}{\tau}E_{(\rho,v);N}\\
        &\lesa \tau^{\frac{3}{n}}\big[(E_{h^{00};N-1}^{\mathbf{v}})^{1/2} + (E_{h^{0*};N-1}^{\mathbf{v}})^{1/2}\big](E_{(\rho,v);N})^{1/2} + \frac{1}{\tau^{2-2\delta}}(E_{(\rho,v);N})^{1/2} + \frac{1}{\tau^{1+\frac{3}{2n}-\delta}}E_{(\rho,v);N}\\
        &\lesa \tau^{\frac{3}{n}}\big[(E_{h^{00};N-1}^{\mathbf{v}})^{1/2} + (E_{h^{0*};N-1}^{\mathbf{v}})^{1/2}\big](E_{(\rho,v);N})^{1/2} + \frac{1}{\tau^{3-\frac{3}{2n}-3\delta}}\,.
    \end{align*}
    The result now follows.
\end{proof}

It remains to prove the bound \eqref{sec:energy-fluid.eq:higher-energy-proof-2}. First, we have estimates for the source terms appearing in the Euler equations.
\begin{lemma}[Estimates for the fluid source terms]\label{sec:energy-fluid.lem:source-ests}
    The quantities $\mc{L}_{(\rho,v)}$, $\mc{L}_{(\rho,v)}^i$, $\mc{G}_{(\rho,v)}$, and $\mc{G}_{(\rho,v)}^i$ obey the $H^N$-norm bounds
    \begin{subequations}\label{sec:energy-fluid.eq:source-ests-1}
        \begin{align}
            \|Q(\rho,v)\|_{H^N} &\lesa \frac{1}{\tau^{1+\frac{3}{n}}},\\
            \sum_{i=1}^3\|R^i(\rho,v)\|_{H^N} &\lesa \frac{1}{\tau^{2+\frac{3}{2n}-\delta}}\,,
        \end{align}
    \end{subequations}
    and the higher-order bounds
    \begin{subequations}\label{sec:energy-fluid.eq:source-ests-2}
        \begin{align}
            \big\|\pa_x[Q(\rho,v)]\|_{H^{N-1}} &\lesa \frac{1}{\tau^{2+\frac{3}{2n}-\delta}},\\
            \sum_{i=1}^3\big\|\pa_x [R^i(\rho,v)\big]\big\|_{H^{N-1}} &\lesa \big[(E_{h^{00};N-1}^{\mathbf{v}})^{1/2} + (E_{h^{0*};N-1}^{\mathbf{v}})^{1/2}\big] + \frac{1}{\tau^{2+\frac{9}{2n}-\delta}}\,.
        \end{align}
    \end{subequations}
\end{lemma}
\begin{proof}
    We compute for the $H^N$ bounds:
        \begin{align*}
        \|\mc{L}_{(\rho,v)}\|_{H^N} &\lesa \|\Gamma\|_{H^N} \leq \frac{1}{\tau^{1+\frac{3}{n}}}\,,\\
        \sum_{i=1}^3\|\mc{L}_{(\rho,v)}^i\|_{H^N} &\lesa \frac{1}{\tau}\sum_{i=1}^3 \|h^{0i}\|_{H^N} +  \sum_{i=1}^3\|\Gamma_{00}^i\|_{H^N} \lesa \frac{1}{\tau^{2+\frac{3}{2n}-\delta}},\\
        \|\mc{G}_{(\rho,v)}\|_{H^N} &\lesa \|v\|_{H^N}\|\Gamma\|_{H^N} +\frac{1}{\tau^{1+\frac{6}{n}}}\|p\|_{H^N} \lesa \frac{1}{\tau^{1+\frac{3}{n}}}\,,\\
        \sum_{i=1}^3\|\mc{G}_{(\rho,v)}^i\|_{H^N} &\lesa \|v\|_{H^N}\|\Gamma\|_{H^N} + \frac{1}{\tau^{1+\frac{6}{n}}}\big\|\Pi^{0i}\big\|_{H^N} \lesa \frac{1}{\tau^{2+\frac{9}{2n}-\delta}}\,,
    \end{align*}
    implying \eqref{sec:energy-fluid.eq:source-ests-1}. For \eqref{sec:energy-fluid.eq:source-ests-2}, we use Lemma~\ref{sec:stability.lem:christoffel} to estimate the differentiated Christoffel symbols:
    \begin{align*}
        \big\|\pa_x \mc{L}_{(\rho,v)}\big\|_{H^{N-1}} &\lesa \frac{1}{\tau^{2+\frac{3}{2n}-\delta}}\,,\\
        \sum_{i=1}^3\big\|\pa_x \mc{L}_{(\rho,v)}^i\big\|_{H^{N-1}} &\lesa  \sum_{\mu=0}^3 \big\|\pa \pa_xh^{0\mu}\big\|_{H^{N-1}} + \frac{1}{\tau^{2+\frac{9}{2n}-2\delta}}\,.
    \end{align*}
    The quantities $\mc{G}_{(\rho,v)}$, $\mc{G}_{(\rho,v)}^i$, once differentiated, become rapidly-decaying error terms like
    \begin{align*}
        \big\|\pa_x \mc{G}_{(\rho,v)}\big\|_{H^{N-1}} \lesa &\,\|v\|_{H^N}\big(\|\Gamma\|_{H^{N-1}} + \|\pa_x\Gamma\|_{H^{N-1}}\big)\\
        &+\frac{1}{\tau^{1+\frac{6}{n}}}\big(\|\pa_x \rho\|_{H^{N-1}} + \|\pa_x v\|_{H^{N-1}}\big) \lesa \frac{1}{\tau^{2+\frac{9}{2n}-\delta}}\,,\\
        \sum_{i=1}^3\big\|\pa_x \mc{G}_{(\rho,v)}^i\big\|_{H^{N-1}} \lesa&\,  \|v\|_{H^N}\big(\|\Gamma\|_{H^{N-1}} + \|\pa_x\Gamma\|_{H^{N-1}}\big)\\
        &+ \frac{1}{\tau^{1+\frac{6}{n}}}\sum_{i=1}^3\big(\|\pa_x\rho\|_{H^{N-1}}\|\Pi^{0i}\|_{H^{N-1}} + \|\pa_x \Pi^{0i}\|_{H^{N-1}}\big)\lesa \frac{1}{\tau^{2+\frac{9}{2n}-\delta}}\,.
    \end{align*}
    After bounding 
    \begin{equation*}
        \sum_{\mu=0}^3\big\|\pa \pa_x h^{0\mu}\big\|_{H^{N-1}} \lesa (E_{h^{00};N-1}^{\mathbf{v}})^{1/2} + (E_{h^{0*};N-1}^{\mathbf{v}})^{1/2} + \frac{1}{\tau^{2+\frac{9}{n}-\delta}}\,,
    \end{equation*}
    we are done.
\end{proof}

We wish to apply Lemma~\ref{sec:energy-fluid.lem:source-ests} to estimate the quantities $Q(\pa_x^I \rho,v^I)$, $R^i(\pa_x^I \rho,v^I)$. Like in the previous section, this will require us to commute the spatial derivatives out of the operators $Q$, $R^i$. In order to do this, however, we will also need to compare the quantities $Q(\pa_x^I \rho,v^I)$, $R^i(\pa_x^I \rho,v^I)$ with the more natural $Q(\pa_x^I \rho,\pa_x^I v)$, $R^i(\pa_x^I \rho,\pa_x^I v)$. We have already commented that the $\pa_x^I v$ and $v^I$ are principally the same, and we now make this notion rigorous.

Since the quantities $R^i(\varrho,\xi)$ do not have a dependence on $\xi^0$, it follows immediately that
\begin{equation}\label{sec:energy-fluid.eq:R-equivalence}
    R^i(\pa_x^I \rho,v^I) \equiv R^i(\pa_x^I \rho,\pa_x^I v)
\end{equation}
for all $|I| \leq N$. However, the same identity is not true for the operator $Q$. Nonetheless, we show in the following Lemma that the quantities $Q(\pa_x^I \rho, v^I)$, $Q(\pa_x^I \rho, \pa_x^Iv)$ are principally the same, and differ by only a fast-decaying error.
\begin{lemma}\label{sec:energy-fluid.lem:four-velocity}
    For all $1 \leq |I| \leq N$, the quantities $\pa_x^I v^0 - (v^I)^0$, $\pa(\pa_x^I v^0 - (v^I)^0)$ are well-defined, and satisfy the bounds
    \begin{equation}\label{sec:energy-fluid.eq:four-velocity-1}
        \sum_{1\leq |I| + |J| \leq N+1}\Big\|\pa_x^J\big[\pa_x^I v^0 - v^I\big]\Big\|_{L^2}+ \sum_{1\leq|I| + |J| \leq N}\Big\|\pa_\tau \pa_x^J\big[\pa_x^I v^0 - v^I\big]\Big\|_{L^2} \lesa \frac{1}{\tau^{2+\frac{3}{2n}-\delta}}\,.
    \end{equation}
    Moreover, we have
    \begin{subequations}\label{sec:energy-fluid.eq:four-velocity-2}
        \begin{align}
            \sum_{1 \leq |I| \leq N}\big\|Q(\pa_x^I \rho, v^I)\big\|_{L^2} &\lesa \sum_{1 \leq |I| \leq N}\big\|Q(\pa_x^I \rho, \pa_x^Iv)\big\|_{L^2} + \frac{1}{\tau^{2+\frac{3}{2n}-\delta}}\,,\\
            \sum_{1 \leq |I| \leq N-1}\Big\|\pa_x \big[Q(\pa_x^I \rho, v^I)\big]\Big\|_{L^2} &\lesa \sum_{1 \leq |I| \leq N-1}\Big\|\pa_x \big[Q(\pa_x^I \rho, \pa_x^I v)\big]\Big\|_{L^2} + \frac{1}{\tau^{2+\frac{3}{2n}-\delta}}\,.
        \end{align}
    \end{subequations}
\end{lemma}
\begin{proof}
    Repeatedly differentiating the identity $g_{\alpha\beta}v^\alpha v^\beta = -1$, we find that for $1 \leq |I| \leq N$, $\pa_x^I v^0$ satisfies the identity
    \begin{align*}
        0 =&\, 2g_{\alpha\beta}v^\alpha \pa_x^I v^\beta + (\pa_x^I g_{\alpha\beta})v^\alpha v^\beta \\
        &+ \sum_{1 \leq J \leq I-1}c_{I,J}^{(3)}g_{\alpha\beta}(\pa_x^J v^\alpha)(\pa_x^{I-J}v^\beta) + \sum_{1\leq J \leq I-1}c_{I,J}^{(4)} (\pa_x^J g_{\alpha\beta})\pa_x^{I-J}(v^\alpha v^\beta)\\
        =&\, 2v_0 \big(\pa_x^I v^0 - (v^I)^0\big) + (\pa_x^I g_{\alpha\beta})v^\alpha v^\beta\\
        &+ \sum_{1 \leq J \leq I-1}c_{I,J}^{(3)}g_{\alpha\beta}(\pa_x^J v^\alpha)(\pa_x^{I-J}v^\beta) + \sum_{1\leq J \leq I-1}c_{I,J}^{(4)} (\pa_x^J g_{\alpha\beta})\pa_x^{I-J}(v^\alpha v^\beta)\,,
    \end{align*}
    where $c_{I,J}^{(3)}$, $c_{I,J}^{(4)}$ are coefficients computed using the Leibniz rule. Out of the four groups of terms on the right hand side of the above identity, we observe that the second is linear in spatial derivatives of $g_{00}$ (and quadratic in derivatives of $g_{i\mu}$, and $v^i$), the third is quadratic in spatial derivatives of the fluid four-velocity, and the fourth is quadratic in spatial derivatives of both $g$ and $v$. We also note that the top-order derivatives appearing in the second, third, and fourth collection of terms are $|I|$-order derivatives of $g$ and $|I|-1$-order derivatives of $v$. This implies that the difference $\pa_x^I v^0 - (v^I)^0$ gains one degree of regularity in comparison to $\pa_x^I v^0$, $(v^I)^0$, and allows us to take one more derivative of $\pa_x^I v^0 - (v^I)^0$ if $|I|=N$. Using the Sobolev embedding theorem (in the same manner as in Proposition~\ref{sec:energy-met.prop:com-ests}), we thus obtain the norm bounds 
    \begin{subequations}
        \begin{align}
            \label{sec:energy-fluid.eq:four-velocity-proof-1}
            \sum_{1\leq |I| + |J| \leq N+1}\Big\|\pa_x^J\big[\pa_x^I v^0 - (v^I)^0\big]\Big\|_{L^2} 
            &\lesa \|\pa_x h^{00}\|_{H^N} + \|\pa_x h\|_{H^N}^2 + \|\pa_x v\|_{H^{N-1}}^2 \lesa \frac{1}{\tau^{2+\frac{3}{2n}-\delta}},\\
            \label{sec:energy-fluid.eq:four-velocity-proof-2}
            \sum_{1 \leq|I| + |J| \leq N}\Big\|\pa_\tau \pa_x^J\big[\pa_x^I v^0 - (v^I)^0\big]\Big\|_{L^2} 
            &\lesa \|\pa \pa_x h^{00}\|_{H^{N-1}} + \|\pa h\|_{H^N}^2+ \|\pa v\|_{H^{N}}^2 \lesa \frac{1}{\tau^{2+\frac{3}{2n}-\delta}}\,,
        \end{align}
    \end{subequations}
    where the worst-decaying error terms are those linear in $\pa_x h^{00}$. This gives \eqref{sec:energy-fluid.eq:four-velocity-1}. The remaining estimates follow from the bounds
    \begin{align*}
        \sum_{1\leq |I|\leq N}\Big\|Q(\pa_x^I \rho,v^I)-Q(\pa_x^I \rho,\pa_x^I v)\Big\|_{L^2} &\lesa \sum_{1\leq |I|\leq N}\Big\|\pa_\tau \big[\pa_x^I v^0 - (v^I)^0\big]\Big\|_{L^2},\\
        \sum_{1\leq |I|\leq N-1}\Big\|\pa_x\Big[Q(\pa_x^I \rho,v^I)-Q(\pa_x^I \rho,\pa_x^I v)\Big]\Big\|_{L^2} &\lesa \sum_{1\leq |I|\leq N-1}\Big\|\pa_\tau \big[\pa_x^I v^0 - (v^I)^0\big]\Big\|_{H^1},
    \end{align*}
    the bounds \eqref{sec:energy-fluid.eq:four-velocity-proof-1}--\eqref{sec:energy-fluid.eq:four-velocity-proof-2}, and the triangle inequality.
\end{proof}
We are finally ready to prove the bound \eqref{sec:energy-fluid.eq:higher-energy-proof-2}.
\begin{proposition}[Commuted Euler equations]\label{sec:energy-fluid.prop:com-ests}
    We have the commuted estimates
    \begin{subequations}\label{sec:energy-fluid.eq:com-ests-1}
        \begin{multline}
            \sum_{1\leq |I\leq N}\Big(\big\|Q(\pa_x^I\rho,v^I)\big\|_{L^2} + \tau^{\frac{3}{n}}\sum_{i=1}^3\big\|R^i(\pa_x^I\rho,v^I)\big\|_{L^2}\Big)\\
            \lesa \tau^{\frac{3}{n}}\big[(E_{h^{00};N-1}^{\mathbf{v}})^{1/2} + (E_{h^{0*};N-1}^{\mathbf{v}})^{1/2}\big] + \frac{1}{\tau^{2-2\delta}}\,,
        \end{multline}
        \begin{multline}
            \sum_{0 \leq |I|\leq N-1}\Big(\Big\|\pa_x\big[Q(\pa_x^I\rho,v^I)\big]\Big\|_{L^2} + \tau^{\frac{3}{n}}\sum_{i=1}^3\Big\|\pa_x\big[R^i(\pa_x^I\rho,v^I)\big]\Big\|_{L^2}\Big)\\
            \lesa \tau^{\frac{3}{n}}\big[(E_{h^{00};N-1}^{\mathbf{v}})^{1/2} + (E_{h^{0*};N-1}^{\mathbf{v}})^{1/2}\big] + \frac{1}{\tau^{2-2\delta}}\,,
        \end{multline}
    \end{subequations}
    which immediately imply the bound \eqref{sec:energy-fluid.eq:higher-energy-proof-2}.
\end{proposition}
\begin{proof}
    Combining the results of Lemmas~\ref{sec:energy-fluid.lem:source-ests},~\ref{sec:energy-fluid.lem:four-velocity}, we compute the bounds
    \begin{multline*}
        \sum_{1\leq |I|\leq N}\Big(\big\|Q(\pa_x^I\rho,v^I)\big\|_{L^2} + \tau^{\frac{3}{n}}\sum_{i=1}^3\big\|R^i(\pa_x^I\rho,v^I)\big\|_{L^2}\Big) \\\lesa \sum_{1\leq |I\leq N}\Big(\big\|Q(\pa_x^I\rho,\pa_x^Iv)\big\|_{L^2} + \tau^{\frac{3}{n}}\sum_{i=1}^3\big\|R^i(\pa_x^I\rho,\pa_x^Iv)\big\|_{L^2}\Big)+\frac{1}{\tau^{2+\frac{3}{2n}-\delta}}\\
        \lesa \sum_{1\leq|I|\leq N}\Big(\Big\|\big[Q,\pa_x^I\big](\rho,v)\Big\|_{L^2} + \tau^{\frac{3}{n}}\sum_{i=1}^3\Big\|\big[R^i,\pa_x^I\big](\rho,v)\Big\|_{L^2}\Big)\\
        +\tau^{\frac{3}{n}}\big[(E_{h^{00};N-1}^{\mathbf{v}})^{1/2} + (E_{h^{0*};N-1}^{\mathbf{v}})^{1/2}\big] + \frac{1}{\tau^{2+\frac{3}{2n}-\delta}}\,.
    \end{multline*}
    Similarly
    \begin{multline*}
        \sum_{1\leq |I|\leq N-1}\Big(\Big\|\pa_x\big[Q(\pa_x^I\rho,v^I)\big]\Big\|_{L^2} + \tau^{\frac{3}{n}}\sum_{i=1}^3\Big\|\pa_x\big[R^i(\pa_x^I\rho,v^I)\big]\Big\|_{L^2}\Big) \\
        \lesa\sum_{1\leq|I|\leq N-1}\Big(\Big\|\pa_x\Big(\big[Q,\pa_x^I\big](\rho,v)\Big)\Big\|_{L^2} + \tau^{\frac{3}{n}}\sum_{i=1}^3\Big\|\pa_x\Big(\big[R^i,\pa_x^I\big](\rho,v)\Big)\Big\|_{L^2}\Big)\\
        +\tau^{\frac{3}{n}}\big[(E_{h^{00};N-1}^{\mathbf{v}})^{1/2} + (E_{h^{0*};N-1}^{\mathbf{v}})^{1/2}\big] + \frac{1}{\tau^{2+\frac{3}{2n}-\delta}}\,.
    \end{multline*}
    Thus we must bound the commutators 
    \begin{equation*}
        [Q,\pa_x^I](\rho,v) = Q(\pa_x^I \rho, \pa_x^Iv) - \pa_x^I Q(\rho,v)\,,\qquad [R^i,\pa_x^I](\rho,v) = R^i(\pa_x^I \rho, \pa_x^Iv) - \pa_x^I R^i(\rho,v)\,.
    \end{equation*}
    We give a sketch of this, as the idea is very similar to what is contained in Proposition~\ref{sec:energy-met.prop:com-ests}. We write 
    \begin{align*}
        \big[Q,\pa_x^I\big](\rho,v) &= \sum_{1\leq J\leq I} c_{I,J}^{(1)}\Big[(\pa_x^J v^\mu)\big(\pa_\mu \pa_x^{I-J}\rho\big)  + c_{I,J}^{(1)}(\pa_x^J \chi)\big(\pa_\mu \pa_x^{I-J} v^\mu\big)\Big]\,,\\
        \big[R^i,\pa_x^I\big](\rho,v) &= \sum_{1\leq J\leq I} c_{I,J}^{(1)}\Big[(\pa_x^J v^\mu)\big(\pa_\mu \pa_x^{I-J}v^i\big) + \frac{1}{\tau^{\frac{6}{n}}}\pa_x^J \Big(\frac{p'(\rho)}{\chi}\Pi^{\mu i}\Big)\big(\pa_\mu \pa_x^{I-J} \rho\big) + \frac{2}{\tau}(\pa_x^{J}v^0)(\pa_x^{I-J}v^i)\Big]\,.
    \end{align*}
    Bounding all terms using the Sobolev embedding, we thus have
    \begin{equation*}
        \sum_{1\leq |I|\leq N}\Big\|\big[Q,\pa_x^I\big](\rho,v)\Big\|_{L^2} + \sum_{1\leq |I|\leq N-1}\Big\|\pa_x\Big(\big[Q,\pa_x^I\big](\rho,v)\Big)\Big\|_{L^2} \lesa \|\pa v\|_{H^N}\|\pa \rho\|_{H^N} \leq \frac{1}{\tau^{2-2\delta}}\,,
    \end{equation*}
    \begin{multline*}
        \sum_{1\leq |I|\leq N}\Big\|\big[R^i,\pa_x^I\big](\rho,v)\Big\|_{L^2} + \sum_{1\leq |I|\leq N-1}\Big\|\pa_x\Big(\big[R^i,\pa_x^I\big](\rho,v)\Big)\Big\|_{L^2} \\\lesa \|\pa v\|_{H^N}^2 + \frac{1}{\tau^{\frac{6}{n}}}\big(\|\pa_x \rho\|_{H^N} + \|\pa_xv\|_{H^N}\big)\|\pa \rho\|_{H^N} \lesa \frac{1}{\tau^{2+\frac{3}{n}-2\delta}}\,.
    \end{multline*}
    We point out that in both of these estimates, the bad (non-integrable) decay of $\rho$ is compensated for by either good decay in $v$, or additional $\tau$-weights. This yields \eqref{sec:energy-fluid.eq:com-ests-1}.
\end{proof}
\section{Improving the bootstrap assumption}\label{sec:boot-imp}
We are now ready to improve the bounds \eqref{sec:stability.eq:boot-3}, \eqref{sec:stability.eq:boot-4} on the higher-order derivatives of the evolution variables from the bootstrap assumption. Combined with the improvement on the bounds for the spatial averages from Proposition~\ref{sec:avg-ests.prop:avg-ests}, this will complete the proof of Theorem~\ref{sec:stability.thm:main}.

We define the total higher-order energy by
\begin{multline}
    \mc{E}_{N}(\tau) := \tau^{4+\frac{3}{n}-2\delta}E_{h^{00};N-1}^{\mathbf{v}}(\tau) + \tau^{4+\frac{3}{n}-2\delta}E_{h^{0*};N-1}^{\mathbf{v}}(\tau)\\+ \tau^{4-2\delta}E_{h^{**};N-1}^{\mathbf{v}}(\tau) + \tau^{2-\frac{3}{2}-2\delta}E_{(\rho,v);N}(\tau)\,.
\end{multline}
Observe that $\mc{E}_N$ and the combined higher-order norms introduced in \eqref{sec:stability.eq:boot} are equivalent, in the sense that
\begin{equation}\label{sec:boot-imp.eq:equivalence}
\mc{E}_N(\tau) \sim (S_{h;N}(\tau))^2 + (S_{h;N-1}^{\mathbf{v}}(\tau))^2 + (S_{(\rho,v);N}(\tau))^2\,.
\end{equation}

\begin{theorem}\label{sec:boot-imp.thm:higher-ests}
    The energy $\mc{E}_{N}(\tau)$ satisfies the bound
    \begin{equation}\label{sec:boot-imp.eq:higher-ests-1}
        \pa_\tau \mc{E}_{N}(\tau) \leq \frac{C}{\tau^{1+\frac{3}{2n}-\delta}}\,.
    \end{equation}
    Moreover, we have the improved bound on the higher-order norms
    \begin{equation}\label{sec:boot-imp.eq:higher-ests-2}
        S_{h;N}(\tau) +S_{h;N-1}^{\mathbf{v}}(\tau)+ S_{(\rho,v);N}(\tau) \leq C\ve^{3/2}\,.
    \end{equation}
\end{theorem}
\begin{proof}
    For convenience, we define the weighted energies 
    \begin{gather*}
        \mc{E}_{h^{00};N-1}^{\mathbf{v}} := \tau^{4+\frac{3}{n}-2\delta}E_{h^{00};N-1}^{\mathbf{v}}\,,\qquad \mc{E}_{h^{0*};N-1}^{\mathbf{v}} := \tau^{4+\frac{3}{n}-2\delta}E_{h^{0*};N-1}^{\mathbf{v}}\,,\\
        \mc{E}_{h^{**};N-1}^{\mathbf{v}} :=\tau^{4-2\delta}E_{h^{**};N-1}^{\mathbf{v}}\,, \qquad \mc{E}_{(\rho,v);N} := \tau^{2-\frac{3}{n}-2\delta}E_{(\rho,v);N}\,.
    \end{gather*}
    With these weights, we have the relations
    \begin{gather*}
        S_{h;N-1}^{\mathbf{v}}(\tau) + S_{h;N}(\tau) \sim (\mc{E}_{h^{00};N-1}^{\mathbf{v}})^{1/2} + (\mc{E}_{h^{0*};N-1}^{\mathbf{v}})^{1/2} + (\mc{E}_{h^{**};N-1}^{\mathbf{v}})^{1/2},\\S_{(\rho,v);N}(\tau) \sim (\mc{E}_{(\rho,v);N})^{1/2}\,,
    \end{gather*}
    By Proposition~\ref{sec:energy-met.prop:higher-energy}, the weighted metric energies obey the bounds
    \begin{subequations}
        \begin{align}
            \label{sec:boot-imp.eq:higher-ests-proof-1}
            \pa_\tau \mc{E}_{h^{00};N-1}^{\mathbf{v}} + \frac{2\kappa-8-\frac{3}{n}+2\delta}{\tau}\mc{E}_{h^{00};N-1}^{\mathbf{v}} 
            &\lesa \frac{1}{\tau}(\mc{E}_{(\rho,v);N})^{1/2}(\mc{E}_{h^{00};N-1}^{\mathbf{v}})^{1/2} + \frac{1}{\tau^{1+\frac{3}{2n}-\delta}}\,,\\
            \label{sec:boot-imp.eq:higher-ests-proof-2}
            \pa_\tau \mc{E}_{h^{0*};N-1}^{\mathbf{v}} + \frac{\kappa-4-\frac{3}{n}+2\delta}{\tau}\mc{E}_{h^{0*};N-1}^{\mathbf{v}} &\lesa \frac{1}{\tau}(\mc{E}_{h^{00};N-1}^{\mathbf{v}})^{1/2}(\mc{E}_{h^{0*};N-1}^{\mathbf{v}})^{1/2}\\
            &\quad + \frac{1}{\tau^{1+\frac{3}{n}}}(\mc{E}_{(\rho,v);N})^{1/2}(\mc{E}_{h^{0*};N-1}^{\mathbf{v}})^{1/2} + \frac{1}{\tau^{1+\frac{3}{2n}-\delta}}\,,\nonumber\\
            \label{sec:boot-imp.eq:higher-ests-proof-3}
            \pa_\tau \mc{E}_{h^{**};N-1}^{\mathbf{v}} + \frac{2\delta}{\tau}\mc{E}_{h^{**};N-1}^{\mathbf{v}} &\lesa \frac{1}{\tau^{1+\frac{3}{2n}-\delta}}(\mc{E}_{h^{0*};N-1}^{\mathbf{v}})^{1/2} (\mc{E}_{h^{**};N-1}^{\mathbf{v}})^{1/2}\\
            &\quad + \frac{1}{\tau^{1+\frac{3}{2n}-\delta}}(\mc{E}_{(\rho,v);N})^{1/2}(\mc{E}_{h^{**};N-1}^{\mathbf{v}})^{1/2} + \frac{1}{\tau^{1+\frac{3}{2n}-\delta}}\,.\nonumber
        \end{align}
    \end{subequations}
    
    Moreover, by Proposition~\ref{sec:energy-fluid.prop:higher-energy}, the weighted fluid energy satisfies
    \begin{equation}\label{sec:boot-imp.eq:higher-ests-proof-4}
        \pa_\tau \mc{E}_{(\rho,v);N} + \frac{2\delta}{\tau}\mc{E}_{(\rho,v);N} \lesa \frac{1}{\tau}\big[(\mc{E}_{h^{00};N-1}^{\mathbf{v}})^{1/2} + (\mc{E}_{h^{0*};N-1}^{\mathbf{v}})^{1/2}\big](\mc{E}_{(\rho,v);N})^{1/2} + \frac{1}{\tau^{1+\frac{3}{2n}-\delta}}\,.
    \end{equation}
    We can treat several of the terms in \eqref{sec:boot-imp.eq:higher-ests-proof-2} and \eqref{sec:boot-imp.eq:higher-ests-proof-3} as error terms:
    \begin{gather*}
        \frac{1}{\tau^{1+\frac{3}{n}}}(\mc{E}_{(\rho,v);N})^{1/2}(\mc{E}_{h^{0*};N-1}^{\mathbf{v}})^{1/2} \leq \frac{C}{\tau^{1+\frac{3}{2n}-\delta}}\,,\\
        \frac{C}{\tau^{1+\frac{3}{2n}-\delta}}\Big[(\mc{E}_{h^{0*};N-1}^{\mathbf{v}})^{1/2} + (\mc{E}_{(\rho,v);N})^{1/2}\Big](\mc{E}_{h^{**};N-1}^{\mathbf{v}})^{1/2} \leq \frac{C}{\tau^{1+\frac{3}{2n}-\delta}}\,.
    \end{gather*}
    Using these bounds, we write out the full system \eqref{sec:boot-imp.eq:higher-ests-proof-1}--\eqref{sec:boot-imp.eq:higher-ests-proof-4} with explicit constants $C_{h^{00}}$, $C_{h^{0*}}$, $C_{(\rho,v)}$ for the critical terms, and a running constant for error terms:
    \begin{subequations}\label{sec:boot-imp.eq:higher-ests-proof-5}
        \begin{align}
            \label{sec:boot-imp.eq:higher-ests-proof-6}
            &\pa_\tau \mc{E}_{h^{00};N-1}^{\mathbf{v}} + \frac{2\kappa-8-\frac{3}{n}+2\delta}{\tau}\mc{E}_{h^{00};N-1}^{\mathbf{v}} \leq \frac{C_{h^{00}}}{\tau}(\mc{E}_{(\rho,v);N})^{1/2}(\mc{E}_{h^{00};N-1}^{\mathbf{v}})^{1/2} + \frac{C}{\tau^{1 + \frac{3}{2n}-\delta}}\,,\\
            \label{sec:boot-imp.eq:higher-ests-proof-7}
            &\pa_\tau \mc{E}_{h^{0*};N-1}^{\mathbf{v}} + \frac{\kappa-4-\frac{3}{n}+2\delta}{\tau}\mc{E}_{h^{0*};N-1}^{\mathbf{v}} \leq \frac{C_{h^{0*}}}{\tau}(\mc{E}_{h^{00};N-1}^{\mathbf{v}})^{1/2}(\mc{E}_{h^{0*};N-1}^{\mathbf{v}})^{1/2} + \frac{C}{\tau^{1+\frac{3}{2n}-\delta}}\,,\\
            \label{sec:boot-imp.eq:higher-ests-proof-8}
            &\pa_\tau \mc{E}_{h^{**};N-1}^{\mathbf{v}} + \frac{2}{\sqrt{\kappa}}\frac{1}{\tau}\mc{E}_{h^{**};N-1}^{\mathbf{v}} \leq\frac{C}{\tau^{1+\frac{3}{2n}-\delta}}\,,\\
            \label{sec:boot-imp.eq:higher-ests-proof-9}
            &\pa_\tau \mc{E}_{(\rho,v);N} + \frac{2}{\sqrt{\kappa}}\frac{1}{\tau}\mc{E}_{(\rho,v);N} 
            \leq \frac{C_{(\rho,v)}}{\tau}\big[(\mc{E}_{h^{00};N-1}^{\mathbf{v}})^{1/2} + (\mc{E}_{h^{0*};N-1}^{\mathbf{v}})^{1/2}\big](\mc{E}_{(\rho,v);N})^{1/2} + \frac{C}{\tau^{1+\frac{3}{2n}-\delta}}\,.
        \end{align}
    \end{subequations}
    In \eqref{sec:boot-imp.eq:higher-ests-proof-8}, \eqref{sec:boot-imp.eq:higher-ests-proof-9}, we used the lower bound of $\delta > \frac{1}{\sqrt{\kappa}}$ stated earlier in \eqref{sec:stability.eq:delta-choice} to shrink the bulk terms on the left hand sides. We stress that the constants $C_{h^{00}}$, $C_{h^{0*}}$, $C_{(\rho,v)}$ do not depend on $\tau$, $\kappa$, or the size of any of the evolution variables. The key idea is that we choose $\kappa$ to be sufficiently large that the bulk terms in the inequalities \eqref{sec:boot-imp.eq:higher-ests-proof-6}, \eqref{sec:boot-imp.eq:higher-ests-proof-6}, \eqref{sec:boot-imp.eq:higher-ests-proof-9} can absorb the critical terms on the RHS of the various equations. This would then imply the bound \eqref{sec:boot-imp.eq:higher-ests-1}. We will choose $\kappa$ in a way that depends only on the polytropic index $n$ and the constants $C_{h^{00}}$, $C_{h^{0*}}$, $C_{(\rho,v)}$. Through repeated applications of Young's inequality $ab \leq \frac{1}{2}a^2 + \frac{1}{2}b^2$, we have
    \begin{align*}
        \frac{C_{h^{00}}}{\tau}(\mc{E}_{(\rho,v);N})^{1/2}(\mc{E}_{h^{00};N-1}^{\mathbf{v}})^{1/2} &\leq \frac{1}{\tau}\Big[\frac{C_{h^{00}}^2}{4\kappa} \mc{E}_{(\rho,v);N} + \kappa\mc{E}_{h^{00};N-1}^{\mathbf{v}}\Big],\\
        \frac{C_{h^{0*}}}{\tau}(\mc{E}_{h^{00};N-1}^{\mathbf{v}})^{1/2}(\mc{E}_{h^{0*};N-1}^{\mathbf{v}})^{1/2} &\leq \frac{1}{\tau}\Big[\frac{C_{h^{0*}}^2}{2\kappa}\mc{E}_{h^{00};N-1}^{\mathbf{v}} + \frac{\kappa}{2}\mc{E}_{h^{0*}}\Big],\\
        \frac{C_{(\rho,v)}}{\tau}\big[(\mc{E}_{h^{00};N-1}^{\mathbf{v}})^{1/2} + (\mc{E}_{h^{0*};N-1}^{\mathbf{v}})^{1/2}\big](\mc{E}_{(\rho,v);N})^{1/2} &\leq \frac{1}{\tau}\Big[\frac{\sqrt{\kappa} C_{(\rho,v)}^2}{2}\mc{E}_{h^{00};N-1}^{\mathbf{v}} + \frac{\sqrt{\kappa} C_{(\rho,v)}^2}{2}\mc{E}_{h^{0*};N-1}^{\mathbf{v}} + \frac{1}{\sqrt{\kappa}}\mc{E}_{(\rho,v)}\Big].
    \end{align*}
    The final terms on the right hand side of each of the above bounds can then be absorbed into the left hand side of the equations \eqref{sec:boot-imp.eq:higher-ests-proof-5}. This implies the following estimate for the higher order energy $\mc{E}_N$:
    \begin{multline}\label{sec:boot-imp.eq:higher-ests-proof-10}
        \pa_\tau \mc{E}_N + \frac{\kappa - 8 - \frac{3}{n} + 2\delta}{\tau}\mc{E}_{h^{00};N-1}^{\mathbf{v}} + \frac{\frac{\kappa}{2} - 4 - \frac{3}{n} + 2\delta}{\tau}\mc{E}_{h^{0*};N-1}^{\mathbf{v}}  + \frac{1}{\sqrt{\kappa}}\frac{1}{\tau}\mc{E}_{(\rho,v);N}\\
        \leq \frac{1}{\tau}\Big(\frac{C_{h^{0*}}^2}{2\kappa} + \frac{\sqrt{\kappa}C_{(\rho,v)}^2}{2}\Big)\mc{E}_{h^{00};N-1}^{\mathbf{v}} + \frac{\sqrt{\kappa}C_{(\rho,v)}^2}{2}\frac{1}{\tau}\mc{E}_{h^{0*};N-1}^{\mathbf{v}} + \frac{C_{h^{00}}^2}{4\kappa}\frac{1}{\tau}\mc{E}_{(\rho,v);N}  + \frac{C}{\tau^{1+\frac{3}{2n}-\delta}}.
    \end{multline}
    In order to prove \eqref{sec:boot-imp.eq:higher-ests-1}, we must fix the parameter $\kappa$ to simultaneously satisfy the three inequalities:
    \begin{align*}
        \kappa - 8 - \frac{3}{n} + 2\delta &\geq \frac{C_{h^{0*}}^2}{2\kappa} + \frac{\sqrt{\kappa}C_{(\rho,v)}^2}{2}\,,\\
        \frac{\kappa}{2} - 4 - \frac{3}{n} + 2\delta &\geq \frac{\sqrt{\kappa}C_{(\rho,v)}^2}{2}\,,\\
        \frac{1}{\sqrt{\kappa}} &\geq \frac{C_{h^{00}}^2}{4\kappa}\,.
    \end{align*}

    This can be accomplished by choosing $\kappa$ to be sufficiently large, as the left hand side of each equation will dominate the right hand side as $\kappa \ra \infty$. For example, a choice that works is
    \begin{equation}
        \kappa = \max\Big\{2\Big(8+\frac{3}{n}\Big),C_{h^{00}}^4,C_{h^{0*}}^4,4(1+C_{(\rho,v)}^2)^2\Big\}\,.
    \end{equation}
    By fixing $\kappa$ in this manner, all critical terms on the right hand of the inequality \eqref{sec:boot-imp.eq:higher-ests-proof-10} side can be absorbed into the left hand side, which immediately implies \eqref{sec:boot-imp.eq:higher-ests-1}. 
    
    Finally, we recall the assumption that $\tau_0 \geq \ve^{-4n}$, and use this to trade some decay in $\tau$ for additional weights in $\ve$ for the error term in \eqref{sec:boot-imp.eq:higher-ests-1}, which implies
    \begin{equation*}
        \pa_\tau \mc{E}_N(\tau) \leq \frac{C\ve^3}{\tau^{1+\frac{3}{4n}-\delta}}\,,
    \end{equation*}
    Since $\delta < \frac{3}{4n}$, we may integrate in time over the interval $[\tau_0,\tau]$ to arrive at the bound
    \begin{equation*}
        \mc{E}_N(\tau) \leq \mc{E}_N(\tau_0) + C\ve^3\,.
    \end{equation*}
    Using the norm equivalence \eqref{sec:boot-imp.eq:equivalence}, as well as the initial bound \eqref{sec:stability.eq:initial-bound}, we obtain the improved bound \eqref{sec:boot-imp.eq:higher-ests-2} on the higher-order norms, completing the proof of this theorem, and indeed of Theorem~\ref{sec:stability.thm:main}.
\end{proof}
\appendix
\section{Existence and asymptotics of the FLRW solution}\label{app:FLRW}
We study solutions to the Einstein--Euler system with polytropic equation of state. Restricting to the class of FLRW solutions with $\TT^3$-spatial geometry, the system reduces to
\begin{equation}
    \mc{M} = (0,\infty)\times \TT^3,\qquad \bar{g} = -\de t^2 + a^2(t)\sum_{i=1}^3 (\de x^i)^2,\qquad \bar{\rho} =b(t),\qquad \bar{v}^\mu = \delta_0^\mu,
\end{equation}
where the functions $(a(t),b(t))$ solve the ODE\footnote{In these equations we assume the polytropic constant satisfies $C = 1$. In the results that follow, however, one can take $C$ to be any positive value without ultimately changing the conclusions.}
\begin{subequations}
    \begin{align}
        \label{app:FLRW.eq:fried-1}
        3\Big(\frac{\dot{a}(t)}{a(t)}\Big)^2 - b(t) &= 0\,,\\
        \label{app:FLRW.eq:fried-2}
        \frac{\dot{b}(t)}{b(t)\big(1+b^{\frac{1}{n}}(t)\big)} + 3\frac{\dot{a}(t)}{a(t)} &=0\,.
    \end{align}
\end{subequations}
These equations are often called the Friedmann equations. The first can be derived from the Hamiltonian constraint, while the second follows from the Euler continuity equation. In fact one generally considers the Friedmann equations to be the equations \eqref{app:FLRW.eq:fried-1}, \eqref{app:FLRW.eq:fried-2}, along with the evolution equation
\begin{equation*}
    -2\frac{\ddot{a}}{a}-\Big(\frac{\dot{a}}{a}\Big)^2  =\bar{p} = b^{1+\frac{1}{n}}(t)\,,
\end{equation*}
obtained from the spatial trace of the Einstein equations $\bar{g}^{ij}(R_{ij}[\bar{g}] - \frac{1}{2}\bar{g}_{ij}R[\bar{g}]\big) = \bar{g}^{ij}T_{ij}$. However, this system of three equations is overdetermined, and it is equivalent to solve just the two equations \eqref{app:FLRW.eq:fried-1}, \eqref{app:FLRW.eq:fried-2}.

In the following lemma, we establish the future global existence of solutions to this system, and study their long-time asymptotics.

\begin{lemma}[FLRW solutions to the Einstein--Euler system]\label{app:FLRW.lem:FLRW}
    Let $a_0 > 0$. There exists a one-parameter family of FLRW solutions $(a(t),b(t))$ to the Friedmann equations \eqref{app:FLRW.eq:fried-1}, \eqref{app:FLRW.eq:fried-2}, parametrised by $a_0$ with initial data
    \begin{equation}\label{app:FLRW.eq:initial-data}
        a(1) = a_0\,,\qquad b(1)= b_0(a_0)\,.
    \end{equation}
    Then functions $a(t)$ and $b(t)$ have the long-time asymptotics
    \begin{equation}\label{app:FLRW.eq:FLRW-asymp}
        a(t) = a_\infty t^{2/3} + O(t^{\frac{2}{3}-\frac{1}{n}})\,,\qquad b(t) = \frac{4}{9t^2} + O(t^{-2-\frac{1}{n}})\,.
    \end{equation}
\end{lemma}
\begin{proof}
    The continuity equation \eqref{app:FLRW.eq:fried-2} can be solved in closed form, giving
    \begin{equation*}
    n\log \Big(\frac{b^{\frac{1}{n}}(t)}{b^\frac{1}{n}(t)+1}\Big) - n\log \Big(\frac{b_0^{\frac{1}{n}}}{b_0^\frac{1}{n}+1}\Big) = -3\Big(\log a(t) -\log a_0\Big)\,.
    \end{equation*}
    Solving for $b(t)$ then yields the formula
    \begin{equation*}
        b(t) = \frac{b_0}{f(t)} \Big(\frac{a_0}{a(t)}\Big)^3\,,
    \end{equation*}
    where $f(t)$ is the function
    \begin{equation*}
        f(t) = \Bigg[b_0^{\frac{1}{n}}+1-\Big(\frac{b_0 a_0^3}{a^3(t)}\Big)^{\frac{1}{n}}\Bigg]^n\,.
    \end{equation*}
    We point out that if $a(t) \ra \infty$ as $t \ra \infty$, then $f(t)$ approaches a constant, and therefore $b(t) \sim a^{-3}(t)$. This is the same scale relation as is satisfied by the dust solution with a pressureless fluid.

    Next we plug in this formula to the Hamiltonian constraint \eqref{app:FLRW.eq:fried-1} and rearrange for $a(t)$, to obtain
    \begin{equation}
        \dot{a}(t) = \sqrt{\frac{b_0 a_0^3}{3f(t)}} \cdot \frac{1}{\sqrt{a(t)}}\,.
    \end{equation}
    Taking the square root of all these quantities is valid, as we assumed $a_0,b_0 > 0$, moreover $f(a_0) = 1$, so the above equation shows that $a(t)$ is monotonically increasing, and hence positive for all $t \geq t_0$. By the monotonicity properties, $f(t)$ is bounded below by $1$ and above by $(b_0^{\frac{1}{n}}+1)^n$, so we integrate the inequalities
    \begin{equation*}
        \sqrt{\frac{b_0 a_0^3}{3(b_0^{\frac{1}{n}}+1)^n}} \leq \sqrt{a}\dot{a} \leq \sqrt{\frac{b_0a_0^3}{3}}\,,
    \end{equation*}
    to arrive at the rough bound
    \begin{equation*}
        t\sqrt{\frac{b_0 a_0^3}{3(b_0^{\frac{1}{n}}+1)^n}} \leq \frac{2}{3} a^{\frac{3}{2}}(t) \leq t\sqrt{\frac{b_0a_0^3}{3}}\,,
    \end{equation*}
    where we have implicitly fixed the time coordinate to be such that $a(0) = 0$. This gives that $a(t) \sim t^{\frac{2}{3}}$ for large $t$. 
    
    We can then insert these asymptotics back into the Hamiltonian constraint \eqref{app:FLRW.eq:fried-1} equation to derive the asmyptotic expansions \eqref{app:FLRW.eq:FLRW-asymp}. Note that this also fixes the initial value of $b_0$ in terms of $a_0$.
\end{proof}
\begin{remark}[Solutions from scattering data]
    Through a straightforward asymptotic analysis, one can also construct FLRW solutions to \eqref{intro:Einstein-Euler} with prescribed data at infinity. More specifically, one can show that for any $a_\infty > 0$, there exists an FLRW solution to the Einstein--Euler system with polytropic equation of state such that $\lim_{t\ra \infty}a(t) = a_\infty$. We leave the details to the reader.
\end{remark}
We can extend the previous result to construct a family of spatially homogeneous solutions to \eqref{intro:Einstein-Euler}. Let $(\mc{M},\bar{g},\bar{\rho},\bar{v})$ be a spatially homogeneous solution to the Einstein--Euler system with $\TT^3$ spatial topology containing a quiet fluid, where we remind the reader that in this paper we take $\TT^3 = [-\pi,\pi]^3$ with the ends identified. Then there exists a coordinate system $(t,x^1,x^2,x^3) \in I \times [-\pi,\pi]^3$ such that
\begin{equation}\label{app:FLRW.eq:hom-sol}
    \mc{M} = I \times \TT^3,\qquad \bar{g} = -\de t^2 + a_{ij}^2(t)\de x^i \de x^j,\qquad \bar{\rho}=  b(t),\qquad \bar{v}^\mu = \delta_0^\mu,
\end{equation}
where $a_{ij},b:I \ra \RR$ are smooth functions on an interval $I$ containing $1$, such that for each $t \in I$, $b(t) > 0$ and the functions $a_{ij}(t)$ constitute components of a symmetric, positive definite $3\times 3$ matrix.

\begin{lemma}[Spatially homogeneous solutions to the Einstein--Euler system]
        Let $(a_{ij}^0)$, $(a_{ij}^\infty)$ be symmetric, positive definite $3\times3$ matrices. Then:
        \begin{enumerate}
            \item There exists a spatially homogeneous solution to \eqref{intro:Einstein-Euler} of the form \eqref{app:FLRW.eq:hom-sol} such that
            \begin{equation*}
                a_{ij}(1) = a_{ij}^0,\qquad i,j=1,2,3.
            \end{equation*}
            \item There exists a spatially homogeneous solution to \eqref{intro:Einstein-Euler} of the form \eqref{app:FLRW.eq:hom-sol} such that
            \begin{equation*}
                \lim_{t \ra \infty}a_{ij}(t) = a_{ij}^\infty,\qquad i,j=1,2,3.
            \end{equation*}
        \end{enumerate}
    \end{lemma}
\begin{proof}
    We prove statement $(I)$; statement $(II)$ follows by an similar argument. We carry out a linear coordinate transformation of the spatial coordinates
    \begin{equation*}
        \Phi: [-\pi,\pi]^3 \ra \RR^3,\qquad \Phi(x^1,x^2,x^3) = (y^1,y^2,y^3)
    \end{equation*}
    which transforms the matrix $a_{ij}|_{t=1}$ like
    \begin{equation*}
        a_0\delta_{ij} = \tilde{a}\Big(\frac{\pa}{\pa y^i},\frac{\pa}{\pa y^j}\Big)\Big|_{t=1} = \frac{\pa x^a}{\pa y^i}\frac{\pa x^b}{\pa y^j}a\Big(\frac{\pa}{\pa x^a},\frac{\pa}{\pa x^b}\Big)\Big|_{t=1} = \frac{\pa x^a}{\pa y^i}\frac{\pa x^b}{\pa y^j}a_{ab}(1),
    \end{equation*}
    where $a_0$ is chosen so that the initial data $\tilde{a}(1) = a_0$, $b(1) = b_0$ satisfy the relation \eqref{app:FLRW.eq:initial-data}.
    
    The solution \eqref{app:FLRW.eq:hom-sol} can now be written as
    \begin{equation}\label{app:FLRW.eq:hom-sol-transformed}
         \bar{g} = -\de t^2 + \tilde{a}_{ij} \de y^i \de y^j,\qquad \bar{\rho} = b(t),\qquad \bar{v}^\mu = \delta_0^\mu,
    \end{equation}
    where the coordinates $(y^1,y^2,y^3) \in \Phi([-\pi,\pi]^3)$. Critically, as $\tilde{a}_{ij}|_{t=1} = a_0\delta_{ij}$, the solution  \eqref{app:FLRW.eq:hom-sol-transformed} induces a valid initial data set to the Friedmann equations \eqref{app:FLRW.eq:fried-1}--\eqref{app:FLRW.eq:fried-2}, which by Lemma~\ref{app:FLRW.lem:FLRW} launches an FLRW solution to \eqref{intro:Einstein-Euler}. In particular, we can write
    \begin{equation*}
        \tilde{a}_{ij}(t) = \tilde{a}(t) \delta_{ij}.
    \end{equation*}
    Finally, one can recover the original form of the metric by reversing the coordinate transformation $\Phi$.
\end{proof}
\section{Proof of div-curl identity \texorpdfstring{\eqref{sec:energy-fluid.eq:corrected-energy-proof-4}}{ }}\label{app:divcurl}
In this appendix we give a proof of the div-curl identity \eqref{sec:energy-fluid.eq:corrected-energy-proof-4} which is loosely based on \protect{\cite[Lemma 4.4.1]{ChrKlai:minkowskistab:93}}. We write
\begin{equation}\label{app:divcurl-proof-1}
    2\int_{\TT^3} g^{ab}g^{cd}(\mathrm{curl}_{ac}\,\xi)(\mathrm{curl}_{bd}\,\xi) = \underbrace{\int_{\TT^3} g^{ab}g^{cd}g_{ea}g_{fb}\pa_c \xi^e \pa_d \xi^f}_{I}-\underbrace{\int_{\TT^3} g^{ab}g^{cd}g_{ea}g_{fd}\pa_c \xi^e \pa_b \xi^f}_{II}.
\end{equation}
For the integral labelled $I$, we observe that $g^{ab}g_{ea}g_{fb} = g_{ef} -g^{0b}g_{0e}g_{fb}$, and after relabelling indices this gives
\begin{equation}\label{app:divcurl-proof-2}
    I = \int_{\TT^3} g^{ab}g_{cd}\pa_a \xi^c \pa_b \xi^d - \int_{\TT^3} g^{0b}g_{0e}g_{fb}g^{cd}\xi^e \pa_d \xi^f \,.
\end{equation}
We perform integration by parts twice on the integral labelled $II$, which yields the identity
\begin{equation}\label{app:divcurl-proof-3}
    II = \underbrace{\int_{\TT^3}g^{ab}g^{cd}g_{ea}g_{fd}\pa_b \xi^e \pa_c \xi^f}_{III} + \int_{\TT^3}\pa_b \Big[g^{ab}g^{cd}g_{ea}g_{fd}\Big]\xi^e \pa_c \xi^f - \int_{\TT^3}\pa_c \Big[g^{ab}g^{cd}g_{ea}g_{fd}\Big]\xi^e \pa_b \xi^f\,,
\end{equation}
and for the integral labelled $III$, we have
\begin{equation*}
    g^{ab}g^{cd}g_{ea}g_{fd} = \delta_e^b \delta_f^c - g^{0b}g^{cd}g_{0e}g_{fd} - g^{ab}g^{0c}g_{ea}g_{0f} - g^{0b}g^{0c}g_{0e}g_{0f}\,,
\end{equation*}
which implies
\begin{equation}\label{app:divcurl-proof-4}
    III = \int_{\TT^3} (\mathrm{div}\,\xi)^2 -2\int_{\TT^3}g^{0c}g_{0f}g^{ab}g_{ea}\pa_b \xi^e \pa_c\xi^f - \int_{\TT^3}g^{0b}g^{0c}g_{0e}g_{0f}\pa_b\xi^e\pa_c\xi^f\,.
\end{equation}
Putting together \eqref{app:divcurl-proof-1}--\eqref{app:divcurl-proof-4}, we have the identity
\begin{equation*}
     \int_{\TT^3} g^{ab}g^{cd}(\mathrm{curl}_{ab}\,\xi)(\mathrm{curl}_{cd}\,\xi) = \int_{\TT^3} g^{ab}g_{cd}\pa_a \xi^c \pa_b \xi^d  -  \int_{\TT^3} (\mathrm{div}\,\xi)^2 - \mc{J}\,,
\end{equation*}
where
\begin{multline}\label{app:divcurl-proof-5}
    \mc{J} = \int_{\TT^3}\pa_c \Big[g^{ab}g^{cd}g_{ea}g_{fd}\Big]\xi^e \pa_b \xi^f - \int_{\TT^3}\pa_b\Big[g^{ab}g^{cd}g_{ea}g_{fd}\Big]\xi^e\pa_c\xi^f + \int_{\TT^3} g^{0b}g_{0e}g_{fb}g^{cd}\xi^e \pa_d \xi^f\\
    +2\int_{\TT^3}g^{0c}g_{0f}g^{ab}g_{ea}\pa_b \xi^e \pa_c\xi^f + \int_{\TT^3}g^{0b}g^{0c}g_{0e}g_{0f}\pa_b\xi^e\pa_c\xi^f\,,
\end{multline}
which after rearranging is \eqref{sec:energy-fluid.eq:corrected-energy-proof-4}. Additionally, if the components of $\xi$ have vanishing spatial average, and $g$ is a bounded metric then $\mc{J}$ satisfies the bound 
\begin{equation}
    |\mc{J}| \leq C\Big(\|\pa_xg\|_{L^\infty} + \sum_{i=1}^3\|g^{0i}\|_{L^\infty}^2\Big)\|\pa_x \xi\|_{L^2}^2\,.
\end{equation}
\bibliographystyle{amsalpha}
\bibliography{main.bib}

@article{AndMon:Milnestab:11,
    title = {{E}instein spaces as attractors for the {E}instein flow},
    author = {Lars Andersson and Vincent Moncrief},
    journal = {J. Differential. Geom.},
    volume = {89},
    number = {1},
    year = {2011},
    pages = {1-47}
}

@article{Ber:decelerated:25,
    title = {Future stability of solutions of the {E}instein-nonlinear scalar field
    system undergoing decelerated expansion},
    author = {Louie Bernhardt},
    journal = {Preprint, ar{X}iv:2508.15303},
    year = {2025}
}

@book{Ber:thesis:26,
    title = {Stability and scattering problems for expanding cosmologies in general relativity},
    author = {Louie Bernhardt},
    publisher = {PhD thesis, University of Melbourne},
    year = {2026}
}

@Article{BrauerRendallReula:1994,
author = {Uwe Brauer and Alan Rendall and Oscar Reula},
title = {The cosmic no-hair theorem and the non-linear stability of homogeneous {N}ewtonian cosmological models},
journal = {Class. Quant. Grav},
volume = {11},
number = {9},
year = {1994},
pages= {2283-2296}
}

@article{Cho:Einsteinlwp:52,
    title = {Th{\'e}or{\`e}me d'existence pour certains syst{\`e}mes d'{\'e}quations aux d{\'e}riv{\'e}es partielles non lin{\'e}aires},
    author = {Yvonne Foures (Choquet)-Bruhat},
    journal = {Acta Math.},
    volume = {88},
    year = {1952},
    pages = {141-225}
}

@article{Choquet-Bruhat:2001qwz,
    author = "Choquet-Bruhat, Yvonne and Moncrief, Vincent",
    title = "{Future global in time Einsteinian space-times with U(1) isometry group}",
    eprint = "gr-qc/0112049",
    archivePrefix = "arXiv",
    doi = "10.1007/s00023-001-8602-5",
    journal = "Annales Henri Poincare",
    volume = "2",
    pages = "1007--1064",
    year = "2001"
}

@book{ChrKlai:minkowskistab:93,
    title = {The Global Nonlinear Stability of the {M}inkowski Space},
    author = {Demetrios Christodoulou and Sergiu Klainerman},
    series = {Princeton Mathematical Series},
    volume = {41},
    publisher = {Princeton University Press, Princeton, NJ},
    year = {1993}
}

@article{dafermos2021nonlinearstabilityschwarzschildfamily,
      title={The non-linear stability of the {S}chwarzschild family of black holes}, 
      author={Mihalis Dafermos and Gustav Holzegel and Igor Rodnianski and Martin Taylor},
      journal = {Preprint, ar{X}iv:2104.08222},
      year={2021}
}

@article{EindeS:EdS:32,
    title = {On the relation between the expansion and the mean
    density of the universe},
    author = {Albert Einstein and Willem de Sitter},
    journal = {Proc. Nat. Acad. Sci.},
    volume = {18},
    year = {1932}
}

@article{FajMar:polytropic:26,
    title = {Non-linear stability of the matter dominated universe},
    author = {David Fajman and Elliot Marshall},
    journal = {Preprint, ar{X}iv:2606.16879},
    year = {2026}
}

@article{Fajetal:linEulerstab:24,
    title = {The Stability of Relativistic Fluids in Linearly Expanding Cosmologies},
    author = {David Fajman and Max Ofner and Todd Oliynyk and Zoe Wyatt},
    journal = {International Mathematical Research Notices},
    volume = {5},
    pages = {4238-4383},
    year = {2024}
}

@article{FajOfnWya:EEstab1:24,
    title = {Slowly Expanding Stable Dust Spacetimes},
    author = {David Fajman and Maximilian Ofner and Zoe Wyatt},
    journal = {Arch. Ration. Mech. Anal.},
    volume = {248},
    number = {83},
    year = {2024},
}

@article{Fajetal:decelEuler1:25,
    author = "Fajman, David and Maliborski, Maciej and Ofner, Maximilian and Oliynyk, Todd and Wyatt, Zoe",
    title = "{Phase transition between shock formation and stability in cosmological fluids}",
    eprint = "2405.03431",
    archivePrefix = "arXiv",
    primaryClass = "gr-qc",
    doi = "10.1088/1361-6382/ade929",
    journal = "Class. Quant. Grav.",
    volume = "42",
    number = "14",
    pages = "14LT01",
    year = "2025"
}

@article{Fajetal:decelEuler2:25,
    title = {Stability of fluids in spacetimes with decelerated expansion},
    author = {David Fajman and Maximilian Ofner and Todd A. Oliynyk and Zoe Wyatt},
    journal = {Preprint, ar{X}iv:2501.12798},
    year = {2025}
}

@article{Fajman2021,
  author  = {David Fajman and Todd Oliynyk and Zoe Wyatt},
  title   = {Stabilizing Relativistic Fluids on Spacetimes with Non-Accelerated Expansion},
  journal = {Communications in Mathematical Physics},
  year    = {2021},
  volume  = {383},
  number  = {1},
  pages   = {401--426},
  doi     = {10.1007/s00220-020-03924-9},
  url     = {https://doi.org/10.1007/s00220-020-03924-9},
  issn    = {1432-0916},
  month   = apr
}

@article{Fajetal:decelEuler3:25,
    author = {David Fajman and Maciej Maliborski and Maximilian Ofner and Todd Oliynyk and Zoe Wyatt},
    title = {A critical threshold for the cosmological {E}uler-{P}oisson system},
    journal = {Preprint, ar{X}iv:2512.19454},
    year = {2025}
}

@article{FouMarOli:tilted:24,
    title = {Future stability of perfect fluids with extreme tilt and linear equation of state {$p=c_s^2 \rho$} for the {E}instein-{E}uler system with positive cosmological constant: The range {$1/3 < c_s^2 < 3/7$}},
    author = {Grigorios Fournodavlos and Elliot Marshall and Todd Oliynyk},
    journal = {Sel. Math. New Ser.},
    volume = {32},
    number = {15},
    year = {2024}
}

@article{FouMarOli:tilted2:25,
    title = {Future Stability of Tilted Two-Fluid Bianchi {I} Spacetimes},
    author = {Grigorios Fournodavlos and Elliot Marshall and Todd Oliynyk},
    journal = {Preprint, ar{X}iv:2508.15155},
    year = {2025}
}

@article{Fri:deSitterstab:86,
    title = {On the existence of n-geodesically complete or future complete solutions of {E}instein’s field equations with smooth asymptotic structure},
    author = {Helmut Friedrich},
    journal = {Comm. Math. Phys.},
    volume. = {107},
    number = {4},
    year = {1986},
    pages = {587-609}
}

@article{Friedrich:2017,
title = "Sharp Asymptotics for {E}instein-$\lambda$-Dust Flows",
journal = "Comm. Math. Phys.",
volume = "350",
pages = "803 - 844",
year = "2017",
author = "H. Friedrich"
}

@article{Gong_2024,
doi = {10.1088/1361-6382/ad9132},
url = {https://doi.org/10.1088/1361-6382/ad9132},
year = {2024},
month = {nov},
publisher = {IOP Publishing},
volume = {41},
number = {24},
pages = {245017},
author = {Gong, Xinyu and Wei, Changhua},
title = {Stabilizing effect of the spacetime expansion on the Euler–Poisson equations in Newtonian cosmology},
journal = {Classical and Quantum Gravity},
}

@article{HadSpe:deSitterduststab:15,
    title = {The global future stability of the {FLRW} solutions to the dust-{E}instein system with a positive cosmological constant},
    author = {Mahir Had{\v z}i{\'c} and Jared Speck},
    journal = {J. Hyperbolic Differ. Equ.},
    volume = {12},
    number = {01},
    year = {2015},
    pages = {87-188}
}

@book{Heb:sob:00,
    title = {Nonlinear Analysis on Manifolds: {S}obolev Spaces and Inequalities},
    author = {Emmanuel Hebey},
    volume = {5},
    series = {Courant Lecture Notes in Mathematics},
    publisher = {Courant Institute of Mathematical Sciences, New York, NY; American Mathematical Society, Providence, RI},
    year = {1999}
}

@article{hintz2026nonlinearstabilitysubextremalkerr,
    title = {Nonlinear stability of subextremal {K}err black holes},
    author = {Peter Hintz},
    journal = {Preprint, ar{X}iv:2606.28253},
    year = {2026}
}

@article{HinVas:KdSexpandingstab:24,
    title = {Stability of the expanding region of {K}err-de {S}itter spacetimes and smoothness at the conformal boundary},
    author = {Peter Hintz and Andr{\'a}s Vasy},
    journal = {Preprint, ar{X}iv:2409.15460},
    year = {2024}
}

@article{HuneauStingoWyatt2025,
  author  = {C{\'e}cile Huneau and Annalaura Stingo and Zoe Wyatt},
  title   = {The global stability of the {K}aluza--{K}lein spacetime},
  journal = {Journal of the European Mathematical Society},
  year    = {2025},
  doi     = {10.4171/JEMS/1663},
  note    = {Published online first}
}

@article{Ju2025,
  author  = {Ju, Xianshu and Ke, Xiangkai and Wei, Changhua},
  title   = {The global existence and blowup of the classical solution to the relativistic dust in a {FLRW} geometry},
  journal = {Letters in Mathematical Physics},
  year    = {2025},
  volume  = {115},
  number  = {6},
  pages   = {122},
  doi     = {10.1007/s11005-025-02006-y},
  url     = {https://doi.org/10.1007/s11005-025-02006-y},
  issn    = {1573-0530},
  month   = oct,
  day     = {30}
}

@article{KlaiSzef:Kerr:23,
    title = {Kerr stability for small angular momentum},
    author = {Sergiu Klainerman and J\'er\'emie Szeftel},
    journal = {Pure and Appl. Math. Quart.},
    volume = {19},
    year = {2023},
    number = {3},
    pages = {791-1678}
}

@article {LeFlochWei-Chaplygin,
    AUTHOR = {Philippe G. LeFloch and Changhua Wei},
     TITLE = {Nonlinear stability of self-gravitating irrotational
              {C}haplygin fluids in a {FLRW} geometry},
   JOURNAL = {Ann. Inst. H. Poincar\'e{} C Anal. Non Lin\'eaire},
  FJOURNAL = {Annales de l'Institut Henri Poincar\'e{} C. Analyse Non
              Lin\'eaire},
    VOLUME = {38},
      YEAR = {2021},
    NUMBER = {3},
     PAGES = {787--814},
      ISSN = {0294-1449,1873-1430},
   MRCLASS = {76Y05 (76E30)},
  MRNUMBER = {4227052},
       DOI = {10.1016/j.anihpc.2020.09.005},
       URL = {https://doi.org/10.1016/j.anihpc.2020.09.005},
}

@article{LiuWei2021,
  author  = {Chao Liu and Changhua Wei},
  title   = {Future Stability of the {FLRW} Spacetime for a Large Class of Perfect Fluids},
  journal = {Annales Henri Poincaré},
  year    = {2021},
  volume  = {22},
  number  = {3},
  pages   = {715--770},
  doi     = {10.1007/s00023-020-00987-1},
  url     = {https://doi.org/10.1007/s00023-020-00987-1},
  issn    = {1424-0661},
  month   = mar
}

@article{LubKro:deSitterradstab11,
    title = {A conformal approach for the analysis of the non-linear stability of pure radiation cosmologies},
    author = {Christian L{\"u}bbe and Juan A. Valiente-Kroon},
    journal = {Ann. Phys.},
    volume = {328},
    year = {2013},
    pages = {1-25}
}

@article{Mar:instabilityFLRW:25,
    title = {Instability of Slowly Expanding {FLRW} Spacetimes}, 
    author = {Elliot Marshall},
    journal = {Class. Quantum Grav.},
    volume = {42},
    number = {10},
    year = {2025}
}

@article{Oli:EulerdeSitterstab:16,
    title = {Future stability of the {FLRW} fluid solutions in the presence of a positive cosmological constant},
    author = {Todd Oliynyk},
    journal = {Comm. Math. Phys.},
    volume = {346},
    year = {2016},
    number = {1},
    pages = {293-312}
}

@article{Rin:deSitterENSFstab08,
    title = {Future stability of the {E}instein-non-linear scalar field system},
    journal = {Invent. Math.},
    author = {Hans Ringstr{\"o}m},
    volume = {173},
    number = {1},
    pages = {123-208},
    year = {2008}
}

@article{Rin:powerlaw:09,
    title = {Power law inflation},
    journal = {Comm. Math. Phys.},
    author = {Hans Ringstr{\"o}m},
    volume = {290},
    pages = {155-218},
    year = {2009}
}

@article{RodSpe:irrEulerdeSitterstab:13,
    title = {The Stability of the Irrotational {E}uler-{E}instein System with a Positive Cosmological Constant},
    author = {Igor Rodnianski and Jared Speck},
    journal = {J. Eur. Math. Soc. (JEMS)},
    volume = {15},
    number = {6},
    year = {2013},
    pages = {2369-2462}
}

@ARTICLE{SaWo67,
       author = {{Sachs}, R.~K. and {Wolfe}, A.~M.},
        title = "{Perturbations of a Cosmological Model and Angular Variations of the Microwave Background}",
      journal = {Astrophys. J.},
         year = 1967,
        month = jan,
       volume = {147},
        pages = {73},
          doi = {10.1086/148982},
       adsurl = {https://ui.adsabs.harvard.edu/abs/1967ApJ...147...73S}
}

@article{Spe:EulerdeSitterstab:12,
title = {The nonlinear future stability of the {FLRW} family of solutions to the {E}uler-{E}instein system with a positive cosmological constant},
author = {Jared Speck},
journal = {Selecta Math.},
volume = {18}, 
number = {3},
year = {2012},
pages = {633-715}
}

@article{Spe:relEulerstab:13,
title = {The stabilizing effect of spacetime expansion on relativistic fluids with sharp results for the radiation equation of state},
author = {Jared Speck},
journal = {Arch. Ration. Mech. Anal.},
volume = {210},
number = {2},
year = {2013},
pages = {535-579}
}

@article{Strain:2026fka,
    author = {Robert M. Strain and Martin Taylor and Renato {Velozo Ruiz}},
    title = {{Future global stability of Maxwell-J{\"u}ttner equilibria and vacuum for the massless Boltzmann equation on FLRW spacetimes}},
    journal = {Preprint, ar{X}iv:2606.00175},
    year = {2026}
}

@article{Tay:decelFLRWstab:24,
title = {Future stability of expanding spatially homogeneous {FLRW} solutions of the spherically symmetric {E}instein--massless {V}lasov system with spatial topology {$\mathbb{R}^3$}},
author = {Martin Taylor},
journal = {J. Math. Phys.},
volume = {65},
number = {2},
year = {2024}
}

@article{WeiJDE2018,
title = {Stabilizing effect of the power law inflation on isentropic relativistic fluids},
journal = {Journal of Differential Equations},
volume = {265},
number = {8},
pages = {3441-3463},
year = {2018},
issn = {0022-0396},
doi = {https://doi.org/10.1016/j.jde.2018.05.007},
url = {https://www.sciencedirect.com/science/article/pii/S002203961830278X},
author = {Changhua Wei},
}

@book{Weinberg:2008zzc,
    author = "Weinberg, Steven",
    title = "{Cosmology}",
    isbn = "978-0-19-852682-7",
    year = "2008",
    publisher = "Oxford University Press, Oxford."
}
\end{document}